\documentclass[11pt]{article}

\title{Credit-Based vs. Discount-Based Congestion Pricing: A Comparison Study}

\author{
\begin{tabular}{c}
Chih-Yuan Chiu\\
\small \texttt{cyc@gatech.edu} \\
\small (Corresponding Author) \\
\small Georgia Institute of Technology\\
\small School of Electrical and Computer Engineering\\
\small 85 5th St NW, Atlanta, GA 30332, Room 439 \\ \\
\end{tabular}
\and
\begin{tabular}{c}
Devansh Jalota\\
\small \texttt{devansh.jalota@isye.gatech.edu} \\
\small Georgia Institute of Technology\\
\small School of Industrial and Systems Engineering\\
\small 55 5th St NW, Atlanta, GA 30332, Room 1019 \\ \\
\end{tabular}
\and
\begin{tabular}{c}
Marco Pavone\\
\small \texttt{pavone@stanford.edu} \\
\small Stanford University\\
\small Department of Aeronautics and Astronautics\\
\small 496 Lomita Mall, Stanford, CA 94305, Room 202 \\
\end{tabular}
}

\usepackage{amsmath, amssymb}
\usepackage[utf8]{inputenc}
\usepackage{amsthm}

\usepackage{natbib}

\allowdisplaybreaks

\usepackage[dvipsnames]{xcolor}
\usepackage{graphicx}
\usepackage{bbm}
\usepackage[inline]{enumitem}
\usepackage{url}
\usepackage[colorlinks=true,
raiselinks=true,
linkcolor=Sepia,
citecolor=MidnightBlue,
urlcolor=Sepia]{hyperref}

\usepackage{algorithm}
\usepackage[noend]{algpseudocode}

\usepackage{geometry}
\usepackage{siunitx}

\usepackage{mathtools}
\mathtoolsset{mathic=true}
\usepackage{todonotes}
\usepackage[varqu, varl, scale=0.9]{inconsolata}
\usepackage{booktabs}
\usepackage[normalem]{ulem}
\usepackage{dirtytalk}
\usepackage{multirow}

\numberwithin{equation}{section}

\newtheorem{theorem}{Theorem}[section]
\newtheorem{lemma}[theorem]{Lemma}
\newtheorem{proposition}[theorem]{Proposition}
\newtheorem{corollary}[theorem]{Corollary}

\newtheorem{remark}{Remark}
\newtheorem{assumption}{Assumption}
\newtheorem{example}[theorem]{Example}

\newtheorem{definition}[theorem]{Definition}

\DeclarePairedDelimiterX\loro[1]\rbrack\lbrack{#1}
\DeclarePairedDelimiterX\lorc[1]\rbrack\rbrack{#1}
\DeclarePairedDelimiterX\lcro[1]\lbrack\lbrack{#1}
\DeclarePairedDelimiterX\lcrc[1]\lbrack\rbrack{#1}
\DeclarePairedDelimiterX\set[1]\lbrace\rbrace{#1}

\newcommand{\ra}{\ensuremath{\rightarrow}}

\DeclarePairedDelimiterX\norm[1]\lVert\rVert{#1}

\newcommand{\N}{\ensuremath{\mathbb{N}}}

\newcommand{\R}{\ensuremath{\mathbb{R}}}
\newcommand{\Set}{\ensuremath{\mathcal{S}}}

\newcommand{\eq}{\text{eq}}
\newcommand{\flows}{\mathcal{Y}}
\newcommand{\G}{\mathcal{G}}
\newcommand{\network}{\mathcal{N}}
\newcommand{\Y}{\mathcal{Y}}
\newcommand{\routes}{\textbf{R}}
\newcommand{\nodes}{I}
\newcommand{\edges}{E}

\newcommand{\In}{\text{in}}
\newcommand{\Out}{\text{out}}
\newcommand{\ELP}{\text{ELP}}
\newcommand{\ori}{\text{ori}}
\newcommand{\dest}{\text{dest}}

\newcommand{\boldtau}{\boldsymbol{\tau}}
\newcommand{\boldalpha}{\boldsymbol{\alpha}}
\newcommand{\boldlambda}{\boldsymbol{\lambda}}
\newcommand{\boldy}{\textbf{y}}

\newcommand{\boldr}{\textbf{r}}
\newcommand{\boldzero}{\textbf{0}}

\newcommand{\adm}{\text{adm}}
\newcommand{\iter}{\text{iter}}

\newcommand{\ODPair}{\mathcal{P}}
\newcommand{\city}{\text{city}}
\newcommand{\income}{\text{income}}
\newcommand{\paren}[1]{{({#1})}}

\newcommand{\Unif}{\text{Unif}}
\newcommand{\Sphere}{\mathcal{S}}
\newcommand{\proj}{\text{proj}}

\newcommand{\boldx}{\textbf{x}}

\newcommand{\boldB}{\textbf{B}}

\newcommand{\revisionApp}[1]{\textcolor[rgb]{0,0,0}{#1}}

\newcommand{\revisionTS}[1]{\textcolor[rgb]{0,0,0}{#1}}

\newcommand{\revisionTRPartB}[1]{\textcolor[rgb]{0,0,0}{#1}}

\usepackage{todonotes}

\begin{document}

\maketitle


\begin{abstract}
Credit-based congestion pricing (CBCP) and discount-based congestion pricing (DBCP), which respectively allot travel credits and toll discounts to subsidize low-income users' access to tolled roads, have emerged as promising policies for alleviating the societal inequity concerns of congestion pricing. However, since real-world deployments of CBCP and DBCP remain limited, transportation agencies lack clear guidance on which policy better achieves travel cost reduction and revenue generation objectives. In this work, we compare the efficacy of deploying CBCP and DBCP in reducing user costs and increasing toll revenues. 
We first formulate a non-atomic congestion game in which low-income users receive a travel credit or toll discount for accessing tolled lanes. We establish that, in our formulation, Nash equilibrium flows always exist and can be computed or well approximated via convex programming. Our main result establishes a set of practically relevant conditions under which DBCP provably outperforms CBCP in inducing equilibrium outcomes that minimize a given societal cost, which encodes user cost reduction and toll revenue maximization. Finally, we validate our theoretical contributions via a case study of the 101 Express Lanes Project, a CBCP program implemented in the San Francisco Bay Area. 
Our results provide theoretical and computational foundations for evaluating equity-oriented congestion pricing policies and offer practical guidance for transportation agencies seeking to design congestion pricing programs that better balance travel cost reduction and revenue generation objectives.
\end{abstract}


\vspace{-6mm}
\section{Introduction}
\label{sec: Introduction}


Congestion levels in urban centers have risen sharply around the world in recent years, resulting in extended commute times that exacerbate environmental pollution and cause steep financial losses. As policy makers explore measures to reduce vehicular volume during peak traffic, congestion pricing has attracted increasing attention as a promising mechanism for traffic management. 
The core tenet of congestion pricing is to levy tolls on regularly overcrowded roads, to shift the routing decisions of selfish users and promote the efficient use of traffic infrastructure. 
Beyond cordon-based congestion pricing policies, such as the one recently implemented in New York City, congestion tolls have been deployed on freeway express lanes in the San Francisco Bay Area and other urban centers \citep{CBCP-SanMateo, DBCP-ExpressLaneSTART, Seattle-CongestionPricingWebsite} to offer users faster commutes through congestion zones.

Despite the potential of congestion fees to alleviate traffic congestion, current congestion pricing mechanisms have also attracted criticism for their disproportionate impact on low-income travelers, due to their regressive nature. Specifically, tolled express lanes on highways have been denounced as \say{elitist Lexus lanes} that levy \say{taxes on the working class,} for offering faster commute options to the wealthiest users at the expense of relegating lower-income users to slower lanes \citep{Patterson2008LexusOrCorolla, Klein2021AreTollLanesElististOrProgressive, Watterman2025CongestionPricingWillHurtWorkingClassCommuters, Paybarah2019CongestionPricingTaxOnWorkingClass}.

Given the fairness concerns surrounding current tolling policies, the design of \textit{equitable} congestion pricing policies has attracted increasing attention from policy makers who aim to reduce traffic congestion without worsening social inequities. For instance, traffic authorities in the San Francisco Bay Area \citep{CBCP-SanMateo} have implemented a pilot program for \textit{credit-based} congestion pricing (CBCP), which levies express lane tolls on segments of the US-101 highway while allocating travel credits to low-income users to access tolled express lanes. Meanwhile, \textit{discount-based} congestion pricing (DBCP) policies, which offer users toll discounts instead of travel credits to offset the financial burden of road tolls, have been deployed elsewhere in the San Francisco Bay Area \citep{DBCP-ExpressLaneSTART}
and other urban centers.

Although CBCP and DBCP both show promise in addressing the equity concerns of congestion pricing, limited research has studied the fundamental question of whether CBCP or DBCP outperforms the other in attaining a given societal objective, such as mitigating traffic congestion or alleviating inequities. 
To address this gap and guide the practical design of DBCP and CBCP policies, our work compares the effects of their deployment to manage traffic flow on 
highway networks with tolled express lanes,
in alignment with the real-world deployment of these policies.
Specifically, we present a set of practically relevant conditions on the network structure, user attributes, and societal cost, under which DBCP yields lower societal costs than CBCP, or vice versa.
The identification of these conditions is crucial for policy makers who aim to fully realize the potential of DBCP and CBCP for equitable and efficient traffic management.
\paragraph{Contributions:}
In this work, we study and compare the societal cost outcomes of deploying CBCP and DBCP policies over contiguous multi-lane highway segments with express and general-purpose (GP) lanes.
In alignment with the real-world deployment of CBCP and DBCP policies, we first present a mixed economy model, in which \textit{eligible} users receive either travel credits or toll discounts to access tolled express lanes, while \textit{ineligible} users must pay out-of-pocket to use express lanes. 
Our model generalizes existing CBCP formulations to account for traffic networks with multiple origin-destination pairs 
\citep{Jalota2022CreditBasedCongestionPricing}.
Similarly, we extend existing convex programming methods for tractably computing or approximating CBCP equilibrium flows to the multiple origin-destination setting. 
We also present a convex program to compute DBCP equilibria, based on the convex optimization methods commonly used to compute heterogeneous user equilibria.
As our main contribution, we analyze and contrast the impact of deploying CBCP and DBCP policies, quantified via a societal cost metric that accounts for the travel costs incurred by users of different income levels, as well as the amount of toll revenue generated. Concretely, we first present Assumptions \ref{Assum, h: zero sublevel set of h is bounded}-\ref{Assum: Eligible users cannot access express lane without subsidy under DBCP} on the network structure, user attributes, and societal cost, which are consistent with real-world deployments of CBCP and DBCP policies over highway corridors
in the San Francisco Bay Area. 
In particular, we compare CBCP and DBCP schemes for chain networks representing contiguous highway express-lane corridors with multiple origin-destination pairs (Assumption \ref{Assum: Chain Network}), reflecting the predominant structure of express lane congestion pricing deployments while enabling an edge-wise equilibrium decomposition essential to our theoretical analysis.
Then,
we demonstrate that:
\begin{enumerate}
    \item Under Assumptions \ref{Assum, h: zero sublevel set of h is bounded}-\ref{Assum: Eligible users cannot access express lane without subsidy under DBCP}, if the societal cost is defined to prioritize toll revenue maximization over the minimization of eligible users' travel costs, then DBCP policies induce a lower or equal societal cost at equilibrium compared to CBCP policies (Thm. \ref{Thm: Main}).
    \revisionTS{
    Our focus on revenue objectives reflects their key role in real-world congestion pricing, in which toll revenues are utilized to fund transit investment \citep{MTA2026CongestionPricingMakesForBetterTransit} or lump-sum payments to low-income users  \citep{Jalota2021WhenEfficiencyMeetsEquity}. We note that revenue objectives also underpin many public–private toll road partnerships \citep{ONeill2022PrivatisingAndFinancialisingRoads}.
    }
    

    \item 
    Under Assumptions \ref{Assum, h: zero sublevel set of h is bounded}-\ref{Assum: Eligible users cannot access express lane without subsidy under DBCP}, and an additional mild condition on the CBCP policies under consideration (Assumption \ref{Assum: Optimal CBCP Eq Flow admits eligible users}), if the societal cost \textit{strictly} prioritizes toll revenue maximization over the minimization of eligible users' travel costs, then DBCP policies induce a \textit{strictly} lower minimum societal cost at equilibrium compared to CBCP policies (Cor. \ref{Cor: Main Thm, Strictness}).

    \item If the societal cost is instead defined to prioritize the minimization of eligible users' travel costs over toll revenue maximization, it is possible for CBCP policies to outperform DBCP policies in minimizing the equilibrium societal cost (Prop. \ref{Prop: Assump lambda R geq lambda E removed}).
    

\end{enumerate}

To validate our theoretical contributions, we present a numerical study of the San Mateo 101 Express Lanes Project, in which we adapt existing zeroth-order gradient descent methods \citep{Maheshwari2024FollowerAgnosticLearninginStackelbergGames} to compute first-order stationary CBCP and DBCP policies for consecutive segments of the US-101 highway. We then compare the
users' travel costs and generated toll revenues realized by the stationary CBCP and DBCP policies
under equilibrium flows. Our numerical results indicate that, consistent with our main theoretical results (Thm. \ref{Thm: Main} and Cor. \ref{Cor: Main Thm, Strictness}), DBCP policies 
outperform CBCP policies under the conditions on the network structure, user attributes, and societal cost (Assumptions \ref{Assum, h: zero sublevel set of h is bounded}-\ref{Assum: Eligible users cannot access express lane without subsidy under DBCP}) stated above. 
Based on our theoretical and empirical analysis of DBCP and CBCP policies, we present policy recommendations and propose avenues of future research for improving efficiency and equity in
congestion pricing.

Our paper is structured as follows. Sec. \ref{sec: Related Work} surveys relevant literature. Sec. \ref{sec: CBCP and DBCP Policies} presents the traffic flow and mixed economy models analyzed in our work, and formulates CBCP and DBCP policies and their corresponding Nash equilibrium flow concepts, which we term CBCP and DBCP equilibria, respectively. Then, in Sec. \ref{sec: CBCP and DBCP Equilibrium Computation}, we present convex programs to tractably compute or approximate CBCP and DBCP equilibria. Next, in Sec. \ref{sec: CBCP and DBCP Comparison Study for the Chain Network Setting}, we compare CBCP and DBCP policies on a class of societal cost metrics. 
In Sec. \ref{sec: Numerical Experiments}, we present an empirical study of the San Mateo 101 Express Lanes Project that validates theoretical results presented in Sec. \ref{sec: CBCP and DBCP Comparison Study for the Chain Network Setting}. In addition, we derive policy implications from our empirical results and discuss directions for future research. Finally, in Sec. \ref{sec: Conclusion}, we summarize our contributions.


\section{Related Work}
\label{sec: Related Work}

Our work is broadly inspired by prior research efforts to promote both efficiency and equity in resource allocation across a wide range of engineering domains, including communication network design \citep{Ogryczak2014FairOptimizationandNetworksASurvey, Pioro2003OnEfficientMaxMinFairRoutingAlgorithms} 
and traffic routing \citep{Jahn2005SystemOptimalRoutingofTrafficFlowswithUserConstraintsinNetworkswithCongestion}. 
Specifically, our work builds upon and contributes to existing research on equitable congestion pricing, which aims to balance efficiency and equity when designing incentives to influence users’ routing behaviors in real-world traffic. Although the literature on equitable congestion pricing is extensive, below we review prior works relevant to the two classes of policies compared in our work: 
credit-based and discount-based congestion pricing.

Credit-based congestion pricing (CBCP) refunds a fraction of toll revenues to low-income users as travel budgets that enable access to tolled resources
\citep{Li2020TrafficAndWelfareImpactsOfCBCP, Ferguson2020EffectivenessOfSubsidiesAndTolls}, and is thus related to the 
literature on
revenue redistribution policies, which similarly transfer toll revenues to low-income users as lump-sum payments  \citep{Daganzo1995AParetoOptimumCongestionReductionScheme, Guo2010ParetoImprovingCongestionPricing, Jalota2021WhenEfficiencyMeetsEquity}.
Prior works \citep{Small1992UsingRevenuesFromCongestionPricing, 
Guo2010ParetoImprovingCongestionPricing, Jalota2021WhenEfficiencyMeetsEquity} prescribed revenue redistribution policies that can attain system efficiency and address equity issues 
across a range of traffic congestion models.
Similarly, our CBCP model considers allocating a fraction of the toll revenue to low-income users. However, unlike prior works which propose distributing monetary payouts 
to users, our CBCP formulation allocates credits whose sole function is to enable express lane access. 

Aside from CBCP, our work also studies discount-based congestion pricing (DBCP) policies, which instead 
provide toll \textit{discounts} to disadvantaged users. DBCP policies impose distinct toll charges to different user groups, and are thus closely related to the heterogeneous tolling policies commonly studied in the congestion pricing literature. While existing work on heterogeneous tolling differentiates users based on various attributes, e.g., whether their vehicles operate autonomously or use renewable energy sources, the DBCP policies formulated in our work specifically subsidize \textit{low-income} users to address equity issues in existing tolling policies.
\revisionTRPartB{
Our analysis also builds upon multi-class, multi-criteria traffic equilibrium models, which characterize heterogeneous users' routing decisions based on tradeoffs between travel time and monetary costs \citep{Dial1996BicriterionTrafficAssignment, Yang2004MulticlassMulticriteriaTrafficNetworkEquilibrium}. While DBCP fits naturally within the conventional heterogeneous tolling framework, where users face 
differentiated
tolls and routing decisions can be analyzed independently across periods \citep{Cole2003PricingNetworkEdgesforHeterogeneousSelfishUsers}, CBCP introduces additional mathematical challenges. Specifically, eligible users face intertemporal credit budget constraints that couple their routing decisions across periods, such that credit expenditure in one period affects the feasible routing choices in subsequent periods. Consequently, unlike conventional multi-class user equilibrium models, CBCP equilibria cannot generally be characterized through independent, period-by-period routing decisions. Moreover, these intertemporal constraints complicate the convex programming characterization of the resulting equilibrium flows, particularly when users' values of time vary across periods. Related intertemporal credit considerations have been studied in multi-period tradable credit schemes \citep{Miralinaghi2016MultiPeriodEquilibriumModelingPlanningFramework, Miralinaghi2018AMultiPeriodTradableCreditScheme}, which focus on credit trading and intertemporal credit allocation. In contrast, our work considers targeted, non-tradable travel credits for subsidy-eligible users and establishes equilibrium existence and convex programming characterizations (or approximation guarantees) in this setting. Building on these foundations, we further provide theoretical comparisons of the societal costs induced by credit-based and discount-based congestion pricing, thereby identifying conditions under which one subsidy mechanism outperforms the other.}


From a methodological perspective, our work contributes to the rich literature on characterizing heterogeneous user equilibria in routing games 
\citep{Fleischer2004TollsforHeterogeneousSelfishUsers, Maheshwari2024CongestionPricingforEfficiencyandEquity}. In particular, our work is closely aligned with \citet{Jalota2022CreditBasedCongestionPricing}, which analyzed equilibrium flow patterns induced by CBCP policies.
However, the equilibrium flow analysis in \citet{Jalota2022CreditBasedCongestionPricing}
was presented only for the \textit{single} highway segment setting.
In contrast, our work presents novel characterizations of CBCP and DBCP equilibria not found in 
\citet{Jalota2022CreditBasedCongestionPricing} (see Lemmas \ref{Lemma: DBCP Eq Decomposition, Chain Network}, \ref{Lemma: DBCP Equilibria and Effective VoTs}, and \ref{Lemma: Main, Explicit Construction of DBCP Policy}), to 
contrast the deployment of CBCP and DBCP policies across \textit{multiple consecutive} highway segments. 
Thus, compared to 
\citet{Jalota2022CreditBasedCongestionPricing}, our work offers a more realistic description of the equilibrium traffic patterns induced by congestion pricing mechanisms, most of which are currently deployed over consecutive highway segments.
Moreover, since our CBCP formulation enforces budget constraints (see Definition \ref{Def: CBCP Equilibria}), our work is also connected to the literature on equilibrium traffic models with side constraints (e.g., capacity constraints) \citep{LarssonPatriksson1999SideConstrainedEquilibriumModels, Correa2004SelfishRoutinginCapacitatedNetworks, Jahn2005SystemOptimalRoutingofTrafficFlowswithUserConstraintsinNetworkswithCongestion}. However, the side constraints introduced in the traffic models of \citet{LarssonPatriksson1999SideConstrainedEquilibriumModels, Correa2004SelfishRoutinginCapacitatedNetworks, Jahn2005SystemOptimalRoutingofTrafficFlowswithUserConstraintsinNetworkswithCongestion} merely impose restrictions on each user's routing decision along a \textit{single journey} from origin to destination. In contrast, the budget constraints in our CBCP policy formulation introduce dependencies between the routing decisions of each user \textit{across multiple journeys over periods} that cannot be disentangled, since,
in our traffic model, both the toll levied on each edge and each user's value of time may vary over periods.

Our work 
is also related to
the literature on 
artificial currency mechanisms, 
and zeroth-order descent methods for bilevel optimization, which we survey in App. \ref{sec: App, Additional Related Work}. 
We also describe the substantial ways in which we have built upon a preliminary conference version of this paper.


\section{CBCP and DBCP Policies}
\label{sec: CBCP and DBCP Policies}

In this section, we introduce the network, user population, and traffic flow models used in our work (Sec. \ref{subsec: Setup}). We then present discount-based congestion pricing (DBCP) and credit-based congestion pricing (CBCP) policies and corresponding equilibrium flow concepts (Sec. \ref{subsec: DBCP, Formulation}-\ref{subsec: CBCP, Formulation}). 
Our model formulation in this section, and the properties of DBCP and CBCP equilibria presented in Sec. \ref{sec: CBCP and DBCP Equilibrium Computation}, 
\revisionTRPartB{provide the technical foundation for}
our main theoretical results in Sec. \ref{sec: CBCP and DBCP Comparison Study for the Chain Network Setting}, which compare the societal cost at DBCP and CBCP equilibrium flows on highways with tolled express lanes.

\subsection{Setup}
\label{subsec: Setup}

We formulate and analyze CBCP and DBCP policies on an acyclic traffic network $\network = (\nodes, \edges)$, where $\nodes$ and $\edges$ denote the set of nodes and the set of edges in $\network$, respectively. 
We denote the set of all incoming edges and all outgoing edges at each node $i \in \nodes$ by $\edges_i^\In$ and $\edges_i^\Out$, respectively.
Let $\ODPair$ denote the set of origin-destination (o-d) pairs
in the network $\network$. Each o-d pair $p \in \ODPair$ is associated with an origin node $p_o$, a destination node $p_d$, and a set of routes $R(p)$. Concretely, each route associated with an o-d pair $p \in \ODPair$ is defined to be an ordered subset of edges that originates at $p_o$ and terminates at $p_d$. Each edge $e \in \edges$ in the network $\network$ comprises two lanes $k \in [2] := \{1, 2\}$, an \textit{express} lane that can be tolled ($k = 1$) and \textit{general-purpose} (GP) lanes which must remain toll-free ($k = 2$). 
Although we model each edge as two lanes for simplicity, we note in Remark \ref{Remark: Latency Function Model for Multiple GP Lanes} that our theoretical results 
generalize to the
setting in which each 
network edge comprises
one express lane and multiple GP lanes. 



Users commute through the traffic network $\network$ at each period $t \in [T]$ over a finite time horizon $[T] := \{1, \cdots, T\}$. We partition users into 
a finite set of \textit{eligible} user groups $G^E$ who receive travel subsidies (credits or discounts), and a finite set of \textit{ineligible} user groups $G^I$ who receive no subsidy. 
Each group $g \in G := G^E \cup G^I$
is identified with a specific o-d pair $p^g := (p_o^g, p_d^g) \in \ODPair$, a value of time $v_t^g$ at each period $t \in [T]$, and a travel demand $d^g$ describing the flow of users in group $g \in G$ traveling through $\network$ at each $t \in [T]$. 
We emphasize that the $T$ time intervals represent distinct static traffic equilibrium instances, coupled through eligible users' intertemporal credit budgets, rather than a dynamic traffic flow model requiring explicit vehicle flow propagation.

The travel demand of each group is assumed to be fixed over the $T$ periods, reflecting weekday rush hour traffic patterns observed, e.g., in the Caltrans' Performance Measurement System (PeMS) database \citep{pems-database}. For each edge $e \in \edges$, we denote by $G_e$ the set of all groups $g \in G$ whose associated o-d pair is connected by at least one route $r \in R(p^g)$ containing $e$, i.e., $G_e := \{ g \in G: \exists \ r \in R(p^g) \text{ s.t. } e \in r \}$. Moreover, we set $G_e^E := G_e \cap G^E$ and $G_e^I := G_e \cap G^I$. We also define $d_e := \sum_{g \in G_e} d^g$ to be the maximum possible demand experienced on edge $e$. Similarly, we define $d_e^I := \sum_{g \in G_e^I} d^g$ and $d_e^E := \sum_{g \in G_e^E} d^g$. We also define $n^E := \sum_{e \in \edges} |G_e^E|$ to be the total count of eligible groups across various edges (here, $|\cdot|$ denotes set cardinality),
and
set $n^I := \sum_{e \in \edges} |G_e^I|$ and $n := \sum_{e \in \edges} |G_e|$. 

Users' routing decisions, across all groups and periods,
generate a \textit{flow} $\boldy := (y_{e,k,t}^g: e \in \edges, k \in [2], t \in [T], g \in G_e) \in \R_{\geq 0}^{2nT}$,
which 
must be component-wise non-negative and satisfy flow continuity at each node $i \in \nodes$ and period $t \in [T]$. Concretely, $\boldy$ must lie within the admissible set $ \Y_\adm$ defined below, where $\textbf{1}(\cdot)$ returns 1 when the input argument is true, and 0 otherwise:
{
\setlength{\abovedisplayskip}{4pt}
\setlength{\belowdisplayskip}{4pt}
\begin{align} \label{Eqn: Admissible Flow Set}
    \Y_\adm &:= \Big\{ \boldy \in \R_{\geq 0}^{2nT}: d^g \ \textbf{1}\{i = p_o^g \} + 
    \sum_{\hat e \in \edges_i^\In}
    \sum_{k=1}^2 y_{\hat e,k,t}^g = d^g \ \textbf{1}\{i = p_d^g \} + 
    \sum_{e' \in \edges_i^\Out} 
    \sum_{k=1}^2 y_{e',k,t}^g, \\ \nonumber
    &\hspace{3.6cm} \forall \ i \in \nodes, g \in G, t \in [T] \Big\}.
\end{align}
}
\hspace{-1mm}As we elucidate in Sec. \ref{subsec: DBCP, Formulation}-\ref{subsec: CBCP, Formulation} below, congestion pricing policies additionally restrict the set of feasible flows. Next, we denote the \textit{lane flows} corresponding to the flow $\boldy$ by $x := (x_{e,k,t}: e \in \edges, k \in [2], t \in [T]) \in \R_{\geq 0}^{2|\edges|T}$, where
$x_{e,k,t} := \sum_{g \in G_e} y_{e,k,t}^g, \ \forall e \in \edges, k \in [2], t \in [T].$
The lane flows induced by users' routing decisions generate latency over the network. Concretely, we model the travel time on either lane of each edge $e$, at each period $t \in [T]$, by the strictly increasing, convex function $\ell_e: \R \ra \R$. Since in practice, the express and GP lanes associated with an edge $e$ represent different lanes on the same road or highway segment, we assume in Sec. \ref{sec: CBCP and DBCP Policies}-\ref{sec: CBCP and DBCP Comparison Study for the Chain Network Setting} that the two lanes share the same latency function $\ell_e$. Our network model and theoretical contributions in Sec. 
\ref{sec: CBCP and DBCP Policies}-\ref{sec: CBCP and DBCP Comparison Study for the Chain Network Setting}
generalize straightforwardly to the setting where the lane index $k = 2$ corresponds to multiple, identical GP lanes on each edge of a real-world traffic network (see Remark \ref{Remark: Latency Function Model for Multiple GP Lanes}).

In Sec. \ref{subsec: DBCP, Formulation} and \ref{subsec: CBCP, Formulation}, we first define the set of feasible flows associated with a given CBCP or DBCP policy, which prescribes flow continuity, non-negativity, and budget constraints.
Then, we specify the latency and toll costs associated with each edge $e \in \edges$, lane $k \in [2]$, period $t \in [T]$, and group $g \in G_e$ for eligible and ineligible users, as a function of flows $\boldy$ that are admissible under a given CBCP or DBCP policy.
Finally, we define equilibrium flows associated with a given CBCP or DBCP policy, which describe the corresponding steady-state user behavior.


\vspace{-1mm}
\begin{remark} \label{Remark: Latency Function Model for Multiple GP Lanes}
Our latency model can be generalized to real-world networks in which each edge comprises one express lane and $n_{GP}$ general-purpose (GP) lanes for some $n_{GP} > 1$, where the impact of traffic flow on travel time is identical across all $n_{GP} + 1$ lanes. Concretely, for each edge $e$, we formulate two latency functions $\ell_{e,1}, \ell_{e,2}: \R_{\geq 0} \ra \R_{\geq 0}$ satisfying the constraint:
{
\setlength{\abovedisplayskip}{4pt}
\setlength{\belowdisplayskip}{4pt}
\begin{align} \label{Eqn: Latency Function, Multiple GP Lanes}
    \ell_{e,2}(z) = \ell_{e,1}\Big(\frac{z}{n_{GP}} \Big), \hspace{1cm} \forall z \geq 0.
\end{align}
}
\hspace{-1mm}In words, the impact of evenly distributing the traffic flow on edge $e$ outside the express lane (represented by $z$ above) over $n_{GP}$ identical GP lanes is identical to the impact of the flow $z/n_{GP}$ on the travel time over a single GP lane.
Likewise, our formulation generalizes to the setting where each edge of the traffic network comprises multiple express lanes and multiple GP lanes. 

Since our model formulation and analysis in Sec. \ref{sec: CBCP and DBCP Policies}, \ref{sec: CBCP and DBCP Equilibrium Computation}, and \ref{sec: CBCP and DBCP Comparison Study for the Chain Network Setting} can be adapted to accommodate any positive integer value of $n_{GP}$, we assume 
in Sec. \ref{sec: CBCP and DBCP Policies}, \ref{sec: CBCP and DBCP Equilibrium Computation}, and \ref{sec: CBCP and DBCP Comparison Study for the Chain Network Setting} that $n_{GP} = 1$. 
Then, by \eqref{Eqn: Latency Function, Multiple GP Lanes}, $\ell_{e,2} = \ell_{e,1}$ for every $e \in \edges$, so we replace $\ell_{e,1}$ and $\ell_{e,2}$ with the same edge latency function $\ell_e$.
\end{remark}

\subsection{Discount-Based Congestion Pricing (DBCP)}
\label{subsec: DBCP, Formulation}

A discount-based congestion pricing (DBCP) policy is characterized by a tuple $(\boldtau, \boldalpha)$, where $\boldtau := (\tau_{e,t}: e \in \edges, t \in [T]) \in \R_{\geq 0}^{|\edges|T}$ denotes the toll levied on express lanes across edges and periods, while $\boldalpha := (\alpha_{e,t}: e \in \edges, t \in [T]) \in [0, 1]^{|\edges|T}$ denotes the discount provided to each eligible group across each edge and period. 

At each period $t$, each ineligible user who accesses the express lane on an edge $e \in \edges$ must pay the full toll $\tau_{e,t}$ out of pocket, while each eligible user pays the \textit{discounted} toll $(1 - \alpha_{e,t}) \tau_{e,t}$ out of pocket. The $(\boldtau, \boldalpha)$-DBCP policy induces a congestion game $\G^D(\boldtau, \boldalpha)$ with a feasible flow set $\Y(\G^D(\boldtau, \boldalpha))$ encoding flow continuity and non-negativity constraints:
{
\setlength{\abovedisplayskip}{2pt}
\setlength{\belowdisplayskip}{2pt}
\begin{align} 
\label{Eqn: Flow Constraint Set, DBCP}
    \Y(\G^D(\boldtau, \boldalpha)) &= \Y_\adm.
\end{align}
}

Under the $(\boldtau, \boldalpha)$-DBCP policy, on each edge $e \in \edges$ and lane $k \in [2]$, at each period $t \in [T]$, the cost incurred for each user in group $g \in G_e$, which we denote by $c_{e,k,t}^g(\cdot; \G^D(\boldtau, \boldalpha) ): \Y(\G^D(\boldtau, \boldalpha)) \ra \R_{\geq 0}$, is given below as a function of user flow $\boldy$:
{
\setlength{\abovedisplayskip}{4pt}
\setlength{\belowdisplayskip}{4pt}
\begin{align} \label{Eqn: Costs, DBCP}
    &c_{e,k,t}^g(y; \G^D(\boldtau, \boldalpha)) := \begin{cases}
        v_t^g \ell_e(x_{e,k,t}) + \tau_{e,t}, \hspace{1cm} &g \in G^I, k = 1, \\
        v_t^g \ell_e(x_{e,k,t}) + (1 - \alpha_{e,t}) \tau_{e,t}, \hspace{1cm} &g \in G^E, k = 1, \\
        v_t^g \ell_e(x_{e,k,t}), &\text{else}
    \end{cases}.
\end{align}
}

In words, at each period $t \in [T]$, on each edge $e \in \edges$, each ineligible user in group $g \in G^I$ who uses the express lane ($k=1$) incurs the full toll cost $\tau_{e,t}$, and the cost of travel latency, $v_t^g \ell_e(x_{e,1,t})$. Meanwhile, each eligible user in group $g \in G^E$ who uses the express lane incurs the \textit{discounted} toll cost $(1 - \alpha_{e,t}) \tau_{e,t}$, and the same cost of travel latency,
i.e., $v_t^g \ell_e(x_{e,1,t})$. Since tolls are not levied on GP lanes ($k = 2$), eligible and ineligible users alike who use the GP lane on edge $e \in \edges$ at period $t \in [T]$ incur only the cost of travel latency, $v_t^g \ell_e(x_{e,2,t})$. Under DBCP policies, constraints in \eqref{Eqn: Flow Constraint Set, DBCP} decouple across periods $t \in [T]$ for \textit{both eligible and ineligible users}, and thus all users' travel decisions decouple across periods.

We now define the Nash equilibrium flows corresponding to a $(\boldtau, \boldalpha)$-DBCP policy, which we term $(\boldtau, \boldalpha)$-DBCP equilibrium flows. As described in Definition \ref{Def: DBCP Equilibria} below, given any $(\boldtau, \boldalpha)$-DBCP policy, the DBCP equilibrium concept captures the steady-state user flow pattern at which no user can strictly lower their cost via a unilateral deviation.

\begin{definition}(\textbf{DBCP Equilibrium}) \label{Def: DBCP Equilibria}
Given tolls $\boldtau \in \R_{\geq 0}^{|\edges|T}$ and discounts $\boldalpha \in [0, 1]^{|\edges|T}$, we call $\boldy^\star$ a ($\boldtau, \boldalpha$)-DBCP equilibrium if $\boldy^\star \in \Y(\G^D(\boldtau, \boldalpha))$ and for each group $g \in \G$:
{
\setlength{\abovedisplayskip}{4pt}
\setlength{\belowdisplayskip}{4pt}
\begin{align} \label{Eqn: DBCP Equilibrium, Def}
    &\sum_{e \in \edges} \sum_{k=1}^2 \sum_{t=1}^T (y_{e,k,t}^g - y_{e,k,t}^{g \star}) c_{e,k,t}^g \big(\boldy^\star; \G^D(\boldtau, \boldalpha) \big) \geq 0, \hspace{1cm} \forall \ \boldy \in \Y(\G^D(\boldtau, \boldalpha)).
\end{align}
}
\hspace{-1mm}We denote the set of all $(\boldtau, \boldalpha)$-DBCP equilibrium flows by $\Y^\eq(\G^D(\boldtau, \boldalpha))$.
\end{definition}

If $\boldalpha = 0$, all users receive zero discount, and thus
the DBCP equilibrium concept reduces to the concept of the Nash equilibrium flow for non-atomic congestion games in which participating users have heterogeneous VoTs \citep{Cole2003PricingNetworkEdgesforHeterogeneousSelfishUsers}.

\subsection{Credit-Based Congestion Pricing (CBCP)}
\label{subsec: CBCP, Formulation}

A credit-based congestion pricing (CBCP) policy is specified by a tuple $(\boldtau, \boldB)$, where $\boldtau := (\tau_{e,t}: e \in \edges, t \in [T]) \in \R_{\geq 0}^{|\edges|T}$ denotes the toll levied on the express lanes in the network $\network$ across edges and periods, while $\boldB := (B^g: g \in G^E) \in \R_{\geq 0}^{|G^E|}$ denotes the total travel credit (i.e., budget) provided to each eligible group for use throughout the $T$ periods. 
\revisionTRPartB{Crucially, the allotted budget may differ across eligible groups and thus across O–D pairs, allowing different levels of subsidized express-lane access throughout the traffic network.}
At each period $t$, each ineligible user who accesses the express lane on an edge $e \in \edges$ must pay the toll $\tau_{e,t}$ out-of-pocket, while eligible users expend their available budget to access the express lane. 
\revisionTS{
Consistent with real-world CBCP policies, in which eligible groups comprise low-income users who are generally unable to pay tolls out-of-pocket, we assume, as in prior work (e.g., \citet{Jalota2022CreditBasedCongestionPricing}), that eligible users must pay for express lane access solely via travel credits.
This assumption is without loss of generality when eligible users' VoTs are sufficiently low, since
such
users would not pay express lane tolls out-of-pocket when traversing consecutive highway segments, even if they are not explicitly prohibited from doing so. For details, see App. \ref{subsec: Low-VoT Users Will Not Voluntarily Pay Express Lane Tolls Out-of-Pocket in Chain Networks}.
}

The $(\boldtau, \boldB)$-CBCP policy induces a congestion game, denoted $\G^C(\boldtau, \boldB)$ below, with a feasible flow set $\Y(\G^C(\boldtau, \boldB))$, encoding flow continuity, non-negativity, and budget constraints:
{
\setlength{\abovedisplayskip}{4pt}
\setlength{\belowdisplayskip}{4pt}
\begin{align} \label{Eqn: Flow Constraint Set, CBCP}
    \Y(\G^C(\boldtau, \boldB)) &:= \Bigg\{ \boldy \in \R_{\geq 0}^{2nT}:
    \boldy \in \Y_\adm, \
    \sum_{t \in [T]} \sum_{e \in \edges} y_{e,1,t}^g \tau_{e,t} \leq B^g, \ \forall \ g \in G^E. \Bigg\}
\end{align}
}

Under the $(\boldtau, \boldB)$-CBCP policy, on each edge $e \in \edges$ and lane $k \in [2]$, at each period $t \in [T]$, the cost incurred for each user in group $g \in G_e$, which we denote by $c_{e,k,t}^g(\cdot; \G^C(\boldtau, \boldB) ): \Y(\G^C(\boldtau, \boldB)) \ra \R_{\geq 0}$, is defined below as a function of the user flow $\boldy \in \Y(\G^C(\boldtau, \boldB))$:
{
\setlength{\abovedisplayskip}{4pt}
\setlength{\belowdisplayskip}{1pt}
\begin{align} \label{Eqn: Costs, CBCP}
    &c_{e,k,t}^g(y; \G^C(\boldtau, \boldB)) := \begin{cases}
        v_t^g \ell_e(x_{e,k,t}) + \tau_{e,t}, \quad &g \in G^I, k = 1, \\
        v_t^g \ell_e(x_{e,k,t}), &\text{else}
    \end{cases}, \forall \ e \in \edges, k \in [2], t \in [T], g \in G_e.
\end{align}
}

In words, \eqref{Eqn: Costs, CBCP} asserts that at each period $t \in [T]$ and edge $e \in \edges$, each ineligible user from group $g \in G^I$ in the express lane ($k=1$) incurs the full toll $\tau_{e,t}$ 
and
the cost of travel latency, $v_t^g \ell_e(x_{e,1,t})$. Meanwhile, each eligible user from group $g \in G^E$ in the express lane pays the toll using their available budget, thus incurring only the cost of travel latency, $v_t^g \ell_e(x_{e,1,t})$. 
All users
who use the toll-free GP lane $(k=2)$ on each edge $e \in \edges$ and period $t \in [T]$ incur only the cost of travel latency, $v_t^g \ell_e(x_{e,2,t})$. We note that, since the constraints \eqref{Eqn: Flow Constraint Set, CBCP} decouple across periods $t \in [T]$ for ineligible users, their travel decisions can be made independently across periods.
\revisionTRPartB{In contrast, the budget constraint couples eligible users' decisions across periods, allowing them to strategically conserve credits for periods in which express lane access provides greater value.
}

With the above definitions of the feasible flow set and costs induced by a $(\boldtau, \boldB)$-CBCP policy, we proceed to define the Nash equilibrium flows under the $(\boldtau, \boldB)$-CBCP policy, which we term $(\boldtau, \boldB)$-CBCP equilibrium flows. As described in Definition \ref{Def: CBCP Equilibria} below, given any $(\boldtau, \boldB)$-CBCP policy, the CBCP equilibrium concept captures the steady-state user flow pattern at which no user can strictly lower their cost via a unilateral deviation.

\begin{definition}(\textbf{CBCP Equilibrium}) \label{Def: CBCP Equilibria}
Given tolls $\boldtau \in \R_{\geq 0}^{|\edges|T}$ and budgets $\boldB \in \R_{\geq 0}^{|G^E|}$, we call $\boldy^\star$ a ($\boldtau, \boldB$)-CBCP equilibrium  if $\boldy^\star \in \Y(\G^C(\boldtau, \boldB))$ and for each group $g \in \G$:
{
\setlength{\abovedisplayskip}{4pt}
\setlength{\belowdisplayskip}{4pt}
\begin{align} \label{Eqn: CBCP Equilibrium, Def}
    &\sum_{e \in \edges} \sum_{k=1}^2 \sum_{t=1}^T (y_{e,k,t}^g - y_{e,k,t}^{g \star}) c_{e,k,t}^g (\boldy^\star; \G^C(\boldtau, \boldB)) \geq 0, \hspace{1cm} \forall \ \boldy \in \Y(\G^C(\boldtau, \boldB)).
\end{align}
}
We denote the set of all $(\boldtau, \boldB)$-CBCP equilibrium flows by $\Y^\eq(\G^C(\boldtau, \boldB))$.
\end{definition}

\begin{remark}
\revisionTS{
While the above discussion defined CBCP policies, budget constraints, and equilibria over flows on network lanes, it is possible to formulate an alternative CBCP model over route-based flows (see App. \ref{subsec: Route Flow Formulation for Credit-Based Congestion Pricing (CBCP)}), in which budget constraints impose route restrictions on eligible groups, and are thus enforced on the level of each individual eligible user.
As described in App. \ref{subsec: Route Flow Formulation for Credit-Based Congestion Pricing (CBCP)}, our main theoretical results (Lemma \ref{Lemma: Main, Explicit Construction of DBCP Policy} and Thm. \ref{Thm: Main}) also hold under this alternate route-based CBCP model.
However, we present our main theoretical and empirical results over lane-based CBCP models, since, as described in App. \ref{subsec: Route Flow Formulation for Credit-Based Congestion Pricing (CBCP)}, it is in general intractable to compute CBCP equilibrium flows under route restrictions induced by budget constraints.
}

\end{remark}



\section{CBCP and DBCP Equilibrium Computation}
\label{sec: CBCP and DBCP Equilibrium Computation}


\revisionTRPartB{In this section, we analyze the properties of DBCP and CBCP equilibrium flows.
Our analysis provides the technical foundation for our main theoretical contributions in Sec. \ref{sec: CBCP and DBCP Comparison Study for the Chain Network Setting}, in which we compare the equilibrium societal cost obtained under CBCP and DBCP policies on highways with tolled express lanes.}

In Sec. \ref{subsec: Existence of DBCP and CBCP Equilibria}, we prove that each CBCP (respectively, DBCP) policy induces at least one corresponding CBCP (respectively, DBCP) equilibrium flow (Prop. \ref{Prop: Existence of DBCP Equilibria}-\ref{Prop: Existence of CBCP Equilibria}). In Sec. \ref{subsec: Convex Programs for Computing DBCP Equilibrium Flows}, we
present a convex program for computing DBCP equilibrium flows
(Prop. \ref{Prop: Convex Program, DBCP}). In Sec. \ref{subsec: Convex Programs for Computing CBCP Equilibrium Flows}, we prove that \textit{when eligible users' VoTs are time-invariant}, 
CBCP equilibrium flows can likewise be formulated as the minimizer set of a convex program (Prop. \ref{Prop: Convex Program, CBCP}). Moreover, when eligible users' VoTs vary across periods, the minimizer set of the same convex program \textit{approximates} the set of equilibrium flows corresponding to a given $(\boldtau, \boldB)$-CBCP policy, with the approximation error scaling linearly in the maximum variation in eligible users' VoTs across periods (Prop. \ref{Prop: CBCP Equilibria, Sensitivity}).

The convex programs we formulate in Sec. \ref{subsec: Convex Programs for Computing DBCP Equilibrium Flows} and \ref{subsec: Convex Programs for Computing CBCP Equilibrium Flows} to characterize DBCP and CBCP equilibria, respectively, differ from the convex programs used to compute CBCP equilibria in  
\citet{Jalota2022CreditBasedCongestionPricing}.
As described in Sec. \ref{sec: Related Work}, \citet{Jalota2022CreditBasedCongestionPricing} consider only the setting of single-edge networks.
In contrast, the convex programs presented below characterize CBCP and DBCP equilibrium flows over acyclic traffic networks with arbitrary structure, shared by groups traveling between different pairs of origin and destination nodes.

\subsection{Existence of DBCP and CBCP Equilibria}
\label{subsec: Existence of DBCP and CBCP Equilibria}

We begin by establishing the existence of CBCP and DBCP equilibria, noting that Prop. \ref{Prop: Existence of DBCP Equilibria} follows from established results on heterogeneous tolling and is included for completeness.

\begin{proposition}(\textbf{Existence of DBCP Equilibria}) 
\label{Prop: Existence of DBCP Equilibria}
Given $\boldtau \in\R_{\geq 0}^{|\edges|T}$ and 
$\boldalpha \in [0, 1]^{|\edges|T}$,
the $(\boldtau, \boldalpha)$-DBCP policy admits a non-empty DBCP equilibrium flow set, i.e., $\Y^{\eq}(\G^D(\boldtau, \boldalpha)) \ne \emptyset$.
\end{proposition}

\begin{proposition}(\textbf{Existence of CBCP Equilibria}) 
\label{Prop: Existence of CBCP Equilibria}
Given $\boldtau \in \R_{\geq 0}^{|\edges|T}$ and $\boldB \in\R_{\geq 0}^{|G^E|}$, the $(\boldtau, \boldB)$-CBCP policy admits a non-empty CBCP equilibrium flow set, i.e., $\Y^{\eq}(\G^C(\boldtau, \boldB)) \ne \emptyset$.
\end{proposition}

\looseness=-1The proofs of Prop. \ref{Prop: Existence of DBCP Equilibria} and \ref{Prop: Existence of CBCP Equilibria} follow by observing that in the variational inequalities that define DBCP and CBCP equilibria, i.e., \eqref{Eqn: DBCP Equilibrium, Def} and \eqref{Eqn: CBCP Equilibrium, Def}, respectively, the cost functions $c_{e,k,t}^g(\cdot; \G^D(\boldtau, \boldalpha))$ and $c_{e,k,t}^g(\cdot; \G^C(\boldtau, \boldB))$ are continuous, and the constraint sets $\Y(\G^D(\boldtau, \boldalpha))$ and $\Y(\G^C(\boldtau, \boldB))$ are compact. The existence of DBCP and CBCP equilibria thus follows from standard arguments in the literature on variational inequalities.



The 
proofs of Prop. \ref{Prop: Existence of DBCP Equilibria} and \ref{Prop: Existence of CBCP Equilibria} merely establish the \textit{existence} of CBCP and DBCP equilibria using 
variational inequalities,
without prescribing a tractable method for \textit{computing} CBCP and DBCP equilibria. Indeed, the problem of efficiently solving variational inequalities is known to be challenging in general
\citep{KinderlehrerStampacchia2000AnIntroductiontoVariationalInequalitiesandTheirApplications}.
Thus, in Sec. \ref{subsec: Convex Programs for Computing DBCP Equilibrium Flows}-\ref{subsec: Convex Programs for Computing CBCP Equilibrium Flows},
we instead characterize CBCP and DBCP equilibrium flows as the minimizer sets of suitably formulated convex programs, which allow us to tractably compute or approximate equilibrium flows.


\subsection{Convex Programs for Computing DBCP Equilibrium Flows}
\label{subsec: Convex Programs for Computing DBCP Equilibrium Flows}


In Prop. \ref{Prop: Convex Program, DBCP} below, we present a convex program for computing DBCP equilibrium flows. 
Our formulation is mathematically equivalent to the well-established convex programming characterization of heterogeneous tolling equilibria \citep{Cole2003PricingNetworkEdgesforHeterogeneousSelfishUsers}, and is included for completeness.



\begin{proposition} \label{Prop: Convex Program, DBCP}
Given a $(\boldtau, \boldalpha)$-DBCP policy with $\boldtau \in\R_{\geq 0}^{|\edges|T}$ and $\boldalpha \in [0, 1]^{|\edges|T}$, a feasible flow $\textbf{y}^\star \in \Y(\G^D(\boldtau, \boldalpha))$ is a $(\boldtau, \boldalpha)$-DBCP equilibrium if and only if $\textbf{y}^\star \in \Y(\G^D(\boldtau, \boldalpha))$ is a minimizer of the following convex program:
{
\setlength{\abovedisplayskip}{4pt}
\setlength{\belowdisplayskip}{4pt}
\begin{subequations}
\label{Eqn: Convex Program Statement, DBCP}
\begin{align} 
    \min_{\substack{\textbf{y} \in \Y(\G^D(\boldtau, \boldalpha)) \\ \revisionTS{\boldx \in \R^{2|\edges|T}}}} \hspace{2mm} &\sum_{t \in [T]} \sum_{e \in \edges} \left[ \sum_{k=1}^2 \int_0^{x_{e,k,t}} 
    \ell_e(w)
    \hspace{0.5mm} dw  + \sum_{g \in \G_e^I} \frac{y_{e,1,t}^g \tau_{e,t}}{v_t^g} + \sum_{g \in \G_e^E} \frac{y_{e,1,t}^g (1 - \alpha_{e,t}) \tau_{e,t}}{v_t^g} \right], \\
    \revisionTS{\emph{s.t.}} \hspace{5mm} &\revisionTS{x_{e,k,t} = \sum_{g \in G_e} y_{e,k,t}^g, \quad \forall \ e \in \edges, k \in [2], t \in [T]}.
\end{align}
\end{subequations}
}
\end{proposition}

Concretely, the proof of Prop. \ref{Prop: Convex Program, DBCP}, which we omit for brevity, follows by establishing that the first-order optimality conditions of the convex program \eqref{Eqn: Convex Program Statement, DBCP} describe the same subset of feasible flows as the variational inequality that defines DBCP equilibria, i.e., \eqref{Eqn: DBCP Equilibrium, Def}. 


\subsection{Convex Programs for Computing CBCP Equilibrium Flows}
\label{subsec: Convex Programs for Computing CBCP Equilibrium Flows}

We first demonstrate that if all eligible users have time-invariant VoTs, CBCP equilibria can be efficiently computed by solving a convex program (Prop. \ref{Prop: Convex Program, CBCP}). Then, we apply sensitivity analysis (Prop. \ref{Prop: CBCP Equilibria, Sensitivity}) to show that even when eligible users' VoTs vary across the $T$ periods, CBCP equilibria are still well approximated by minimizers of the convex program presented in Prop. \ref{Prop: Convex Program, CBCP}.

\begin{proposition}(\textbf{Convex Program for CBCP Equilibria}) 
\label{Prop: Convex Program, CBCP}
Given a $(\boldtau, \boldB)$-CBCP policy with $\boldtau \in\R_{\geq 0}^{|\edges|T}$ and $\boldB \in\R_{\geq 0}^{|G^E|}$, if each eligible group has time-invariant VoT, i.e., for each $g \in G^E$ and $t, t' \in [T]$, we have $v_t^g = v_{t'}^g$, then a feasible flow $\textbf{y}^\star \in \Y(\G^C(\boldtau, \boldB))$ is a $(\boldtau, \boldB)$-CBCP equilibrium if and only if 
it
is a minimizer of the following convex program:
{
\setlength{\abovedisplayskip}{4pt}
\setlength{\belowdisplayskip}{4pt}
\begin{subequations}
\label{Eqn: Convex Program Statement, CBCP}
\begin{align} 
    \min_{\substack{\textbf{y} \in \Y(\G^C(\boldtau, \boldB)) \\ \revisionTS{\boldx \in \R^{2|\edges|T}}}} \hspace{2mm} &\sum_{t \in [T]} \sum_{e \in \edges} \left[ \sum_{k=1}^2 \int_0^{x_{e,k,t}} 
    \ell_e(w)
    \hspace{0.5mm} dw + \sum_{g \in \G_e^I} \frac{y_{e,1,t}^g \tau_{e,t}}{v_t^g} \right], \\
    \revisionTS{\emph{s.t.}} \hspace{1cm} &\revisionTS{x_{e,k,t} = \sum_{g \in G_e} y_{e,k,t}^g, \quad \forall \ e \in \edges, k \in [2], t \in [T]}.
\end{align}
\end{subequations}
}
\end{proposition}


\revisionTRPartB{
We note that, in 
real-world traffic management settings, assuming time-invariant user VoTs often provides a reasonable approximation. 
For instance, when modeling highway traffic flows during weekday rush hours, users may assign similar values to travel-time savings across weekdays.
}


\revisionTRPartB{When eligible users' VoTs vary across periods, however, their intertemporal budget constraints complicate this exact correspondence.
Nonetheless,}
as shown below in Prop. \ref{Prop: CBCP Equilibria, Sensitivity}, we can apply sensitivity analysis to illustrate that minimizers of the convex program \eqref{Eqn: Convex Program Statement, CBCP} approximate CBCP equilibria well even when eligible users' VoTs are time-varying, with the approximation error scaling linearly in the magnitude of the fluctuation of eligible users' VoTs across the $T$ periods.
To facilitate our sensitivity analysis,
we first introduce the notion of an approximate CBCP equilibrium flow.

\begin{definition}(\textbf{$\epsilon$-$(\boldtau, \boldB)$-CBCP equilibrium})
\label{Def: epsilon CBCP Equilibria}
Given tolls $\boldtau \in \R_{\geq 0}^{|\edges|T}$ and budgets $\boldB \in \R_{\geq 0}^{|G^E|}$, we call $\boldy^\star$ an $\epsilon$-($\boldtau, \boldB$)-CBCP equilibrium if $\boldy^\star \in \Y(\G^C(\boldtau, \boldB))$ and, for each group $g \in \G$ and each $y \in \Y(\G^C(\boldtau, \boldB))$:
{
\setlength{\abovedisplayskip}{4pt}
\setlength{\belowdisplayskip}{4pt}
\begin{align} \label{Eqn: epsilon CBCP Equilibrium, Def}
    &\sum_{e \in \edges} \sum_{k=1}^2 \sum_{t=1}^T (y_{e,k,t}^g - y_{e,k,t}^{g \star}) c_{e,k,t}^g (\boldy^\star; \G^C(\boldtau, \boldB)) \geq -\epsilon.
\end{align}
}
\end{definition}

\revisionTRPartB{In words, an $\epsilon$-$(\boldtau, \boldB)$-CBCP equilibrium is an admissible flow at which no infinitesimal user can lower their cost by more than $\epsilon$ through a unilateral route deviation. 
We note that, while \eqref{Eqn: epsilon CBCP Equilibrium, Def} is expressed in terms of group-level aggregate flows, individual users are non-atomic and their unilateral deviations do not affect aggregate flows or travel costs.}
The notion of the $\epsilon$-$(\boldtau, \boldB)$-CBCP equilibrium flow is closely related to the well-studied concept of the $\epsilon$-Nash equilibrium in the game theory literature, which reduces to the Nash equilibrium when $\epsilon = 0$ \citep{Cole2003PricingNetworkEdgesforHeterogeneousSelfishUsers, Bubelis1979OnEquilibriainFiniteGames, Daskalakis2009TheComplexityofComputingaNashEquilibrium}.
Similarly, 
when $\epsilon = 0$, the sets of $(\boldtau, \boldB)$-CBCP equilibria and $\epsilon$-$(\boldtau, \boldB)$-CBCP equilibria likewise coincide.

Our sensitivity analysis of the minimizer set of the convex program \eqref{Eqn: Convex Program Statement, CBCP} proceeds as follows. Let $\delta_v$ be defined as the maximum amount of fluctuation in eligible users' VoTs across periods:
{
\setlength{\abovedisplayskip}{4pt}
\setlength{\belowdisplayskip}{4pt}
\begin{align} \label{Eqn: delta v, Def}
    \delta_v := \frac{1}{2} \max_{g \in G^E} \left( \max_{t \in [T]} v_t^g - \min_{t \in [T]} v_t^g \right).
\end{align}
}
\hspace{-1mm}Below, Prop. \ref{Prop: CBCP Equilibria, Sensitivity} states that any solution to the convex program \eqref{Eqn: Convex Program Statement, CBCP} is an $\epsilon$-$(\boldtau, \boldB)$-equilibrium flow, where $\epsilon$ is directly proportional to $\delta_v$, the horizon $T$, and the worst-case maximum total latency collectively incurred by an eligible user group, as given by $\max_{g \in G^E} d^g \sum_{e \in \edges: g \in G_e^E} \ell_e(d_e)$.
\revisionTRPartB{In particular, for any desired tolerance $\epsilon > 0$, the bound prescribed by Prop. \ref{Prop: CBCP Equilibria, Sensitivity} can be inverted to determine the maximum allowable VoT variation $\delta_v$ for which solutions to \eqref{Eqn: Convex Program Statement, CBCP} remain $\epsilon$-CBCP equilibria.}

\begin{proposition}(\textbf{Sensitivity Analysis for CBCP Equilibria}) 
\label{Prop: CBCP Equilibria, Sensitivity}
Given a $(\boldtau, \boldB)$-CBCP policy with $\boldtau \in \R_{\geq 0}^{|\edges|T}$ and $\boldB \in \R_{\geq 0}^{|G^E|}$, if a feasible flow $\bar \boldy \in \Y(\G^C(\boldtau, \boldB))$ is a minimizer of the convex program \eqref{Eqn: Convex Program Statement, CBCP},
then $\bar \boldy \in \Y(\G^C(\boldtau, \boldB))$ is an $\epsilon$-$(\boldtau, \boldB)$-CBCP equilibrium, where:
{
\setlength{\abovedisplayskip}{4pt}
\setlength{\belowdisplayskip}{4pt}
\begin{align} \label{Eqn: epsilon for CBCP equilibrium approximation, Def}
    \epsilon := \delta_v T \cdot \max_{g \in G^E} \Bigg\{ d^g \sum_{e \in \edges: g \in G_e^E} \ell_e(d_e) \Bigg\}.
\end{align}
}
\end{proposition}

\begin{proof}(\textbf{Proof Sketch})
We note that the formulation of the convex program \eqref{Eqn: Convex Program Statement, CBCP} does not explicitly depend upon the eligible users' VoTs. 
Thus, the minimizers of \eqref{Eqn: Convex Program Statement, CBCP} are unchanged when the eligible users' time-varying VoTs are replaced by time-invariant values. 
Thus, by Prop. \ref{Prop: Convex Program, CBCP}, $\bar y$ is a $(\boldtau, \boldB)$-CBCP equilibrium when each eligible group's VoT is fixed at:
{
\setlength{\abovedisplayskip}{4pt}
\setlength{\belowdisplayskip}{4pt}
\begin{align} \label{Eqn: tilde v g, Def}
    \tilde v^g := \frac{1}{2} \Big( \max_{t \in [T]} v_t^g + \min_{t \in [T]} v_t^g \Big).
\end{align}
}
We then complete the proof of Prop. \ref{Prop: CBCP Equilibria, Sensitivity} by applying the definition of CBCP equilibria (Def. \ref{Def: CBCP Equilibria}) to $\bar\boldy$, to establish that $\bar\boldy$ satisfies \eqref{Eqn: epsilon CBCP Equilibrium, Def}, with $\epsilon$ given by \eqref{Eqn: epsilon for CBCP equilibrium approximation, Def}. \hfill $\square$
\end{proof}

For the full proof of Prop. \ref{Prop: CBCP Equilibria, Sensitivity}, see App. 
\ref{subsec: App, Proof of Prop, CBCP Equilibria, Sensitivity}.
We note that, if the VoT of each eligible group is time-invariant, then $\delta_v = 0$, and Prop. \ref{Prop: CBCP Equilibria, Sensitivity} reduces to Prop. \ref{Prop: Convex Program, CBCP}.


\section{CBCP and DBCP Comparison Study for the Chain Network Setting}
\label{sec: CBCP and DBCP Comparison Study for the Chain Network Setting}

Although both CBCP \citep{CBCP-SanMateo} and DBCP \citep{DBCP-ExpressLaneSTART} have been deployed for years in real-world traffic networks in California, comparisons between the relative merits of CBCP and DBCP policies have yet to be formalized.
Thus, in this section, we compare the 
impact of deploying CBCP and DBCP on traffic systems with tolled express lanes. 
\revisionTRPartB{We first formulate societal cost objectives that account for eligible and ineligible users' travel costs and collected toll revenues (Sec. \ref{subsec: Societal Objective}). We then establish conditions under which DBCP attains lower societal costs under equilibrium flows compared to CBCP (Sec. \ref{subsec: When DBCP Outperforms CBCP}), and characterize settings in which CBCP may instead outperform DBCP in minimizing the equilibrium societal cost (Sec. \ref{subsec: When CBCP Can Outperform DBCP}).}

\subsection{Societal Cost Objective}
\label{subsec: Societal Objective}

We begin by defining \textit{societal cost terms} and \textit{societal welfare constraints} which encode the efficiency and equity objectives of congestion pricing, and are commonly utilized in the literature on redistributive market mechanisms \citep{Kang2022OptimalDesignForRedistributionsAmongEndogenousBuyersandSellers, Dworczak2021RedistributionThroughMarkets, Akbarpour2020RedistributiveAllocationMechanisms}.
In particular, our societal costs and constraints encode travel cost reduction for eligible and ineligible users, toll revenue collection, and express lane access guarantees for eligible users.
We then formalize the comparison of the merits and drawbacks of CBCP and DBCP policies as the comparison between the minimum attainable societal cost at equilibrium under CBCP and DBCP policies, respectively, subject to a given set of societal welfare constraints.

Let $\network = (\nodes, \edges)$ denote a traffic network with the express and GP lane structure defined in Sec. \ref{sec: CBCP and DBCP Policies}, and let $G$
denote a group of users traversing $\network$. Given any $(\boldtau, \boldB)$-CBCP policy deployed over $T$ periods, and any set of feasible user flows $\boldy \in \Y(\G^C(\boldtau, \boldB))$ with associated lane flows $\boldx \in \R^{2|\edges|T}$, we consider the societal cost objective 
$f_{C, \boldlambda}: \R_{\geq 0}^{2nT} \times \R^{2|\edges|T} \times \R^{|\edges|T} \ra \R$, 
given by:
{
\setlength{\abovedisplayskip}{4pt}
\setlength{\belowdisplayskip}{4pt}
\begin{align} \label{Eqn: f C lambda, Def}
    f_{C, \boldlambda}(\boldy, \boldx, \boldtau)
    = \ &\lambda_E \cdot \sum_{t=1}^T \sum_{e \in \edges} \sum_{k=1}^2 \sum_{g \in G_e^E} v_t^g \ell_e(x_{e,k,t}) y_{e,k,t}^g - \lambda_R \cdot \sum_{t=1}^T \sum_{e \in \edges} \sum_{g \in G_e^I} \tau_{e,t} y_{e,1,t}^g \\ \nonumber
    &\hspace{1cm} + \lambda_I \cdot \sum_{t=1}^T \sum_{e \in \edges} \sum_{k=1}^2 \sum_{g \in G_e^I} \big( v_t^g \ell_e(x_{e,k,t}) + \textbf{1}\{k=1\} \tau_{e,t} \big) y_{e,k,t}^g,
\end{align}
}
\hspace{-1mm}where $\boldlambda := (\lambda_E, \lambda_R, \lambda_I)$, and $\lambda_E, \lambda_R, \lambda_I \geq 0$ are weights that encode the degree to which the eligible users' travel cost, total collected toll revenue, and ineligible users' travel cost contribute to the overall societal cost objective, respectively. 

\revisionTRPartB{
Our societal cost formulation captures a broad and practically relevant class of policy objectives that arise in real-world congestion pricing programs, which must balance reducing travel costs for users across income levels with generating revenues to finance socially beneficial transportation projects, such as public transit investments. Specifically, by assigning separate weights to eligible-user costs, ineligible-user costs, and toll revenues, our formulation accommodates different societal priorities rather than prescribing a single welfare criterion. Importantly, when $\lambda_E=\lambda_I=\lambda_R$, toll payments and revenues cancel, and the societal cost reduces to the aggregate value-of-time-weighted travel time, recovering the conventional system-efficiency objective. More generally, varying $\boldlambda$ allows us to capture objectives ranging from minimizing the travel costs of particular user groups to maximizing toll revenues, as well as weighted combinations that explicitly account for equity and revenue-generation considerations. Consistent with welfare formulations for redistributive allocation mechanisms \citep{Akbarpour2020RedistributiveAllocationMechanisms, Dworczak2021RedistributionThroughMarkets}, the weight $\lambda_R$ reflects the societal value of the public investments and redistributive programs financed by toll revenues, rather than treating toll payments themselves as generating social surplus. Thus, our formulation provides a flexible and economically interpretable framework for evaluating the fundamental trade-offs faced by congestion pricing authorities, while remaining amenable to theoretical analysis. Crucially, our theoretical comparisons do not depend on selecting particular numerical weights. 
Rather, they establish policy performance guarantees uniformly over an entire class of societal objectives satisfying $\lambda_R\geq\lambda_E$, while also identifying settings in which the relative performance of the two policies can reverse outside this class. Although this class captures several central policy objectives, extending the analysis to alternative societal welfare formulations
remains an important direction for future work.
}

Similarly, given any $(\hat \boldtau, \boldalpha)$-DBCP policy deployed over $T$ periods, and any set of feasible user flows $\boldy \in \Y(\G^D(\hat \boldtau, \boldalpha))$ with associated lane flows $\boldx \in \R^{2|\edges|T}$, we consider the societal cost objective 
$f_{D, \boldlambda}: \R_{\geq 0}^{2nT} \times \R^{2|\edges|T} \times \R^{|\edges|T} \times [0, 1]^{|\edges|T} \ra \R$,
given by:
{
\setlength{\abovedisplayskip}{4pt}
\setlength{\belowdisplayskip}{4pt}
\begin{align} \label{Eqn: f D lambda, Def}
    f_{D, \boldlambda}(\hat \boldy, \hat \boldx, \hat \boldtau, \boldalpha) &= \lambda_E \cdot \sum_{t=1}^T \sum_{e \in \edges} \sum_{k=1}^2 \sum_{g \in G_e^E} \big( v_t^g \ell_e(\hat x_{e,k,t}) + \textbf{1}\{k=1\} (1 - \alpha_{e,t}) \hat \tau_{e,t} \big) \hat y_{e,k,t}^g \\ \nonumber
    &\hspace{1cm} + \lambda_I \cdot \sum_{t=1}^T \sum_{e \in \edges} \sum_{k=1}^2 \sum_{g \in G_e^I} \big( v_t^g \ell_e(\hat x_{e,k,t}) + \textbf{1}\{k=1\} \hat \tau_{e,t} \big) \hat y_{e,k,t}^g \\ \nonumber
    &\hspace{1cm} - \lambda_R \cdot \sum_{t=1}^T \sum_{e \in \edges} \Bigg( \sum_{g \in G_e^E} (1 - \alpha_{e,t}) \hat \tau_{e,t} \hat y_{e,1,t}^g + \sum_{g \in G_e^I} \hat \tau_{e,t} \hat y_{e,1,t}^g \Bigg).
\end{align}
}

We aim to compare the effectiveness of CBCP and DBCP at minimizing their respective societal cost metric while adhering to constraints 
that encode the CBCP and DBCP equilibrium conditions (Def. \ref{Def: DBCP Equilibria} and \ref{Def: CBCP Equilibria}), as well as additional specifications 
on the flows $\boldy$ and tolls $\boldtau$ that the traffic authority may wish to enforce. 
Concretely, let $h: \Y_\adm \times \R^{|\edges|T} \ra \R^{n_c}$ encode $n_c$ constraints on the total eligible user flow on each lane, the user flow of each ineligible group on each lane, and the values of tolls deployed on express lanes across all edges and periods. 

We define the optimal societal cost under CBCP with weights $\boldsymbol{\lambda} := (\lambda_E, \lambda_R, \lambda_I) \in \R_{\geq 0}^3$ by:
{
\setlength{\abovedisplayskip}{3pt}
\setlength{\belowdisplayskip}{1pt}
\begin{subequations} \label{Eqn: f C lambda star, Def}
\begin{align} \label{Eqn: f C lambda star, Def, Objective}
    f_{C, \boldlambda}^\star := 
    \inf_{\boldy, \boldx, \boldtau, \boldB} 
    \hspace{1cm} 
    &f_{C, \boldlambda}(\boldy, \boldx, \boldtau)
    \\ \label{Eqn: f C lambda star, Def, Constraint on non-negative tau, B, y}
    \text{s.t.} \hspace{1cm} &\boldtau \in \R_{\geq 0}^{|\edges|T}, \ \boldB \in \R_{\geq 0}^{|G^E|}, 
    \
    \boldy \in \Y^{\eq}(\G^C(\boldtau, \boldB)), \\ \label{Eqn: f C lambda star, Def, Constraint on x def}
    &x_{e,k,t} = \sum_{g \in G_e} y_{e,k,t}^g, \hspace{5mm} \forall e \in \edges, k \in [2], t \in [T], \\ \label{Eqn: f C lambda star, Def, Constraint of h}
    &h(\boldy, \boldtau) \leq 0.
\end{align}
\end{subequations} 
}
\hspace{-1mm}\revisionTRPartB{Above, \eqref{Eqn: f C lambda star, Def, Constraint on non-negative tau, B, y}-\eqref{Eqn: f C lambda star, Def, Constraint of h} respectively enforce the CBCP equilibrium conditions, define the aggregate lane flows, and impose the societal welfare constraints encoded by the function $h$.
}
Similarly, we define the optimal societal cost under DBCP policies with weights $\boldsymbol{\lambda} = (\lambda_E, \lambda_R, \lambda_I) \in \R_{\geq 0}^3$ by:
{
\setlength{\abovedisplayskip}{4pt}
\setlength{\belowdisplayskip}{1pt}
\begin{subequations} \label{Eqn: f D lambda star, Def}
\begin{align} \label{Eqn: f D lambda star, Def, Objective}
    f_{D, \boldlambda}^\star := 
    \inf_{\hat \boldy, \hat \boldx, \hat \boldtau, \boldalpha}
    \hspace{1cm} &f_{D, \boldlambda}(\hat \boldy, \hat \boldx, \hat \boldtau, \boldalpha) \\ \label{Eqn: f D lambda star, Def, Constraint on non-negative tau, alpha}
    \text{s.t.} \hspace{1cm} &\hat \boldtau \in \R_{\geq 0}^{|\edges|T}, \ \boldalpha \in [0, 1]^{|\edges|T}, 
    \
    \hat \boldy \in \Y^{\eq}(\G^D(\hat \boldtau, \boldalpha)), \\ \label{Eqn: f D lambda star, Def, Constraint on x def}
    &\hat x_{e,k,t} = \sum_{g \in G_e} \hat y_{e,k,t}^g, \hspace{5mm} \forall e \in \edges, k \in [2], t \in [T], \\ \label{Eqn: f D lambda star, Def, Constraint of h}
    &h(\hat \boldy, \hat \boldtau) \leq 0.
\end{align}
\end{subequations}
}

To facilitate our theoretical analysis, 
we impose mild conditions on the constraint function $h$. 

\begin{assumption} \label{Assum, h: zero sublevel set of h is bounded}
The constraint function $h: \Y_\adm \times \R^{|\edges|T} \ra \R^{n_c}$
is continuous and has a bounded zero-sublevel set, i.e., the set $\{(\boldy, \boldtau) \in \Y_\adm \times \R^{|\edges|T}: h(\boldy, \boldtau) \leq 0 \}$ is bounded.
\end{assumption}


\begin{assumption} \label{Assum, h: Eligible user total express lane flow}
Given any toll $\boldtau \in \R_{\geq 0}^{|\edges|T}$, and any flows $\boldy, \bar \boldy \in \Y_\adm$ satisfying, across all $e \in \edges$, $k \in [2]$, and $t \in [T]$, (1) $\sum_{g \in G_e^E} y_{e,k,t}^g = \sum_{g \in G_e^E} \bar y_{e,k,t}^g$ and (2) $y_{e,k,t}^{g'} = \bar y_{e,k,t}^{g'}$ for all $g' \in G_e^I$, we have $h(\boldy, \boldtau) = h(\bar \boldy, \boldtau)$.
\end{assumption}

In words, we require $h$ to be continuous with a bounded zero-sublevel set
(Assumption \ref{Assum, h: zero sublevel set of h is bounded}). Moreover, we allow $h$ to have arbitrary dependence on ineligible user flows, 
but require $h$ to depend only upon the \textit{total eligible user flow} on each lane, edge, and period, i.e., $\{\sum_{g \in G_e^E} y_{e,k,t}^g: e \in \edges, k \in [2], t \in [T] \}$, and not on the individual eligible user flow components $\{y_{e,k,t}^g: e \in \edges, k \in [2], t \in [T], g \in G_e^E \}$
(Assumption \ref{Assum, h: Eligible user total express lane flow}).
\revisionTRPartB{
While Assumption 2 requires the constraints encoded by $h(\cdot)$ to depend on eligible users' aggregate flows rather than individual group-level flows, our framework remains flexible enough to capture a broad range of practically relevant congestion management objectives. 
Such objectives include express lane access guarantees for eligible users, quality of service considerations, and restrictions on toll implementation (Ex. \ref{Ex: h, Eligible User Express Lane Access Guarantees}-\ref{Ex: h, Upper Bounds for Express Lane Tolls}).
}

\begin{example}(\textbf{Express Lane Access Guarantees for Eligible Users}) 
\label{Ex: h, Eligible User Express Lane Access Guarantees}
To ensure a pre-specified, baseline level $m_y^E > 0$ of express lane access to eligible groups, $h$ can encode constraints of 
the form $\sum_{g \in G_e^E} y_{e,1,t}^g \geq m_y^E$ on each edge $e \in \edges$ and at each period $t \in [T]$.
Enforcing $\sum_{g \in G_e^E} y_{e,1,t}^g \geq m_y^E$
prevents eligible users from being priced out of express lane access, 
an inequity concern at the source of much criticism directed towards existing congestion pricing schemes.
\end{example}

\begin{example}(\textbf{Quality of Service Guarantees for the Express Lane}) 
\label{Ex: h, Quality of Service Guarantees for Express Lane}
To uphold the service quality of the express lane, $h$ can enforce that express lane travel is faster than travel on the GP lane by at least a pre-designed margin $m_\ell > 0$. Specifically, in our model, travel latencies depend only upon the lane flows $x_{e,k,t}$, which in turn depend only on the \textit{total} eligible and ineligible user flows on each lane.
Thus, under Assumption \ref{Assum, h: Eligible user total express lane flow}, $h$ can enforce constraints of the form $\ell_e(x_{e,1,t}) \leq \ell_e(x_{e,2,t}) - m_\ell$, which
guarantees that express lane travel is reliably faster than travel on the GP lane (by a fixed margin $m_\ell$), a
key efficiency objective of congestion pricing.
\end{example}

\begin{example}(\textbf{Upper Bounds for Express Lane Tolls}) \label{Ex: h, Upper Bounds for Express Lane Tolls}
Since the express lane tolls $\boldtau$ are also inputs to $h$, the system designer can design $h$ to 
restrict tolls to a pre-specified range at each edge and period. Concretely, given upper bounds $M_{\boldtau} = ( (M_{\boldtau})_{e,t} \geq 0: e \in \edges, t \in [T]) \in \R_{\geq 0}^{|\edges|T}$ for toll values across edges and periods, $h$ can be designed to enforce constraints such as $\tau_{e,t} \leq (M_{\boldtau})_{e,t}, \forall e \in \edges, t \in [T]$.
In particular, by setting $(M_{\boldtau})_{e,t} = 0 \ \forall t \in [T]$ for all edges $e$ in a strict subset of edges $\edges_{\boldtau} \subset E$, 
the system designer could restrict 
express lane tolls to be levied only on the edges in $\edges \backslash \edges_{\boldtau}$, effectively instituting a \textit{second-best} tolling mechanism, i.e., a tolling mechanism in which a pre-determined subset of lanes must remain toll-free.
\end{example}

We let $S_{C, \boldlambda}$ and $S_{D, \boldlambda}$ denote the feasible sets of the optimization programs \eqref{Eqn: f C lambda star, Def} and \eqref{Eqn: f D lambda star, Def}
respectively, and let $S_{C, \boldlambda}^\star$ and $S_{D, \boldlambda}^\star$ denote the sets of minimizers of \eqref{Eqn: f C lambda star, Def} and \eqref{Eqn: f D lambda star, Def}, respectively. 

We now establish that under Assumption \ref{Assum, h: zero sublevel set of h is bounded}, \eqref{Eqn: f C lambda star, Def} and \eqref{Eqn: f D lambda star, Def} yield non-empty minimizer sets whenever their respective feasible sets are non-empty (Prop. \ref{Prop: CBCP Minimizer Set is Non-empty} and \ref{Prop: DBCP Minimizer Set is Non-empty}). The existence of minimizers for \eqref{Eqn: f C lambda star, Def} and \eqref{Eqn: f D lambda star, Def} will be invoked when we establish conditions under which DBCP policies are strictly more effective than CBCP policies at reducing societal cost (Cor. \ref{Cor: Main Thm, Strictness}).

\begin{proposition} \label{Prop: CBCP Minimizer Set is Non-empty}
Under Assumption \ref{Assum, h: zero sublevel set of h is bounded}, if
$S_{C, \boldlambda} \ne \emptyset$, then $S_{C, \boldlambda}^\star \ne \emptyset$, i.e.,
if the optimization program \eqref{Eqn: f C lambda star, Def} is feasible,
then \eqref{Eqn: f C lambda star, Def} admits at least one minimizer $(\boldy^\star, \boldx^\star, \boldtau^\star, \boldB^\star)$.
\end{proposition}

\begin{proof}(\textbf{Proof Sketch})
Since $f_{C, \boldlambda}$ is continuous, to prove that \eqref{Eqn: f C lambda star, Def} admits at least one minimizer $(\boldy^\star, \boldx^\star, \boldtau^\star, \boldB^\star)$, it suffices to prove that the constraint set $S_{C, \boldlambda}$, as given by \eqref{Eqn: f C lambda star, Def, Constraint on non-negative tau, B, y}-\eqref{Eqn: f C lambda star, Def, Constraint of h}, can be restricted to a compact set without affecting the minimum societal objective encoded by \eqref{Eqn: f C lambda star, Def}. 
We first prove that, given any $\boldlambda$,
we may restrict our attention to bounded tolls $\boldtau$ and budgets $\boldB$, since increasing tolls or budgets past a certain range can never strictly lower the societal cost.
We then prove that, when $\boldtau$ and $\boldB$ are bounded by compact sets, the set $\{(\boldy, \boldx, \boldtau, \boldB): \boldy \in \Y^{\eq}(\G^C(\boldtau, \boldB))\}$ is likewise compact, and thus the constraint set \eqref{Eqn: f C lambda star, Def, Constraint on non-negative tau, B, y} can be replaced by a compact set without affecting the minimum societal cost encoded by \eqref{Eqn: f C lambda star, Def}. Finally, by Assumption \ref{Assum, h: zero sublevel set of h is bounded}, the zero sub-level set of the continuous map $h$ is bounded, hence compact, so the constraint sets 
\eqref{Eqn: f C lambda star, Def, Constraint on x def} and \eqref{Eqn: f C lambda star, Def, Constraint of h} are compact. The proof is then complete.
\hfill $\square$
\end{proof}

\begin{proposition} \label{Prop: DBCP Minimizer Set is Non-empty}
Under Assumption \ref{Assum, h: zero sublevel set of h is bounded}, if
$S_{D, \boldlambda} \ne \emptyset$, then $S_{D, \boldlambda}^\star \ne \emptyset$,
i.e., if the optimization program \eqref{Eqn: f D lambda star, Def} is feasible,
then \eqref{Eqn: f D lambda star, Def} admits at least one minimizer $(\hat\boldy^\star, \hat\boldx^\star, \hat\boldtau^\star, \boldalpha^\star)$.
\end{proposition}

The full proofs of Prop. \ref{Prop: CBCP Minimizer Set is Non-empty} and \ref{Prop: DBCP Minimizer Set is Non-empty}, which are similar, have been omitted for brevity.



\subsection{When DBCP Outperforms CBCP in Societal Cost Minimization}
\label{subsec: When DBCP Outperforms CBCP}

Below,
we study conditions under which either CBCP or DBCP outperforms the other in minimizing the corresponding equilibrium societal cost, defined in Sec. \ref{subsec: Societal Objective} as $f_{C, \boldlambda}^\star$ and $f_{D, \boldlambda}^\star$. Since $f_{C, \boldlambda}^\star$ and $f_{D, \boldlambda}^\star$ are in general difficult to compute (Remark \ref{Remark: Direct Comparison of f C star and f D star}), we compare $f_{C, \boldlambda}^\star$ and $f_{D, \boldlambda}^\star$ via a \textit{constructive} approach, under additional conditions on the network structure and user population which realistically reflect facets of real-world congestion pricing over highway systems. 
Specifically, $f_{D, \boldlambda}^\star \leq f_{C, \boldlambda}^\star$ would hold if, given any $(\boldtau, \boldB)$-CBCP policy that is feasible under the constraints in \eqref{Eqn: f C lambda star, Def}, we can \textit{construct} a $(\hat \boldtau, \boldalpha)$-DBCP policy, feasible under the constraints \eqref{Eqn: f D lambda star, Def}, that induces a lower societal cost compared to the societal cost induced by the $(\boldtau, \boldB)$-CBCP policy. Motivated by the above observation, in Sec. \ref{subsubsec: Key Assumptions and Their Implications} we introduce key conditions on the network structure (Assumption \ref{Assum: Chain Network}) and eligible users' VoTs (Assumption \ref{Assum: Eligible users cannot access express lane without subsidy under DBCP}). Next, in Sec. \ref{subsubsec: Main Result}, we prove that 
when Assumptions \ref{Assum, h: zero sublevel set of h is bounded}-\ref{Assum: Eligible users cannot access express lane without subsidy under DBCP} hold \textit{and the societal cost definition prioritizes maximizing toll revenue over minimizing eligible users' costs (i.e., $\lambda_R \geq \lambda_E$)}, given any feasible CBCP policy, we can construct a feasible DBCP policy that admits an equilibrium flow attaining a lower societal cost (Lemma \ref{Lemma: Main, Explicit Construction of DBCP Policy}). We then illustrate how the insights of Lemma \ref{Lemma: Main, Explicit Construction of DBCP Policy} serve a crucial role in establishing that under Assumptions \ref{Assum, h: zero sublevel set of h is bounded}-\ref{Assum: Eligible users cannot access express lane without subsidy under DBCP} and the condition $\lambda_R \geq \lambda_E$, DBCP outperforms CBCP in minimizing the societal cost, i.e., $f_{D, \boldlambda}^\star \leq f_{C, \boldlambda}^\star$ (Thm. \ref{Thm: Main}).


\begin{remark}(\textbf{Challenges of Directly Computing and Comparing $f_{C, \boldlambda}^\star$ and $f_{D, \boldlambda}^\star$})
\label{Remark: Direct Comparison of f C star and f D star}
Directly computing and comparing $f_{C, \boldlambda}^\star$ and $f_{D, \boldlambda}^\star$, which entail solving
\eqref{Eqn: f C lambda star, Def} and \eqref{Eqn: f D lambda star, Def}, pose several challenges. First, the constraint $h(\cdot) \leq 0$ in \eqref{Eqn: f C lambda star, Def} and \eqref{Eqn: f D lambda star, Def} can in general encode a non-convex constraint set, since $h$ is not required to be convex. Further, the equilibrium flow sets $\Y^\eq(\G^C(\boldtau, \boldB))$ and $\Y^\eq(\G^D(\hat \boldtau, \boldalpha))$ in \eqref{Eqn: f C lambda star, Def, Constraint on non-negative tau, B, y} and \eqref{Eqn: f D lambda star, Def, Constraint on non-negative tau, alpha}, respectively, are generally encoded either using variational inequalities (Def. \ref{Def: DBCP Equilibria} and \ref{Def: CBCP Equilibria}), or as minimizer sets of convex programs (Prop. \ref{Prop: Convex Program, DBCP} and \ref{Prop: Convex Program, CBCP}). Thus, to analytically characterize $f_{C, \boldlambda}^\star$ and $f_{D, \boldlambda}^\star$, one must in general solve either a semi-infinite or bilevel optimization problem with non-convex constraints, which are in general NP-hard.
\end{remark}


\subsubsection{Key Assumptions and Their Implications}
\label{subsubsec: Key Assumptions and Their Implications}

We begin by specifying Assumptions \ref{Assum: Chain Network}-\ref{Assum: Eligible users cannot access express lane without subsidy under DBCP}, and noting their practical relevance to the design of real-world CBCP and DBCP policies.

First, Assumption \ref{Assum: Chain Network} asserts that the traffic network under study, $\network$, consists of a consecutive sequence of edges (see Fig. \ref{fig: Chain Network}), and thus forms a chain (or \textit{path}) in a graph-theoretic sense. We henceforth refer to traffic networks satisfying Assumption \ref{Assum: Chain Network} as \textit{chain networks}.
\vspace{-1mm}

\begin{assumption} \label{Assum: Chain Network}
For each $e \in \edges$, we require: 
(1) Either $\edges_{i_e}^\In = \{e'\}$ for some $e' \in \edges$, or $\edges_{i_e}^\In = \emptyset$, and (2) either $\edges_{j_e}^\Out = \{e'\}$ for some $e' \in \edges$, or $\edges_{j_e}^\Out = \emptyset$. 
Here, $\emptyset$ denotes the empty set.
\end{assumption}

\vspace{-1mm}


Our focus on chain networks is motivated by the structure of real-world CBCP and DBCP deployments, which primarily target contiguous highway corridors with tolled express lanes \citep{CBCP-SanMateo, MTA2024LowIncomeDiscountPlan}. Importantly, this setting accommodates multiple origin-destination pairs, with users entering and exiting at different highway segments, and thus captures key features of practical express-lane systems. From a theoretical perspective, the chain structure enables a decomposition of DBCP equilibria across network edges and periods (Lemma \ref{Lemma: DBCP Eq Decomposition, Chain Network}), which is central to our comparison of DBCP and CBCP policies. 
Such a decomposition does not generally hold on arbitrary networks with branching routes, where routing decisions across different edges are coupled, introducing additional challenges in establishing comparable theoretical guarantees. We emphasize that our equilibrium existence and computational results in Sec. \ref{sec: CBCP and DBCP Equilibrium Computation} apply to general acyclic networks, while the comparative guarantees in Sec. \ref{sec: CBCP and DBCP Comparison Study for the Chain Network Setting} exploit this additional structure. Moreover, our main results extend to settings incorporating congestible, toll-free off-highway routes as outside options (App. \ref{subsec: Latency Characteristics of Outside Options in Chain Networks}). Extending the comparative analysis to general network topologies with competing routes represents an important direction for future research.

\begin{figure}
    \centering
    \includegraphics[width=0.79\linewidth]{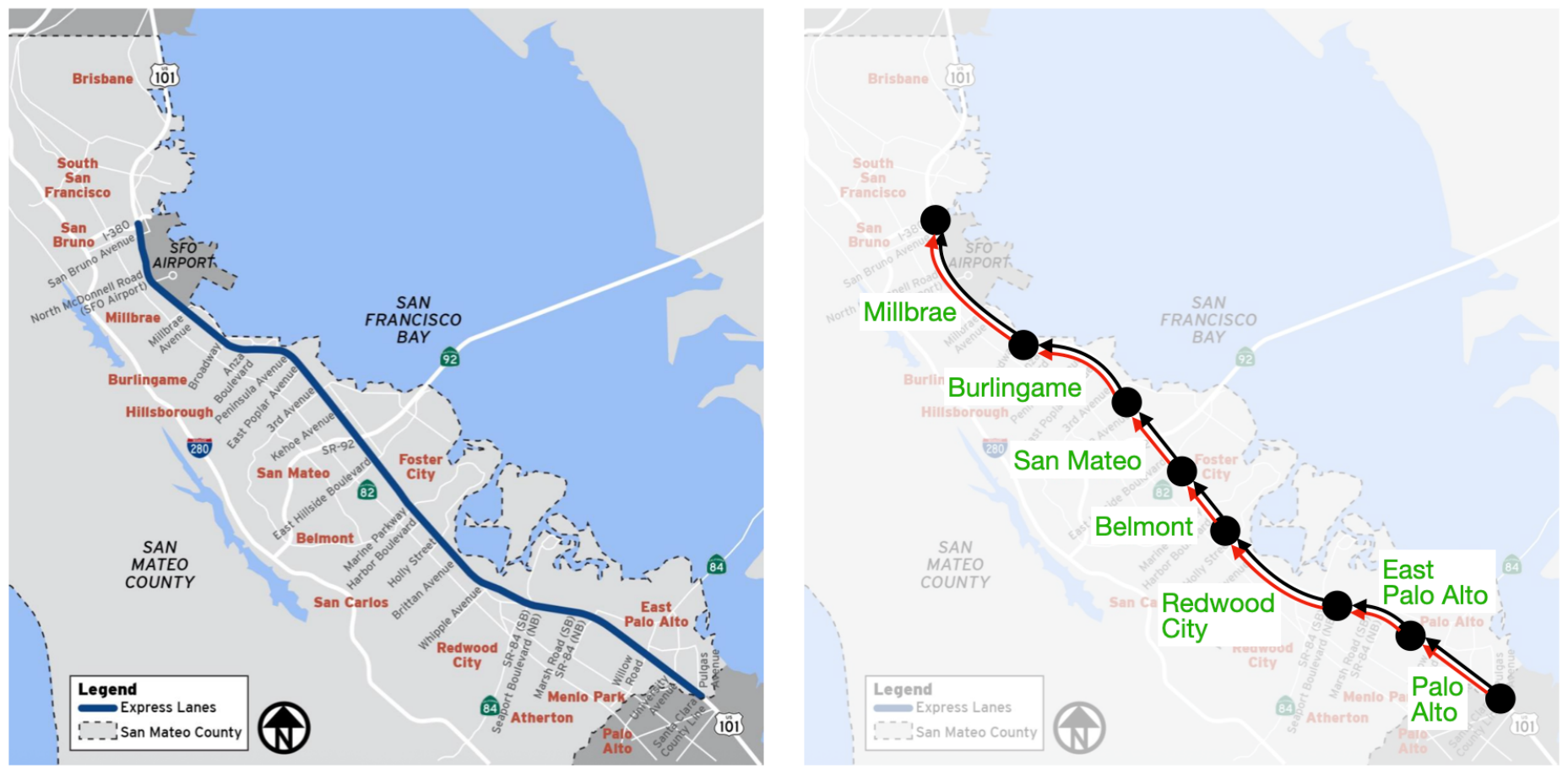}
    \caption{\small \sf
    (Left) Schematic for the
    101 Express Lanes Project \citep{CBCP-SanMateo}, a CBCP policy deployed on the Northbound US-101 freeway from Palo Alto to San Bruno. (Right) Chain network representation of the 
    101 Express Lanes Project, consisting of consecutive edges representing cities (in green) within the 101 Express Lanes Project's jurisdiction. Each edge comprises an express lane that may be tolled (red, $k=1$) and a toll-free GP lane (black, $k=2$).
    }
    \label{fig: Chain Network}
    \vspace{-3mm}
\end{figure}



Next, Assumption \ref{Assum: Eligible users cannot access express lane without subsidy under DBCP} states that at any DBCP equilibrium, no eligible user is able to access the express lane without subsidy.

\vspace{-2mm}
\begin{assumption} \label{Assum: Eligible users cannot access express lane without subsidy under DBCP}
Given any $\boldtau \geq 0$ and $(\boldtau, \textbf{0})$-DBCP equilibrium flow $\boldy^\star \in \Y^{\eq}(\G^D(\boldtau, \textbf{0}))$, for each $e \in \edges$ and $t \in [T]$ such that $\tau_{e,t} > 0$, we have $y_{e,1,t}^{g\star} = 0$ for each $g \in G_e^E$.
\end{assumption}
\vspace{-2mm}

\revisionTS{
Assumption \ref{Assum: Eligible users cannot access express lane without subsidy under DBCP} is without loss of generality when eligible users have sufficiently low VoTs, since such users, without receiving subsidy, would be priced out of express lane access at equilibrium (App. \ref{subsec: Low-VoT Users Will Not Voluntarily Pay Express Lane Tolls Out-of-Pocket in Chain Networks}, Prop. \ref{Prop: (DBCP) No Eligible User with Sufficiently Low VoT will Pay Express Lane Tolls Out-of-Pocket}). 
Thus, from a policy design perspective, Assumption \ref{Assum: Eligible users cannot access express lane without subsidy under DBCP} restricts eligibility to low-VoT users who would be priced out of express lane access at equilibrium without being provided subsidy. Assumption \ref{Assum: Eligible users cannot access express lane without subsidy under DBCP} carries practical significance, since a core objective of equitable congestion pricing is to provide tolled lane access to users who would otherwise have been priced out of efficient but financially costly travel options \citep{CBCP-SanMateo, DBCP-ExpressLaneSTART}.
}

\subsubsection{Main Result}
\label{subsubsec: Main Result}


Our main theoretical contribution establishes that under Assumptions \ref{Assum, h: zero sublevel set of h is bounded}-\ref{Assum: Eligible users cannot access express lane without subsidy under DBCP}, if the societal cost is defined to prioritize maximizing toll revenue over minimizing eligible users' travel costs, i.e., \textit{$\lambda_R \geq \lambda_E$}, the minimum societal cost attainable under DBCP policies is no worse than its counterpart under CBCP policies (Thm. \ref{Thm: Main}). 
Remark \ref{Remark: lambda R geq lambda E} explains the practical significance of the condition $\lambda_R \geq \lambda_E$ in the context of congestion pricing on tolled highways.
Moreover, \textit{if $\lambda_R > \lambda_E$}, and a mild regularity assumption on express lane flows (Assumption \ref{Assum: Optimal CBCP Eq Flow admits eligible users}) holds in addition to Assumptions \ref{Assum, h: zero sublevel set of h is bounded}-\ref{Assum: Eligible users cannot access express lane without subsidy under DBCP}, the optimal DBCP policy achieves strictly lower societal cost compared to any CBCP policy (Cor. \ref{Cor: Main Thm, Strictness}). 
We will demonstrate in Sec. \ref{subsec: When CBCP Can Outperform DBCP} that, if the assumption $\lambda_R \geq \lambda_E$ were violated,
the conclusion of Thm. \ref{Thm: Main} may fail to hold, i.e., the optimal CBCP policy may outperform the optimal DBCP policy in minimizing societal costs.

\begin{theorem} \label{Thm: Main}
Under Assumptions \ref{Assum, h: zero sublevel set of h is bounded}-\ref{Assum: Eligible users cannot access express lane without subsidy under DBCP}, if $\lambda_R \geq \lambda_E$, then $f_{D, \boldlambda}^\star \leq f_{C, \boldlambda}^\star$.
\end{theorem}

For a sketch of the proof of Thm. \ref{Thm: Main}, see Sec. \ref{subsubsec: Proof Sketch of Thm, Main}; the full proof is provided in App. \ref{subsec: App, Proof of Thm, Main}.

\begin{remark} \label{Remark: lambda R geq lambda E}
The condition $\lambda_R \geq \lambda_E$ in Thm. \ref{Thm: Main} can be interpreted as follows. Although toll revenues can finance traffic infrastructure improvements and travel subsidies, tolls also burden eligible travelers, whose access to tolled resources hinges upon receiving travel credits or discounts. Thus, the net impact of deploying tolls could be viewed as either positive or negative, depending on whether the system designer prioritizes collecting toll revenue or minimizing eligible users' travel costs. The assumption $\lambda_R \geq \lambda_E$ describes the view that, \textit{provided the design of the constraint function $h$ guarantees eligible users some express lane access without incurring prohibitively high tolls} (Ex. \ref{Ex: h, Eligible User Express Lane Access Guarantees} and \ref{Ex: h, Upper Bounds for Express Lane Tolls}), it is reasonable to prioritize toll revenue collection no less than the minimization of eligible users' costs.
Our emphasis on revenue collection encoded by the assumption $\lambda_R \geq \lambda_E$ carries practical significance, since revenue generation for transit or traffic infrastructure investment is
an explicit objective of real-world congestion pricing policies across many urban areas \citep{NYC-CongestionPricingWebsite, NYC-CongestionPricingRevenue, MTA2026CongestionPricingMakesForBetterTransit, Singapore-CongestionPricingCaseStudy}.

The condition $\lambda_R \geq \lambda_E$ serves a vital role in establishing our main theoretical contribution that DBCP policies outperform CBCP policies in minimizing the societal cost (Thm. \ref{Thm: Main}). We further verify the necessity of assuming $\lambda_R \geq \lambda_E$ for establishing Thm. \ref{Thm: Main} in Sec. \ref{subsec: When CBCP Can Outperform DBCP}, where we demonstrate that when $\lambda_R < \lambda_E$, CBCP policies can become more effective at minimizing the societal cost compared to their DBCP counterparts (Prop. \ref{Prop: Assump lambda R geq lambda E removed}).
\end{remark}

Two key reasons underlie the conclusion of Thm. \ref{Thm: Main} that under Assumptions \ref{Assum, h: zero sublevel set of h is bounded}-\ref{Assum: Eligible users cannot access express lane without subsidy under DBCP}, DBCP policies generate lower societal cost at equilibrium than CBCP policies. 
First, compared to users with lower VoTs, higher-VoT users associate each unit of travel latency with a higher monetary cost, and thus experience a greater reduction in latency-associated costs when provided with express lane access.
Therefore, since DBCP prioritizes allocating express lane access to high-VoT eligible users (Lemma \ref{Lemma: DBCP Equilibria and Effective VoTs}), DBCP is more effective than CBCP at reducing eligible users' total travel latency cost. Second, 
at the same flow $\boldy$, DBCP generates higher toll revenues than CBCP,
since eligible users accessing an express lane under DBCP policies must pay a fraction of the lane toll, while eligible users accessing an express lane under CBCP policies do not. Since toll revenue collection is considered a net societal benefit when $\lambda_R \geq \lambda_E$, DBCP is capable of inducing lower societal costs at equilibrium compared to CBCP.

We now establish an additional mild condition on the cost weights $\boldsymbol{\lambda}$ and the equilibrium flows of the optimal CBCP policy (Assumption \ref{Assum: Optimal CBCP Eq Flow admits eligible users}), which, together with Assumptions \ref{Assum, h: zero sublevel set of h is bounded}-\ref{Assum: Eligible users cannot access express lane without subsidy under DBCP}, guarantee that the optimal DBCP policy \textit{strictly} outperforms the optimal CBCP policy in minimizing the societal cost, i.e., $f_{D, \boldlambda}^\star < f_{C, \boldlambda}^\star$. 

\begin{assumption} \label{Assum: Optimal CBCP Eq Flow admits eligible users}
There exists $(\boldy^\star, \boldx^\star, \boldtau^\star, \boldB^\star) \in S_{C, \boldlambda}^\star$, $e \in \edges$, $t \in [T]$, and $g \in G_e^E$ such that $(\boldy^\star)_{e,1,t}^g > 0$ and $x_{e,1,t}^\star < x_{e,2,t}^\star$.
\end{assumption}

In other words, Assumption \ref{Assum: Optimal CBCP Eq Flow admits eligible users} asserts that there exists an optimal CBCP equilibrium flow $\boldy^\star$ whose express lane flow in some edge $e$ and period $t$ admits some eligible users (i.e., $(\boldy^\star)_{e,1,t}^g > 0$ for some $g \in G_e^E$) while remaining strictly less than the GP lane flow (i.e., $x_{e,1,t}^\star < x_{e,2,t}^\star$).

\begin{corollary} \label{Cor: Main Thm, Strictness}
Under Assumptions \ref{Assum, h: zero sublevel set of h is bounded}-\ref{Assum: Optimal CBCP Eq Flow admits eligible users}, if $\lambda_R > \lambda_E$ and $S_{C, \boldlambda}, S_{D, \boldlambda} \ne \emptyset$, then $f_{D, \boldlambda}^\star < f_{C, \boldlambda}^\star$.

\end{corollary}

The proof of Cor. \ref{Cor: Main Thm, Strictness}  follows by retracing the proof of Thm. \ref{Thm: Main}, and applying the two additional constraints in the statement of Cor. \ref{Cor: Main Thm, Strictness}. Concretely, we demonstrate that under the conditions listed 
in the statement of Cor. \ref{Cor: Main Thm, Strictness}, given any optimal CBCP policy satisfying $(\boldy^\star)_{e,1,t}^g > 0$ and $x_{e,1,t}^\star < x_{e,2,t}^\star$, the proof of Thm. \ref{Thm: Main} allows us to construct a DBCP policy that generates \textit{strictly} higher revenue at equilibrium, which translates into strictly lower societal cost, since $\lambda_R > \lambda_E$.

For the full proof of Cor. \ref{Cor: Main Thm, Strictness}, see App. \ref{subsec: App, Proof of Cor, Main Thm, Strictness}.

\subsubsection{Proof Sketch of Thm. \ref{Thm: Main}}
\label{subsubsec: Proof Sketch of Thm, Main}

\looseness=-1
The central challenge of proving Thm. \ref{Thm: Main} is that 
under CBCP policies, the equilibrium routing decisions of eligible users are coupled across edges $e$ and periods $t$, which poses difficulties for the comparison of CBCP equilibrium flows against their DBCP counterparts.
To overcome this challenge, 
we demonstrate that under Assumptions \ref{Assum, h: zero sublevel set of h is bounded}-\ref{Assum: Eligible users cannot access express lane without subsidy under DBCP}, unlike CBCP equilibrium flows, \textit{DBCP equilibrium flows} admit structural properties that facilitate the comparison of the effectiveness of CBCP and DBCP policies at minimizing societal cost.
Specifically, we
prove that in the setting of a chain network (i.e., under Assumption \ref{Assum: Chain Network}), the computation of DBCP equilibria can be decoupled across different edges $e \in \edges$ and periods $t \in [T]$ (Lemma \ref{Lemma: DBCP Eq Decomposition, Chain Network}). We then prove that, on each edge $e \in \edges$ of a chain network on which a DBCP policy has been deployed, the emergent equilibrium flow at period $t \in [T]$ assigns express lane access to groups via an order that depends upon the groups' VoTs and the discount offered
(Lemma \ref{Lemma: DBCP Equilibria and Effective VoTs}). The modular structure and flow assignment order of DBCP equilibria for chain networks, characterized in Lemmas \ref{Lemma: DBCP Eq Decomposition, Chain Network}-\ref{Lemma: DBCP Equilibria and Effective VoTs}, facilitate the edge-by-edge construction of $(\boldtau, \boldalpha)$-DBCP policies utilized in the proof of Lemma \ref{Lemma: Main, Explicit Construction of DBCP Policy}, which 
is crucial for establishing Thm. \ref{Thm: Main}.

\begin{lemma} \label{Lemma: DBCP Eq Decomposition, Chain Network}
Under Assumption \ref{Assum: Chain Network}, given any $(\boldtau, \boldalpha)$-DBCP policy,
the flow $\boldy \in \Y(\G^D(\boldtau, \boldalpha))$ is a $(\boldtau, \boldalpha)$-DBCP equilibrium flow if and only if for each edge $e \in \edges$ and period $t \in [T]$, the flow $y_{e,t} := (y_{e,k,t}^g: k \in [2], g \in G_e)$ is an equilibrium flow for the corresponding $(\tau_{e,t}, \alpha_{e,t})$-DBCP game defined over the 
groups $G_e$.
\end{lemma}

\begin{proof}
By Prop. \ref{Prop: Convex Program, DBCP}, $\boldy^\star \in \Y(\G^D(\boldtau, \boldalpha))$ is a $(\boldtau, \boldalpha)$-DBCP equilibrium flow if and only if it minimizes the following objective, first presented in \eqref{Eqn: Convex Program Statement, DBCP}, over the constraint set $\Y(\G^D(\boldtau, \boldalpha))$:
{
\setlength{\abovedisplayskip}{4pt}
\setlength{\belowdisplayskip}{4pt}
\begin{align} \label{Eqn: Convex Program Objective, DBCP}
    \sum_{t \in [T]} \sum_{e \in \edges} \left[ \sum_{k=1}^2 \int_0^{x_{e,k,t}} \ell_{e,k}(w) \hspace{0.5mm} dw  + \sum_{g \in \G_e^I} \frac{y_{e,1,t}^g \tau_{e,t}}{v_t^g} + \sum_{g \in \G_e^E} \frac{y_{e,1,t}^g (1 - \alpha_{e,t}) \tau_{e,t}}{v_t^g} \right].
\end{align}
}
Thus, to establish Lemma \ref{Lemma: DBCP Eq Decomposition, Chain Network}, it suffices to observe that when $\network$ is a chain network (i.e., under Assumption \ref{Assum: Chain Network}), the following conditions hold: (1) The objective \eqref{Eqn: Convex Program Objective, DBCP} can be expressed as the sum of individual cost terms, each of which only depends upon flow variables identified with a particular edge $e \in \edges$ and period $t \in [T]$; (2) The constraint set $\Y(\G^D(\boldtau, \boldalpha))$ likewise decouples across edges $e \in \edges$ and times $t \in [T]$.
This completes the proof of Lemma \ref{Lemma: DBCP Eq Decomposition, Chain Network}.
\end{proof}

Next, we show that DBCP equilibrium flows on a chain network can be characterized using the notion of the \textit{effective VoT}, which we define below.

\begin{definition}(\textbf{Effective VoTs under a DBCP policy}) 
\label{Def: Effective VoTs}
Given a $(\boldtau, \boldalpha)$-DBCP policy, we define the corresponding effective VoT, on each edge $e \in \edges$ at each period $t \in [T]$, for group $g \in G_e$ by $\hat v_{e,t}^g = \frac{v_t^g}{1 - \alpha_{e,t}}$ for each $g \in G_e^E$, and $\hat v_{e,t}^g =v_t^g$ for each $g \in G_e^I$.
\end{definition}

In other words, the VoT and effective VoT of each ineligible group $g \in G^I$ are identical, while the effective VoT of an eligible group $g \in G_e$ on edge $e \in \edges$ at period $t \in [T]$ equals its VoT scaled by $\frac{1}{1 - \alpha_{e,t}}$. To motivate the definition of the effective VoT, we recall that given a $(\boldtau, \boldalpha)$-DBCP policy, an admissible flow $\boldy \in \Y(\G^D(\boldtau, \boldalpha))$ is a $(\boldtau, \boldalpha)$-DBCP equilibrium flow if and only if $\boldy$ is a minimizer of the convex program \eqref{Eqn: Convex Program Statement, DBCP}. Using the definition of the effective VoT, we can rewrite the objective function in the convex program \eqref{Eqn: Convex Program Statement, DBCP} as follows:
{
\setlength{\abovedisplayskip}{4pt}
\setlength{\belowdisplayskip}{4pt}
\begin{align} \label{Eqn: Convex Program for DBCP, using effective VoT}
    &\sum_{t \in [T]} \sum_{e \in \edges} \left[ \sum_{k=1}^2 \int_0^{x_{e,k,t}} \ell_{e,k}(w) \hspace{0.5mm} dw  + \sum_{g \in \G_e} \frac{y_{e,1,t}^g \tau_{e,t}}{\hat v_{e,t}^g} \right],
\end{align}
}
\hspace{-1mm}where $\hat v_{e,t}^g$ is as defined in Def. \ref{Def: Effective VoTs}. In words, providing eligible users with discounts $\boldalpha := (\alpha_{e,t} \in [0, 1]: e \in \edges, t \in [T])$ achieves the same effect, on eligible users' equilibrium behavior, as the effect of artificially inflating each eligible user's VoT by the factor $\frac{1}{1 - \alpha_{e,t}}$ on each edge $e$ and period $t$. Loosely speaking, on each edge $e \in \edges$ at each period $t \in [T]$, each eligible user from group $g \in G_e$ \textit{effectively} takes action as though they face the same toll price $\tau_{e,t}$ as other ineligible users,
but have an inflated value of time given by $\frac{1}{1-\alpha_{e,t}} v_t^g$. We note that, although the VoT of each group does not vary across edges, the effective VoT of eligible groups under a DBCP policy may vary across edges on which different discount values are provided, since steeper discounts would increase eligible users' willingness to pay for express lane access.

We now present Lemma \ref{Lemma: DBCP Equilibria and Effective VoTs}, which establishes that a $(\boldtau, \boldalpha)$-DBCP equilibrium flow $\boldy$ must assign users, across groups $g \in G_e$ on each edge $e \in \edges$ and period $t \in [T]$, to the express lane in decreasing order of their effective VoT 
with respect to the corresponding $(\boldtau, \boldalpha)$-DBCP policy.

\begin{lemma} \label{Lemma: DBCP Equilibria and Effective VoTs}
Suppose Assumption \ref{Assum: Chain Network} holds. Let a $(\boldtau, \boldalpha)$-DBCP policy be given, where $\boldtau \in \R_{\geq 0}^{|\edges|T}$ and $\boldalpha \in [0, 1]^{|\edges|T}$, and let 
$\boldy$
be any $(\boldtau, \boldalpha)$-DBCP equilibrium flow. Then on each edge $e \in \edges$ and $t \in [T]$, $\boldy$ assigns users to the express lane in decreasing order of their effective VoT.
\end{lemma}

\begin{proof}{Proof}
By Lemma \ref{Lemma: DBCP Eq Decomposition, Chain Network}, the computation of DBCP equilibrium flows over chain networks is decomposable across edges and periods. Thus, to prove Lemma \ref{Lemma: DBCP Equilibria and Effective VoTs}, it suffices to show that for any $(\tau_{e,t}, \alpha_{e,t})$-DBCP equilibrium flow $y_{e,t} := (y_{e,k,t}^g: g \in G_e, k \in [2])$ corresponding to the $(\tau_{e,t}, \alpha_{e,t})$-DBCP game on edge $e \in \edges$ at period $t \in [T]$, the flow $\boldy$ assigns users in groups $g \in G_e$ to the express lane in descending order of effective VoT. 
To establish this condition, observe that \eqref{Eqn: Convex Program for DBCP, using effective VoT} is precisely the convex program whose minimizers are Nash equilibrium flows under the setting in which all users are ineligible, and have VoTs replaced by effective VoTs across edges and periods. 
We then apply the fact that, without subsidy, users access tolled lanes in descending order of VoT \citep[Prop. 2.4]{Cole2003PricingNetworkEdgesforHeterogeneousSelfishUsers}.
\hfill $\square$


\end{proof}


Finally, we
establish Lemma \ref{Lemma: Main, Explicit Construction of DBCP Policy}, which states that given any $(\boldtau, \boldB)$-CBCP policy and corresponding equilibrium flow $\boldy$, we can construct a $(\boldtau, \boldalpha)$-DBCP policy and corresponding equilibrium flow $\hat \boldy$ with the following two properties across all edges $e$ and periods $t$: (1) The \textit{total eligible user} express lane flow assigned by $\boldy$ and $\hat \boldy$ are equal, and (2) the express lane flow \textit{for each ineligible group} assigned by $\boldy$ and $\hat \boldy$ are equal. Recall that DBCP equilibria on chain networks assign express lane access to eligible users in order of decreasing VoT (Lemma \ref{Lemma: DBCP Equilibria and Effective VoTs}). Lemma \ref{Lemma: Main, Explicit Construction of DBCP Policy} thus implies that across all edges and periods, $\hat \boldy$ allocates 
express lane access to ineligible users using the same assignment prescribed by $\boldy$, while allocating express lane access to eligible users in decreasing order of VoT until the total eligible user flow on the express lane matches the same amount prescribed by $\boldy$.
By definition of VoT, the same reduction in travel time translates to greater monetary savings for eligible users with higher VoTs compared to those with lower VoTs. As a result, Lemma \ref{Lemma: Main, Explicit Construction of DBCP Policy} implies that DBCP policies are \textit{more effective at lowering eligible users' total travel latency-associated costs} compared to CBCP policies,
an insight that will provide a critical contribution to our proof of Thm. \ref{Thm: Main}.
\revisionTRPartB{Importantly, neither Lemma \ref{Lemma: Main, Explicit Construction of DBCP Policy} nor Thm. \ref{Thm: Main} requires CBCP or DBCP equilibria to be unique.
In particular, the construction summarized above and detailed in the proof of Lemma \ref{Lemma: Main, Explicit Construction of DBCP Policy} applies to any CBCP equilibrium flow and establishes the existence of a corresponding DBCP equilibrium flow satisfying \eqref{Eqn: Main Lemma, matching total eligible user express lane flow}-\eqref{Eqn: Main Lemma, matching ineligible user express lane flow}.}

\begin{lemma} \label{Lemma: Main, Explicit Construction of DBCP Policy}
Suppose Assumptions \ref{Assum, h: zero sublevel set of h is bounded}-\ref{Assum: Eligible users cannot access express lane without subsidy under DBCP} hold. For any $\boldtau, \boldB \geq 0$, given a $(\boldtau, \boldB)$-CBCP equilibrium flow $\boldy \in \Y$ and corresponding lane flows $\boldx \in \R_{\geq 0}^{2|\edges|T}$ for which $(\boldy, \boldx, \boldtau, \boldB) \in S_{C, \boldlambda}^\star$, there exists $\boldalpha \in [0, 1]^{|\edges|T}$ and a $(\boldtau, \boldalpha)$-DBCP equilibrium flow $\hat \boldy$
such that $(\hat \boldy, \boldx, \boldtau, \boldalpha) \in S_{D, \boldlambda}^\star$ and:
{
\setlength{\abovedisplayskip}{4pt}
\setlength{\belowdisplayskip}{4pt}
\begin{align} \label{Eqn: Main Lemma, matching total eligible user express lane flow}
    \sum_{g \in G_e^E} \hat y_{e,1,t}^g &= \sum_{g \in G_e^E} y_{e,1,t}^g, \hspace{5mm} \forall \ e \in \edges, t \in [T], \\
\label{Eqn: Main Lemma, matching ineligible user express lane flow}
    \hat y_{e,1,t}^g &= y_{e,1,t}^g, \hspace{5mm} \forall \ e \in \edges, t \in [T], \ g \in G_e^I.
\end{align}
}
\end{lemma}

\begin{proof}(\textbf{Proof Sketch})
By Lemmas \ref{Lemma: DBCP Eq Decomposition, Chain Network} and \ref{Lemma: DBCP Equilibria and Effective VoTs}, the DBCP equilibrium flows on a chain network assign eligible users to express lanes in descending order of VoT, and can be computed in a manner that is decomposable across edges and periods. Thus, we arbitrarily select an edge $e \in \edges$ and $t \in [T]$, and judiciously select a discount level $\alpha_{e,t} \in [0, 1]$ such that \eqref{Eqn: Main Lemma, matching total eligible user express lane flow} and \eqref{Eqn: Main Lemma, matching ineligible user express lane flow} hold. Specifically, we choose $\alpha_{e,t} \in [0, 1]$ such that the $(\tau_{e,t}, \alpha_{e,t})$-DBCP policy admits an equilibrium flow $\hat \boldy$ which assigns eligible users from groups $G_e^E$ to the express lane in descending order of VoT until the total eligible users' express lane flow equals the level prescribed by $\boldy$, i.e., $\sum_{g \in G_e^E} \hat y_{e,1,t}^g = \sum_{g \in G_e^E} y_{e,1,t}^g$, while preserving the same assignment of ineligible users to the express lane as the assignment provided by $\boldy$, i.e., $\hat y_{e,1,t}^g = y_{e,1,t}^g$ for each $g \in G_e^I$.
\hfill $\square$
\end{proof}

For the full proof of Lemma \ref{Lemma: Main, Explicit Construction of DBCP Policy}, see App. \ref{subsec: App, Proof of Lemma, Main, Explicit Construction of DBCP Policy}.

Armed with Lemma \ref{Lemma: Main, Explicit Construction of DBCP Policy}, we now outline the proof of Thm. \ref{Thm: Main}. 


\begin{proof}{Proof Sketch for Thm. \ref{Thm: Main}}
First, we use Lemma \ref{Lemma: Main, Explicit Construction of DBCP Policy} to establish that compared to CBCP policies, DBCP policies allow eligible users to attain lower travel latency-associated costs on each edge $e$ at each period $t$. Moreover, we establish that compared to CBCP policies, DBCP policies generate equivalent or higher toll revenue by compelling eligible users to pay a fraction of the express lane tolls.
When $\lambda_R \geq \lambda_E$, higher toll revenue is considered a net societal benefit, and thus translates to lower societal cost. We thus conclude that, under Assumptions \ref{Assum, h: zero sublevel set of h is bounded}-\ref{Assum: Eligible users cannot access express lane without subsidy under DBCP}, if $\lambda_R \geq \lambda_E$, DBCP policies outperform CBCP policies at minimizing the societal cost, i.e., $f_{D, \boldlambda}^\star \leq f_{C, \boldlambda}^\star$.
\hfill $\square$
\end{proof}

For the full proof of Thm. \ref{Thm: Main}, see App. \ref{subsec: App, Proof of Thm, Main}.

In addition to directions of future research presented in Sec. \ref{subsec: Discussion and Future Research Directions}, it would be worthwhile to explore whether Thm. \ref{Thm: Main} holds under model extensions 
in which the traffic network has arbitrary structure, or in which the user population has a continuous VoT distribution.

\subsection{When CBCP Can Outperform DBCP in Societal Cost Minimization}
\label{subsec: When CBCP Can Outperform DBCP}

Thm. \ref{Thm: Main} and Cor. \ref{Cor: Main Thm, Strictness} state that DBCP outperforms CBCP in minimizing the equilibrium value of a societal cost whose weights $\boldlambda$ satisfy $\lambda_R \geq \lambda_E$, provided that the traffic network and user population satisfy Assumptions \ref{Assum, h: zero sublevel set of h is bounded}-\ref{Assum: Eligible users cannot access express lane without subsidy under DBCP}.
We now show that if the condition $\lambda_R \geq \lambda_E$ were violated,
i.e., if $\lambda_R < \lambda_E$,
the conclusions of Thm. \ref{Thm: Main} may not hold, i.e., the optimal CBCP policy may outperform the optimal DBCP policy in minimizing societal costs. 

\begin{proposition} \label{Prop: Assump lambda R geq lambda E removed}
Consider a single-edge network $\network = (E, I)$, with $\edges = \{e\}$ and $I = \{i_e, j_e\}$, 
serving
ineligible groups $G_e^I$ and a single eligible group $g^E$. 
Suppose $T = 1$, $\lambda_R < \lambda_E$,
$S_{C, \boldlambda}, S_{D, \boldlambda} \ne \emptyset$,
and 
Assumptions \ref{Assum, h: zero sublevel set of h is bounded}-\ref{Assum: Eligible users cannot access express lane without subsidy under DBCP}
hold. Moreover, suppose the constraint $h(\cdot) \leq 0$ in \eqref{Eqn: f C lambda star, Def} and \eqref{Eqn: f D lambda star, Def} enforces that, for some $m_y^I$, $M_y^I$, $m_y^E > 0$, $M_y^E \in (0, d^{g^E})$ and $M_\boldx \in \left[0, \frac{1}{2}d_e \right)$:
{
\setlength{\abovedisplayskip}{3pt}
\setlength{\belowdisplayskip}{1pt}
\begin{subequations}
\begin{align} \label{Eqn: Counterexample 1, Eligible y bound}
    m_y^E &\leq y_{e,1,1}^{g^E} \leq M_y^E, \\ \label{Eqn: Counterexample 1, Ineligible y bound}
    m_y^I &\leq \sum_{g \in G_e^I} y_{e,1,1}^g \leq M_y^I, \\ \label{Eqn: Counterexample 1, x bound}
    x_{e,1,1} &\leq M_x.
\end{align}
\end{subequations}
}
Then there exists some $M_{\boldtau} > 0$ and $m_\alpha \in (0, 1)$ such that:
{
\setlength{\abovedisplayskip}{2pt}
\setlength{\belowdisplayskip}{2pt}
\begin{align} \label{Eqn: Counterexample 1, optimal f C bound}
    f_{C, \boldlambda}^\star &\leq f_{D, \boldlambda}^\star - (\lambda_E - \lambda_R) (1 - m_\alpha) M_{\boldtau} m_y^E < f_{D, \boldlambda}^\star.
\end{align}
}
\end{proposition}

In words, Prop. \ref{Prop: Assump lambda R geq lambda E removed} states that if $\lambda_R < \lambda_E$, then even in the single-edge ($\edges = \{e\}$) networks and single-period ($T = 1$) setting, there exist mild conditions under which CBCP policies would provably outperform DBCP policies at minimizing the societal cost. Indeed, as described in Sec. \ref{subsec: When DBCP Outperforms CBCP}, the superior performance of DBCP policies in minimizing the societal cost compared to CBCP policies when $\lambda_R \geq \lambda_E$ is contingent upon the following two facts: (1) DBCP assigns eligible users to express lanes in order of decreasing VoT, which lowers the total travel latency-related cost incurred by eligible users, and (2) DBCP generates more revenue from toll collection, which results in a lower societal cost when $\lambda_R \geq \lambda_E$. 
Prop. \ref{Prop: Assump lambda R geq lambda E removed} formalizes the intuition that when the two facts above fail to hold (specifically, when eligible users have the same VoT and $\lambda_R < \lambda_E$), the optimal CBCP policy can attain a lower societal cost than the optimal DBCP policy.

The proof of Prop. \ref{Prop: Assump lambda R geq lambda E removed} follows by reversing the DBCP policy construction process that underlies the proofs of Lemma \ref{Lemma: Main, Explicit Construction of DBCP Policy}, Thm. \ref{Thm: Main}, and Cor. \ref{Cor: Main Thm, Strictness}. Rather than construct DBCP policies from arbitrarily fixed CBCP policies, we instead show that for any DBCP policy, we can construct a CBCP policy that attains a strictly lower societal cost at equilibrium.

For the full proof of Prop. \ref{Prop: Assump lambda R geq lambda E removed}, see App. \ref{subsec: App, Proof of Thm, Assump lambda R geq lambda E removed}.


\section{Numerical Experiments}
\label{sec: Numerical Experiments}

In this section, we perform numerical experiments to compare the effects of deploying CBCP and DBCP based on a real-world case study of the 101 Express Lanes Project in the San Francisco Bay Area \citep{CBCP-SanMateo}.
First, in Sec. \ref{subsec: Network Formulation and Calibration}, we 
construct our network model, estimate user demand and VoT distributions, calibrate latency functions across network edges, and
\revisionTRPartB{present optimization algorithms for computing first-order stationary CBCP and DBCP policies}
as well as corresponding equilibrium flows.
Our numerical results, which we analyze in Sec. \ref{subsec: Equilibrium Comparisons Between Stationary CBCP and DBCP Policies}, validate our theoretical results in Sec. \ref{subsec: When DBCP Outperforms CBCP} regarding the relative benefits of DBCP over CBCP policies, and provide substantive insight into the relative impact of CBCP and DBCP policies on societal welfare.
%
We conclude in Sec. \ref{subsec: Discussion and Future Research Directions} by discussing the policy implications of our theoretical and numerical results, and presenting avenues for future research.
All code associated with this paper has been attached with the submission.
\hspace{-6mm}

\subsection{Network Formulation and Calibration}
\label{subsec: Network Formulation and Calibration}

We construct a traffic network $\network_\ELP = (V_\ELP, E_\ELP)$, containing a set of nodes $V_\ELP$ and a set of consecutive edges $E_\ELP$, 
wherein each edge describes traffic on the 101-N freeway segment passing through one of the cities
in $E_\ELP = \{$\text{Palo Alto}, \text{East Palo Alto}, \text{Redwood City}, \text{Belmont}, \text{San Mateo}, \text{Burlingame}, \text{Millbrae}$\}$ (see
Fig. \ref{fig: Chain Network}).
Each edge $e \in \edges_\ELP$ represents a two-lane Pigou network, whose first and second lanes respectively describe a tolled express lane and a toll-free GP lane. As noted in Remark \ref{Remark: Latency Function Model for Multiple GP Lanes}, our framework can be generalized to capture the fact that the real-world 101-N freeway 
contains
multiple GP lanes (e.g., $n_{GP} = 3$).

We then calibrate latency functions for each lane across $\network_\ELP$
using weekday vehicle flow and travel speed data collected from the Caltrans' Performance Measurement System (PeMS) database \citep{pems-database}. 
\revisionTRPartB{We also estimate user demand and VoT distributions across origin-destination pairs using vehicle flow data collected from the PeMS database and 2020 US Census American Community Survey (ACS) data. Following prior work \citep{Jalota2022CreditBasedCongestionPricing, Athira2016EstimationofValueofTravelTimeforWorkTrips, Tveter2022TheValueOfTravelTime}, we assume that users' VoT distributions are proportional to the income distribution in their city of origin, although alternative VoT estimation methods exist. For further details on our network calibration and parameter estimation procedures, see App. \ref{subsubsec: App, Latency Calibration}-\ref{subsubsec: App, User Value-of-Time (VoT)}.}
We then use the inferred VoT and demand values to formulate cost and constraint functions for the optimization programs \eqref{Eqn: f C lambda star, Def} and \eqref{Eqn: f D lambda star, Def}, which characterize optimal CBCP and DBCP policies.
Specifically, we solve \eqref{Eqn: f C lambda star, Def} and \eqref{Eqn: f D lambda star, Def} with the societal cost parameterized by Pareto weights $\boldlambda$ in the set
$S_{\boldlambda}^\paren{1} \cup S_{\boldlambda}^\paren{2}$ 
considered in App. \ref{sec: App, Data Table for Numerical Experiments}, 
Table \ref{Table: (Avg) All equilibrium results, for all lambda}.


To formulate the constraint $h(\cdot)$ in \eqref{Eqn: f C lambda star, Def, Constraint of h} and \eqref{Eqn: f D lambda star, Def, Constraint of h}, we first fix an upper bound $C_\tau = \$5$ for the toll $\tau_{e,t}$ deployed on each edge $e \in \edges$ at each period $t \in [T]$.
We note that the specific value of $C_\tau$ does not affect our numerical results, provided that $C_\tau$ is reasonably large.
We then define $h(\cdot)$ to constrain all $\tau_{e,t}$ across edges $e \in \edges$ and times $t \in [T]$ to be between $\$0$ and $\$5$. 

We note that \eqref{Eqn: f C lambda star, Def} and \eqref{Eqn: f D lambda star, Def}, which characterize optimal CBCP and DBCP policies and their corresponding equilibria, are in general non-convex bilevel optimization problems whose global minimizers are computationally intractable to compute.
To circumvent
these challenges,
we formulate and execute a refinement of the iterative, zeroth-order optimization algorithm presented in \citet{Maheshwari2024FollowerAgnosticLearninginStackelbergGames} (see App. \ref{subsec: App, Algorithm Details for Computing First-Order Stationary DBCP and CBCP Policies}, Alg. \ref{Alg: Opt CBCP, Zeroth-order Optimization} and \ref{Alg: Opt DBCP, Zeroth-order Optimization}). 
Although zeroth-order algorithms are, in theory, only guaranteed to converge with high probability to first-order stationary points, 
such methods have in practice recovered locally optimal solutions 
for bilevel optimization problems across a wide range of engineering contexts \citep{Maheshwari2024FollowerAgnosticLearninginStackelbergGames, Zhang2024AnIntroductiontoBilevelOptimization}.
Accordingly,
our empirical results should be interpreted as comparisons between first-order stationary DBCP and CBCP policies computed via analogous optimization methods, rather than as exact comparisons between globally optimal DBCP and CBCP policies.


\subsection{Equilibrium Comparisons 
Between Stationary CBCP and DBCP Policies}
\label{subsec: Equilibrium Comparisons Between Stationary CBCP and DBCP Policies}

To validate our theoretical results in Sec. \ref{sec: CBCP and DBCP Comparison Study for the Chain Network Setting}, we use Algs. \ref{Alg: Opt CBCP, Zeroth-order Optimization}-\ref{Alg: Opt DBCP, Zeroth-order Optimization} to compute toll, budget, and discount values for stationary DBCP and CBCP policies, as well as the corresponding equilibrium user flows $\boldy$ and lane flows $\boldx$. 
For each stationary DBCP and CBCP policy, we contrast their efficiency and equity impact by comparing the corresponding equilibrium total societal cost, toll revenues, user costs, travel times, and express lane usage (see Fig. \ref{fig: Figures___Metrics_Comparison_CBCP_vs_DBCP} and 
Table \ref{Table: (Avg) All equilibrium results, for all lambda}).

\begin{figure}[ht]
    \centering
    \includegraphics[width=0.99\linewidth]{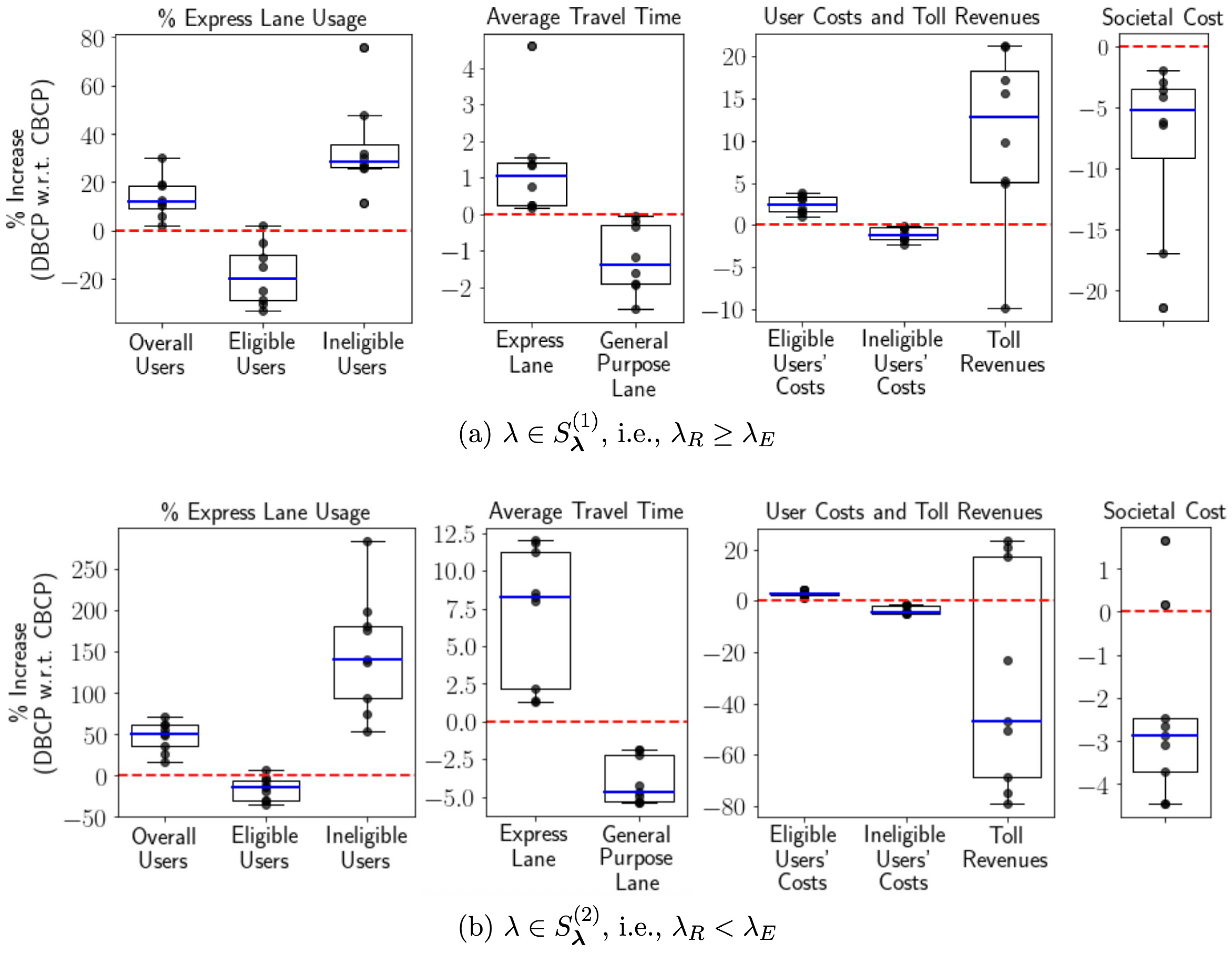}
    \caption{\small \sf Percent increase in the following metrics
    at equilibrium under the stationary DBCP policy, compared to the stationary CBCP policy, for (a) each $\boldlambda \in S_{\boldlambda}^\paren{1}$, and (b) each $\boldlambda \in S_{\boldlambda}^\paren{2}$: Overall, eligible, and ineligible users' express lane usage, average travel time on the express and GP lanes over $T = 5$ weekdays, eligible/ineligible users' costs, toll revenues, and societal costs.
    }
    \vspace{-3mm}
    \label{fig: Figures___Metrics_Comparison_CBCP_vs_DBCP}
\end{figure}

\subsubsection{Societal Cost, Toll Revenues, and User Travel Costs}
\label{subsubsec: Societal Cost, Toll Revenues, and User Travel Costs}

Our numerical results in 
Table \ref{Table: (Avg) All equilibrium results, for all lambda}
validate Thm. \ref{Thm: Main} by illustrating that, 
when
$\lambda_R \geq \lambda_E$, the stationary DBCP policy generates a lower societal cost compared to the stationary CBCP policy, primarily by generating higher toll revenues and lower ineligible user costs while not significantly impacting eligible user costs. 
For instance, for $\boldlambda = (1, 5, 1)$, compared to the stationary CBCP policy,
the stationary DBCP policy
generated equilibrium flows at which the toll revenues rose by $5.3\%$, ineligible users' travel costs decreased by $1.7\%$, and eligible users' travel costs increased by $1.5\%$,
resulting in a net societal cost decrease of $3.0\%$.
Our computed first-order stationary DBCP policies often yield higher toll revenues and lower ineligible user costs compared to their CBCP counterparts, since they often 
provide more express lane access to ineligible users
at equilibrium (see Sec. \ref{subsubsec: Express Lane Usage}).

Moreover, when $\lambda_R \geq \lambda_E$, increases in $\lambda_R$ and $\lambda_E$ relative to $\lambda_I$ amplify the extent to which the stationary DBCP policy outperforms its CBCP counterpart in attaining lower equilibrium societal costs, primarily by generating a greater rise in toll revenues and assigning higher priority to toll revenue generation.
For instance, whereas the stationary DBCP policy for $\boldlambda = (1, 5, 1)$ produced a $5.3\%$ increase in toll revenues and $3.0\%$ reduction in societal cost compared to its CBCP counterpart, the stationary DBCP policy for $\boldlambda = (10, 10, 1)$ produced a $21.1\%$ rise in toll revenues and a $6.5\%$ reduction in societal cost compared to its CBCP counterpart.

\looseness=-1
Next, 
when $\lambda_R < \lambda_E$, the stationary DBCP policy \textit{can but does not always} produce a lower equilibrium societal cost compared to the stationary CBCP policy (see Table \ref{Table: (Avg) All equilibrium results, for all lambda}).
In particular, when $\lambda_E$ slightly exceeds $\lambda_R$, compared to the stationary CBCP policy, the stationary DBCP policy may generate sufficiently higher toll revenues (e.g., when $\boldlambda = (5, 1, 0)$) or lower ineligible user costs (e.g., when $\boldlambda = (5, 1, 1)$) at equilibrium to yield a net decrease in societal cost.
However, when $\lambda_E$ far exceeds $\lambda_R$, stationary CBCP policies may outperform their DBCP counterpart in attaining a lower equilibrium societal cost (see Prop. \ref{Prop: Assump lambda R geq lambda E removed}), primarily by lowering eligible user costs.
For example, at $\boldlambda = (20, 1, 0)$, the stationary DBCP policy  generated a $2.7\%$ increase in eligible user costs at equilibrium compared to its CBCP counterpart, which contributed to a $1.7\%$ rise in societal cost, despite lower ineligible user costs and higher toll revenues.

\subsubsection{Express Lane Usage}
\label{subsubsec: Express Lane Usage}

We observe from 
Table \ref{Table: (Avg) All equilibrium results, for all lambda}
and Fig. \ref{fig: Figures___Metrics_Comparison_CBCP_vs_DBCP} that across all $\boldlambda \in S_{\boldlambda}^\paren{1} \cup S_{\boldlambda}^\paren{2}$, the stationary DBCP policy for $\boldlambda$ provided express lane access at equilibrium to a \textit{higher proportion of both overall users} (to generate higher toll revenues) \textit{and ineligible users} (to reduce ineligible user costs) compared to its CBCP counterpart.
For instance, under the stationary DBCP (resp., CBCP) policy for $\boldlambda = (5, 5, 1) \in S_{\boldlambda}^\paren{1}$, 
roughly $19.6\%$ (resp., $17.6\%$) of overall users and $15.8\%$ (resp., $10.7\%$) of ineligible users accessed the express lane at equilibrium. 

Table \ref{Table: (Avg) All equilibrium results, for all lambda}
also indicates that,
compared to their CBCP counterparts,
the stationary DBCP policies for $\boldlambda \in S_{\boldlambda}^\paren{1} \cup S_{\boldlambda}^\paren{2}$ \textit{often but not always} reallocated a small portion of eligible users' express lane access to ineligible users, in order to generate higher toll revenues or reduce ineligible user costs.
For instance, 
the stationary DBCP and CBCP policies for $\boldlambda = (5, 5, 1) \in S_{\boldlambda}^\paren{1}$ respectively enabled $39\%$ and $52\%$ of eligible users to access the tolled express lane at equilibrium. 
However, when $\lambda_E$ exceeds $\lambda_I$ and $\lambda_R$ substantially, e.g., when $\boldlambda = (20, 1, 0) \in S_{\boldlambda}^\paren{2}$, the stationary DBCP policy provided eligible users a higher express lane access rate ($49\%$) at equilibrium compared to the stationary CBCP policy ($40\%$).

\subsubsection{User Travel Times}
\label{subsubsec: User Travel Times}

From 
Table \ref{Table: (Avg) All equilibrium results, for all lambda}
and Fig. \ref{fig: Figures___Metrics_Comparison_CBCP_vs_DBCP}, we observe that, compared to their CBCP counterparts, stationary DBCP policies across weights $\boldlambda \in S_{\boldlambda}^\paren{1} \cup S_{\boldlambda}^\paren{2}$ induce a longer (resp., shorter) average express (resp., GP) lane travel time at equilibrium. Thus, compared to stationary CBCP policies, stationary DBCP policies offer reduced travel time savings on the express lane at equilibrium, in exchange for providing 
express lane access to a higher fraction of overall users.
For instance, the stationary DBCP (resp., CBCP) policy for $\boldlambda = (5, 5, 1) \in S_{\boldlambda}^\paren{1}$ induced average travel time savings of $2.9$ min (resp., $3.7$ min) on the express lane at equilibrium. Similarly, the stationary DBCP (resp., CBCP) policy for $\boldlambda = (10, 1, 1) \in S_{\boldlambda}^\paren{2}$ induced average travel time savings of $0.8$ min (resp., $3.8$ min) on the express lane at equilibrium.


\subsection{Discussion and Future Research Directions}
\label{subsec: Discussion and Future Research Directions}

Our numerical analysis in Sec. \ref{subsec: Equilibrium Comparisons Between Stationary CBCP and DBCP Policies} has several practical implications for the design of real-world DBCP and CBCP policies. 
First, given a Pareto weight $\boldlambda$, the stationary DBCP policy typically generates lower societal costs at equilibrium compared to its CBCP counterpart, unless $\lambda_E$ far exceeds $\lambda_R$ (i.e., unless eligible users' welfare is heavily prioritized). 
Thus, our numerical results highlight that if toll revenue generation is favored over the minimization of eligible users’ costs, it would be advisable to implement DBCP rather than CBCP. In contrast, policy makers who prioritize reducing eligible users’ travel costs over toll revenue maximization should carefully evaluate the relative merits and demerits of deploying DBCP and CBCP policies, before finalizing their decision to deploy one policy type or the other.

It is also worthwhile to study whether \textit{hybrid} DBCP-CBCP policies, which provide both travel credits and toll discounts, can more effectively minimize societal cost than either policy alone.
As analyzed in Sec. \ref{subsec: When DBCP Outperforms CBCP}, DBCP policies can outperform CBCP policies in maximizing toll revenue,
while
CBCP policies 
provide
each eligible user a degree of express lane access proportional to their allotted budget, regardless of VoT.
Thus, compared to pure DBCP policies, hybrid 
policies have the potential to reduce eligible users' costs and assign express lane access to eligible user groups more equitably,
while retaining effectiveness in
generating toll revenue.

Moreover,
CBCP induces the highest equilibrium travel cost among lowest-VoT ineligible users consigned to the GP lane, since their VoTs are neither high enough to warrant paying express lane toll fees out-of-pocket, nor low enough to justify receiving subsidies for express lane access. Similarly, our theoretical analysis and numerical results indicate that, under DBCP, ineligible users with the lowest VoT incur high travel costs. 
Moreover, under DBCP, since eligible users are assigned to the express lane in order of decreasing VoT at equilibrium (Lemma \ref{Lemma: DBCP Equilibria and Effective VoTs}), 
high-VoT eligible users will outmatch low-VoT eligible users in securing express lane access. 
Thus, at equilibrium, DBCP provides scant express lane access to both ineligible and eligible users with low VoTs, who must then bear the brunt of longer travel times on GP lanes.
\revisionTRPartB{We believe that multi-level subsidy allocations, which provide income-dependent discounts or budgets to individual users, are critical to promote more equitable express lane access across heterogeneous users and O–D pairs, and their design remains an important direction for future work.}


\revisionTRPartB{While our analysis assumes fixed travel demand and focuses on users' lane and route choices, congestion pricing may also induce elastic user responses, including changes in trip frequency, departure time, carpooling, transit use, or the decision not to travel. Accounting for heterogeneity in these responses across income groups is particularly important for assessing the equity impacts of DBCP and CBCP policies, and remains an important direction for future work.}
Indeed, many real-world equitable congestion pricing policies incentivize users to carpool or take public transit as a means to reduce freeway congestion. 
For example, under DBCP policies in the San Francisco Bay Area, users with income below $200\%$ of the federal poverty level
receive larger discounts if they carpool rather than drive alone.
Similarly, CBCP policies in the Bay Area offer users the option to receive public transit fares in lieu of express lane credits \citep{CBCP-SanMateo}. 
The design of express lane tolls and subsidies for HOVs, though beyond the scope of this work, may nonetheless benefit from insights derived in Sec. \ref{sec: CBCP and DBCP Comparison Study for the Chain Network Setting} concerning the impact of DBCP and CBCP on users who commute in SOVs.


\section{Conclusion}
\label{sec: Conclusion}



Our work presented theoretical and empirical analyses to compare the efficacy of deploying CBCP and DBCP on traffic systems to minimize users' travel costs and generate toll revenue.
Building upon existing CBCP and heterogeneous tolling formulations, we developed a mixed-economy model in which eligible users receive either lump-sum budgets or toll discounts to access tolled express lanes, while ineligible users must pay out-of-pocket for express lane access. We presented convex programs to compute or approximate equilibrium flow sets under
CBCP and DBCP.
As our main contribution, we proved that under practically relevant conditions 
(Assumptions \ref{Assum, h: zero sublevel set of h is bounded}-\ref{Assum: Eligible users cannot access express lane without subsidy under DBCP}), when the societal cost is defined to prioritize maximizing toll revenue over minimizing eligible users' costs (i.e., $\lambda_R \geq \lambda_E$), CBCP cannot strictly outperform DBCP in minimizing the societal cost (Thm. \ref{Thm: Main}). 
To complement Thm. \ref{Thm: Main}, we 
proved
that when
$\lambda_R < \lambda_E$,
CBCP may strictly outperform DBCP in minimizing the societal cost (Prop. \ref{Prop: Assump lambda R geq lambda E removed}).
We validated our theoretical results by simulating the deployment of
CBCP and DBCP on segments of the U.S. 101 freeway in the San Francisco Bay Area, and derived key policy implications.

\section*{Declaration of Generative AI and AI-assisted Technologies in the Manuscript Preparation Process}

During the preparation of this work, the authors used ChatGPT (OpenAI) solely to improve the clarity, conciseness, and readability of the manuscript through language editing and refinement. After using this tool, the authors reviewed and edited the content as needed and take full responsibility for the content of the published article.

\bibliographystyle{plainnat}
\bibliography{references}

\clearpage
\appendix

\section{Proofs of Main Lemma and Main Theorem}
\label{sec: App, Proofs of Main Lemma and Main Theorem}

\subsection{Proof of Lemma \ref{Lemma: Main, Explicit Construction of DBCP Policy}}
\label{subsec: App, Proof of Lemma, Main, Explicit Construction of DBCP Policy}

First, if $S_{C, \boldlambda} = \emptyset$, 
then $f_{C, \boldlambda}^\star = \infty \geq f_{D, \boldlambda}^\star$ and we are done. Thus, below, we assume that there exists some $(\boldy, \boldx, \boldtau, \boldB) \in S_{C, \boldlambda}$.
Fix $e \in \edges$ and $t \in [T]$ arbitrarily. For each of the following cases, we will construct $\alpha_{e,t} \in [0, 1]$ and $\hat y_{e,t} := (\hat y_{e,k,t}^g: k \in [2], g \in G_e) \in \R_{\geq 0}^{2|G_e|}$ that satisfy \eqref{Eqn: Main Lemma, matching total eligible user express lane flow}  and \eqref{Eqn: Main Lemma, matching ineligible user express lane flow}: (Case 1) $\tau_{e,t} > 0$ and $x_{e,1,t} < x_{e,2,t}$; (Case 2) $\tau_{e,t} > 0$ and $x_{e,1,t} \geq x_{e,2,t}$; (Case 3) $\tau_{e,t} = 0$.
We then show that $\hat y_{e,t}$ is a $(\tau_{e,t}, \alpha_{e,t})$-DBCP equilibrium on edge $e$ for each $e \in \edges$, $t \in [T]$. Since Lemma \ref{Lemma: DBCP Eq Decomposition, Chain Network} asserts that the computation of DBCP equilibria can be decomposed across edges $e \in \edges$ and across periods $t \in [T]$, we conclude that $\hat \boldy$ is a $(\boldtau, \boldalpha)$-DBCP equilibrium.

\noindent
\textbf{Case 1: $\tau_{e,t} > 0$ and $x_{e,1,t} < x_{e,2,t}$.}
%
Define:
{
\setlength{\abovedisplayskip}{2pt}
\setlength{\belowdisplayskip}{3pt}
\begin{align} 
\label{Eqn: bar vI and bar vE, Def}
    \bar v_{e,t}^I &:= \frac{\tau_{e,t}}{\ell_e(x_{e,2,t}) - \ell_e(x_{e,1,t})}, \qquad
    \bar v_{e,t}^E := \max_{g \in G_e^E} \Bigg\{ v_t^g: \sum_{g' \in G_e^E: v_t^{g'} \geq v_t^g} d^{g'} \geq \sum_{g' \in G_e^E} y_{e,1,t}^{g'} \Bigg\}.
\end{align}
}
$\hspace{-1mm}$Under the assumption that $\tau_{e,t} > 0$ and $x_{e,1,t} < x_{e,2,t}$, we have $0 < \bar v_{e,t}^I < \infty$. (Recall that $\ell_e(\cdot)$ is strictly increasing). Moreover, since $\sum_{g' \in G_e^E: v_t^{g'} \geq \min_{g \in G_e^E} v_t^g} d^{g'} = \sum_{g' \in G_e^E} d^{g'} \geq \sum_{g' \in G_e^E} y_{e,1,t}^{g'}$:
{
\setlength{\abovedisplayskip}{3pt}
\setlength{\belowdisplayskip}{3pt}
\begin{align*}
    \Big\{\min_{g \in G_e^E} v_t^g \Big\} \subseteq \Bigg\{ v_t^g: \ g \in G_e^E, \ \sum_{g' \in G_e^E: v_t^{g'} \geq v_t^g} d^{g'} \geq \sum_{g' \in G_e^E} y_{e,1,t}^{g'} \Bigg\} \subseteq \{v_t^g: \ g \in G_e^E \}.
\end{align*}
}
$\hspace{-1mm}$In words, the set whose maximum value defines $\bar v_{e,t}^E$ is non-empty and finite, so $0 < \bar v_{e,t}^E < \infty$. 
To
interpret $\bar v_{e,t}^E$, consider the scenario where eligible users access the express lane on edge $e$ at period $t$ in descending order of VoT 
until the total eligible users' express lane flow equals the level specified by 
$\boldy$, i.e, $\sum_{g' \in G_e^E} y_{e,1,t}^{g'}$. Then, 
$\bar v_{e,t}^E$ equals the lowest VoT among eligible users with express lane access. 
Since all
users in $g \in G_e$ with $v_t^g \geq v_{e,t}^I$ can access the express lane at equilibrium without subsidy (because $v_t^g \ell_e(x_{e,1,t}) + \tau_{e,t} \leq v_t^g \ell_e(x_{e,2,t})$ for all such $g \in G_e$), by Assumption \ref{Assum: Eligible users cannot access express lane without subsidy under DBCP}, all such groups must be ineligible.
Thus,
$\bar v_{e,t}^E < \bar v_{e,t}^I$.

Next, we define $\boldalpha := (\alpha_{e,t}: e \in \edges, t \in [T]) \in [0, 1]^{|\edges|T}$ component-wise as follows:
{
\setlength{\abovedisplayskip}{3pt}
\setlength{\belowdisplayskip}{2pt}
\begin{align} \label{Eqn: alpha, main theorem, Def}
    \alpha_{e,t} := 1 - \frac{\bar v_{e,t}^E}{\bar v_{e,t}^I}, \hspace{5mm} \forall \ e \in \edges, \ t \in [T].
\end{align}
}
$\hspace{-1mm}$Since $\bar v_{e,t}^I > \bar v_{e,t}^E > 0$, we have $\alpha_{e,t} \in (0, 1)$. 
Next,
for each $k \in [2]$, set $\hat y_{e,k,t}^g := y_{e,k,t}^g$ for each $g \in \G_e^I$, and set, for each $g \in \G_e^E$:
{
\setlength{\abovedisplayskip}{3pt}
\setlength{\belowdisplayskip}{3pt}
\begin{align} 
\label{Eqn: hat y def, expr and gp lanes, Case 1}
    \hat y_{e,1,t}^g &:= \begin{cases}
    d^g, \hspace{4.5cm} &\text{if } v^g > \bar v_{e,t}^E, \\
    \sum\limits_{g \in G_e^E} y_{e,1,t}^g - \sum\limits_{g \in G_e^E: v^g > \bar v_{e,t}^E} d^g, &\text{ if } v^g = \bar v_{e,t}^E, \\
    0, &\text{ if } v^g < \bar v_{e,t}^E.
    \end{cases}, 
    \hspace{1cm} \hat y_{e,2,t}^g = d^g - \hat y_{e,1,t}^g.
\end{align}
}
$\hspace{-1mm}$By construction, $\hat y_{e,t} := (\hat y_{e,k,t}^g: k \in [2], g \in G_e)$ satisfies \eqref{Eqn: Main Lemma, matching total eligible user express lane flow}-\eqref{Eqn: Main Lemma, matching ineligible user express lane flow}, and
assigns users to the express lane in decreasing order of effective VoT; 
we can then readily verify that $\hat y_{e,t}$ satisfies 
\eqref{Eqn: DBCP Equilibrium, Def}.
Thus, $\hat y_{e,t}$ is a $(\tau_{e,t}, \alpha_{e,t})$-DBCP equilibrium.
\noindent
\textbf{Case 2: $\tau_{e,t} > 0$ and $x_{e,1,t} \geq x_{e,2,t}$.} 
%
We first prove that no ineligible user accesses the express lane, and $x_{e,1,t} = x_{e,2,t} = \frac{1}{2}d_e$. For any $g \in G_e^I$, we have $v_t^g \ell_e(x_{e,1,t}) + \tau_{e,t} > v_t^g \ell_e(x_{e,1,t}) \geq v_t^g \ell_e(x_{e,2,t})$ 
so no ineligible user would access the express lane (a strictly more costly travel option).
Moreover, if $x_{e,1,t} > x_{e,2,t}$, then since $\ell_e$ strictly increases, $v_t^g \ell_e(x_{e,1,t}) > v_t^g \ell_e(x_{e,2,t})$ for any $g \in G_e^E$. Thus, if $x_{e,1,t} > x_{e,2,t}$, no user on edge $e$ (eligible or ineligible) would access the express lane, so $x_{e,1,t} = 0$, a contradiction. Therefore, $x_{e,1,t} = x_{e,2,t}$. Since $\network$ is a chain network, $x_{e,1,t} + x_{e,2,t} = d_e$, and so $x_{e,1,t} = x_{e,2,t} = \frac{1}{2}d_e$.

We now set $\alpha_{e,t} := 1$ and consider the $(\tau_{e,t}, \alpha_{e,t})$-DBCP game $\G_{e,t}^D$. Take $\hat y_{e,t} = y_{e,t}$.
Then $\hat y_{e,t}$ satisfies \eqref{Eqn: Main Lemma, matching total eligible user express lane flow} and \eqref{Eqn: Main Lemma, matching ineligible user express lane flow} and $\hat x_{e,1,t} = x_{e,1,t} = \frac{1}{2}d$, $\hat x_{e,2,t} = x_{e,2,t} = \frac{1}{2}d$. 
Since eligible users can access the express lane toll-free when $\alpha_{e,t} = 1$, while ineligible users must pay a non-zero toll, $c_{e,1,t}^g(\hat \boldy; \G_{e,t}^D) > c_{e,2,t}^g(\hat \boldy; \G_{e,t}^D)$ for each $g \in G_e^I$ and $c_{e,1,t}^g(\hat \boldy; \G_{e,t}^D) = c_{e,2,t}^g(\hat \boldy; \G_{e,t}^D)$ for each $g \in G_e^E$.
Thus, $\hat y_{e,t}$ satisfies \eqref{Eqn: DBCP Equilibrium, Def}, and is therefore a $(\tau_{e,t}, \alpha_{e,t})$-DBCP equilibrium flow.
\noindent
\textbf{Case 3: $\tau_{e,t} = 0$.} 
%
We
first observe that the equilibrium $\boldy$ for the $(\boldtau, \boldB)$-CBCP Game $\G^C$ with associated lane flows $\boldx$ satisfies $x_{e,1,t} = x_{e,2,t}$. Indeed, if $x_{e,1,t} < x_{e,2,t}$, then the fact that $\ell_e$ is strictly increasing implies $c_{e,1,t}^g(\boldy; \G^C) = v_t^g \ell_e(x_{e,1,t}) < v_t^g 
\ell_e(x_{e,2,t}) = c_{e,2,t}^g(\boldy; \G^C))$. The variational inequalities that characterize CBCP equilibria (Def. \ref{Def: CBCP Equilibria}) then imply $y_{e,1,t}^g = 0$ for each $g \in G_e$, and so $x_{e,2,t} = \sum_{g \in G_e} y_{e,2,t}^g = 0$.
But then we obtain that $x_{e,1,t} < x_{e,2,t} = 0$, a contradiction.
Similarly, if $x_{e,1,t} > x_{e,2,t}$, then $y_{e,1,t}^g = 0$ for each $g \in G_e$, and so $x_{e,1,t} = \sum_{g \in G_e} y_{e,1,t}^g = 0$, which yields $x_{e,2,t} < x_{e,1,t} = 0$, a contradiction. Thus, $x_{e,1,t} = x_{e,2,t} = \frac{1}{2}d_e$.
We now fix $\hat y_{e,t} = y_{e,t}$,
and $g \in G_e$ and fix $\alpha_{e,t} \in [0, 1]$ arbitrarily. 
Clearly, $\hat y_{e,t}$ satisfies \eqref{Eqn: Main Lemma, matching total eligible user express lane flow} and \eqref{Eqn: Main Lemma, matching ineligible user express lane flow}.
Moreover, since $\tau_{e,t} = 0$ and $x_{e,1,t} = x_{e,2,t}$, we have $c_{e,1,t}^g(\hat \boldy; \G_{e,t}^D) = c_{e,2,t}^g(\hat \boldy; \G_{e,t}^D)$, so by \eqref{Eqn: DBCP Equilibrium, Def}, any user flow on edge $e$ and period $t$ (including $\hat y_{e,t}$) is a $(\tau_{e,t}, \alpha_{e,t})$-DBCP equilibrium flow.

In summary, given any $\boldtau, \boldB \geq 0$, and $(\boldtau, \boldB)$-CBCP equilibrium flow $\boldy$, for each $e \in \edges$ and $t \in [T]$, there exists $\alpha_{e,t} \in [0, 1]$ and a $(\tau_{e,t}, \alpha_{e,t})$-DBCP equilibrium flow $\hat y_{e,t}$ such that $\hat \boldy$ satisfies \eqref{Eqn: Main Lemma, matching total eligible user express lane flow}-\eqref{Eqn: Main Lemma, matching ineligible user express lane flow}. By Lemma \ref{Lemma: DBCP Eq Decomposition, Chain Network}, $\hat \boldy \in \Y^\eq(\G^D(\boldtau, \boldalpha))$ satisfies \eqref{Eqn: Main Lemma, matching total eligible user express lane flow} and \eqref{Eqn: Main Lemma, matching ineligible user express lane flow}, completing our proof.

\subsection{Proof of Thm. \ref{Thm: Main}}
\label{subsec: App, Proof of Thm, Main}

We leverage
Lemma \ref{Lemma: Main, Explicit Construction of DBCP Policy} to show that, for any $(\boldy, \boldx, \boldtau, \boldB) \in S_{C, \boldlambda}$, there exists some $\hat \boldy$, $\hat \boldx$ and $\alpha \in [0, 1]$ such that $(\hat \boldy, \hat \boldx, \boldtau, \boldalpha) \in S_{D, \boldlambda}$ and $f_{D, \boldlambda}(\hat \boldy, \hat \boldx, \boldtau, \boldalpha) \leq f_{C, \boldlambda}(\boldy, \boldx, \boldtau, \boldB)$. We then take the infimum over $(\boldy, \boldx, \boldtau, \boldB) \in S_{C, \boldlambda}$ and $(\hat \boldy, \hat \boldx, \boldtau, \boldalpha) \in S_{D, \boldlambda}$, to establish 
that $f_{D, \boldlambda}^\star \leq f_{C, \boldlambda}^\star$.

Fix $\boldtau \in \R_{\geq 0}^{|\edges|T}, \boldB \geq 0$ arbitrarily.
Let $\hat\boldtau \in \R_{\geq 0}^{|\edges|T}$, $\hat \boldy \in \Y$, and $\hat \boldx \in \R_{\geq 0}^{2|\edges|T}$ be as constructed in the proof of Lemma \ref{Lemma: Main, Explicit Construction of DBCP Policy}, i.e., $\hat\boldtau \geq \boldtau$ component-wise, and $\hat \boldy$ is a $(\hat \boldtau, \boldalpha)$-DBCP equilibrium flow with corresponding equilibrium lane flows $\hat \boldx \in \R_{\geq 0}^{2|\edges|T}$, with $\hat \boldy$ satisfying $\hat y_{e,1,t}^g = \hat y_{e,1,t}^g, \ \forall \ g \in G_e^I$ and $\sum_{g \in G_e^E} \hat y_{e,1,t}^g = \sum_{g \in G_e^E} y_{e,1,t}^g$.
We claim that, for each $e \in \edges$ and $t \in [T]$:
{
\setlength{\abovedisplayskip}{3pt}
\setlength{\belowdisplayskip}{3pt}
\begin{align} \label{Eqn: DBCP policy constructed has lower weighted eligible user cost} 
    \sum_{k=1}^2 \sum_{g \in G_e^E} v_t^g \ell_e(\hat x_{e,k,t}) \hat y_{e,k,t}^g &\leq \sum_{k=1}^2 \sum_{g \in G_e^E} v_t^g \ell_e(x_{e,k,t}) y_{e,k,t}^g.
\end{align}
}
$\hspace{-1mm}$First, the left-hand side of \eqref{Eqn: DBCP policy constructed has lower weighted eligible user cost} can be written as:
{
\setlength{\abovedisplayskip}{3pt}
\setlength{\belowdisplayskip}{3pt}
\begin{align} \nonumber
    \sum_{k=1}^2 \sum_{g \in G_e^E} v_t^g \ell_e(\hat x_{e,k,t}) \hat y_{e,k,t}^g &= \ell_e(x_{e,2,t}) \cdot \sum_{g \in G_e^E} v_t^g \sum_{k=1}^2 \hat y_{e,k,t}^g - \big(\ell_e(x_{e,2,t}) - \ell_e(x_{e,1,t}) \big) \cdot \sum_{g \in G_e^E} v_t^g \hat y_{e,1,t}^g \\ \nonumber
    &= \ell_e(x_{e,2,t}) \cdot \sum_{g \in G_e^E} v_t^g d_e^g - \big(\ell_e(x_{e,2,t}) - \ell_e(x_{e,1,t}) \big) \\ \label{Eqn: Main Lemma, LHS of Eligible User total cost of travel time, rewritten}
    &\hspace{1cm} \cdot \Big( \sum_{g \in G_e^E: \frac{1}{1-\alpha_{e,t}} v_t^g = \bar v_{e,t}^E} v_t^g \hat y_{e,1,t}^g + \sum_{g \in G_e^E: \frac{1}{1-\alpha_{e,t}} v_t^g > \bar v_{e,t}^E} v_t^g d^g \Big),
\end{align}
}
\hspace{-1mm}where $\bar v_{e,t}^E$ is as defined in \eqref{Eqn: bar vI and bar vE, Def}.
Similarly, the right-hand side of \eqref{Eqn: DBCP policy constructed has lower weighted eligible user cost} can be written as:
{
\setlength{\abovedisplayskip}{3pt}
\setlength{\belowdisplayskip}{3pt}
\begin{align} 
\label{Eqn: Main Lemma, RHS of Eligible User total cost of travel time, rewritten}
    &\sum_{k=1}^2 \sum_{g \in G_e^E} v_t^g \ell_e(x_{e,k,t}) y_{e,k,t}^g = \ell_e(x_{e,2,t}) \cdot \sum_{g \in G_e^E} v_t^g d_e^g - \big(\ell_e(x_{e,2,t}) - \ell_e(x_{e,1,t}) \big) \cdot \sum_{g \in G_e^E} v_t^g y_{e,1,t}^g.
\end{align}
}
$\hspace{-1mm}$Thus, establishing \eqref{Eqn: DBCP policy constructed has lower weighted eligible user cost} is equivalent to showing that \eqref{Eqn: Main Lemma, LHS of Eligible User total cost of travel time, rewritten} $\leq$ \eqref{Eqn: Main Lemma, RHS of Eligible User total cost of travel time, rewritten}, which, since $\ell_e(x_{e,2,t}) \geq \ell_e(x_{e,2,t})$, is in turn equivalent to establishing 
the following inequality:
{
\setlength{\abovedisplayskip}{3pt}
\setlength{\belowdisplayskip}{3pt}
\begin{align} \label{Eqn: DBCP policy constructed has lower weighted eligible user cost, re-stated}
    \sum_{g \in G_e^E: \frac{1}{1-\alpha_{e,t}} v_t^g = \bar v_{e,t}^E} v_t^g \hat y_{e,1,t}^g + \sum_{g \in G_e^E: \frac{1}{1-\alpha_{e,t}} v_t^g > \bar v_{e,t}^E} v_t^g d^g &\geq \sum_{g \in G_e^E} v_t^g y_{e,1,t}^g.
\end{align}
}
We prove \eqref{Eqn: DBCP policy constructed has lower weighted eligible user cost, re-stated} as follows:
{
\setlength{\abovedisplayskip}{3pt}
\setlength{\belowdisplayskip}{3pt}
\begin{align*}
    &\sum_{g \in G_e^E: \frac{1}{1-\alpha_{e,t}} v_t^g = \bar v_{e,t}^E} v_t^g \hat y_{e,1,t}^g + \sum_{g \in G_e^E: \frac{1}{1-\alpha_{e,t}} v_t^g > \bar v_{e,t}^E} v_t^g d^g \\
    \geq \ &\sum_{g \in G_e^E: \frac{1}{1-\alpha_{e,t}} v_t^g = \bar v_{e,t}^E} v_t^g \hat y_{e,1,t}^g + \sum_{g \in G_e^E: \frac{1}{1-\alpha_{e,t}} v_t^g > \bar v_{e,t}^E} v_t^g (d^g - y_{e,1,t}^g) + \sum_{g \in G_e^E: \frac{1}{1-\alpha_{e,t}} v_t^g > \bar v_{e,t}^E} v_t^g y_{e,1,t}^g \\
    \geq \ &(1-\alpha) \bar v_{e,t}^E \cdot \Bigg( \sum_{g \in G_e^E: \frac{1}{1-\alpha_{e,t}} v_t^g = \bar v_{e,t}^E} \hat y_{e,1,t}^g + \sum_{g \in G_e^E: \frac{1}{1-\alpha_{e,t}} v_t^g > \bar v_{e,t}^E} (d^g - y_{e,1,t}^g) \Bigg) + \sum_{g \in G_e^E: \frac{1}{1-\alpha_{e,t}} v_t^g > \bar v_{e,t}^E} v_t^g y_{e,1,t}^g \\
    = \ &(1-\alpha) \bar v_{e,t}^E \cdot \Bigg( \sum_{g \in G_e^E} y_{e,1,t}^g - \sum_{g \in G_e^E: \frac{1}{1-\alpha_{e,t}} v_t^g > \bar v_{e,t}^E} y_{e,1,t}^g \Bigg) + \sum_{g \in G_e^E: \frac{1}{1-\alpha_{e,t}} v_t^g > \bar v_{e,t}^E} v_t^g y_{e,1,t}^g \\
    = \ &(1-\alpha) \bar v_{e,t}^E \cdot \sum_{g \in G_e^E: \frac{1}{1-\alpha_{e,t}} v_t^g \leq \bar v_{e,t}^E} y_{e,1,t}^g + \sum_{g \in G_e^E: \frac{1}{1-\alpha_{e,t}} v_t^g > \bar v_{e,t}^E} v_t^g y_{e,1,t}^g \\
    \geq \ &\sum_{g \in G_e^E} v_t^g y_{e,1,t}^g,
\end{align*}
}
$\hspace{-1mm}$thus establishing \eqref{Eqn: DBCP policy constructed has lower weighted eligible user cost}.
For any $\lambda_E, \lambda_R, \lambda_I \geq 0$ satisfying $\lambda_R \geq \lambda_E$, 
from
\eqref{Eqn: f C lambda, Def} and \eqref{Eqn: f D lambda, Def}:
{
\setlength{\abovedisplayskip}{3pt}
\setlength{\belowdisplayskip}{3pt}
\begin{align} \nonumber
    &f_{D, \boldlambda}(\hat \boldy, \hat \boldx, \boldtau, \boldalpha) \\ 
    \label{Eqn: f D leq f C, term for extra revenue collected}
    = \ &\lambda_E \cdot \sum_{t=1}^T \sum_{e \in \edges} \sum_{k=1}^2 \sum_{g \in G_e^E} v_t^g \ell_e(\hat x_{e,k,t}) \hat y_{e,k,t}^g - (\lambda_R - \lambda_E) \cdot \sum_{t=1}^T \sum_{e \in \edges} \sum_{g \in G_e^E} (1 - \alpha_{e,t}) \tau_{e,t} \hat y_{e,1,t}^g \\ \nonumber
    &\hspace{1cm} + \lambda_I \cdot \sum_{t=1}^T \sum_{e \in \edges} \sum_{k=1}^2 \sum_{g \in G_e^I} \big( v_t^g \ell_e(\hat x_{e,k,t}) + \textbf{1}\{k=1\} \tau_{e,t} \big) \hat y_{e,k,t}^g - \lambda_R \cdot \sum_{t=1}^T \sum_{e \in \edges} \sum_{g \in G_e^I} \tau_{e,t} \hat y_{e,1,t}^g \\ 
    \label{Eqn: f D leq f C, comparing hat y to y}
    \leq \ &\lambda_E \cdot \sum_{t=1}^T \sum_{e \in \edges} \sum_{k=1}^2 \sum_{g \in G_e^E} v_t^g \ell_e(x_{e,k,t}) y_{e,k,t}^g \\ \nonumber
    &\hspace{1cm} + \lambda_I \cdot \sum_{t=1}^T \sum_{e \in \edges} \sum_{k=1}^2 \sum_{g \in G_e^I} \big( v_t^g \ell_e(x_{e,k,t}) + \textbf{1}\{k=1\} \tau_{e,t} \big) y_{e,k,t}^g - \lambda_R \cdot \sum_{t=1}^T \sum_{e \in \edges} \sum_{g \in G_e^I} \tau_{e,t} y_{e,1,t}^g \\ \nonumber 
    = \ &f_{C, \boldlambda}(\boldy, \boldx, \boldtau, \boldB),
\end{align}
}
\hspace{-1mm}where 
\eqref{Eqn: f D leq f C, comparing hat y to y} follows from \eqref{Eqn: Main Lemma, matching total eligible user express lane flow}, \eqref{Eqn: Main Lemma, matching ineligible user express lane flow}, \eqref{Eqn: DBCP policy constructed has lower weighted eligible user cost}, and the inequality $\lambda_R \geq \lambda_E$. Thus,
$f_{D, \boldlambda}^\star \leq f_{D, \boldlambda}(\hat \boldy, \hat \boldx, \boldtau, \boldalpha) \leq f_{C, \boldlambda}(\boldy, \boldx, \boldtau, \boldB)$.
Minimizing over $(\boldy, \boldx, \boldtau, \boldB) \in S_{C, \boldlambda}$, we obtain $f_{D, \boldlambda}^\star \leq f_{C, \boldlambda}^\star$. 

\section{Additional Related Work}
\label{sec: App, Additional Related Work}

Below,
we review manuscripts relevant to our work beyond those covered in Sec. \ref{sec: Related Work}.



\paragraph{Artificial Currency Mechanisms:}

Our 
work is related to the use of artificial currency mechanisms to address equity and efficiency concerns in traffic management \citep{Budish2011CombinatorialAssignmentProblem, Elokda2024CARMA, Gorokh2020FromMonetarytoNonmonetaryMechanismDesignviaArtificialCurrencies}. However, most existing artificial currency mechanisms are formulated within a \textit{single-economy} setting, in which each user can access priced resources only by using artificial currencies. In contrast, our \revisionApp{work} considers a \textit{mixed-economy} setting in which only a subset of (\say{eligible}) users are allotted credits that facilitate access to tolled express lanes, while \revisionApp{all other} (\say{ineligible}) users must pay tolls out-of-pocket to access express lanes.


\paragraph{Zeroth-order Descent Methods for Bilevel Optimization:}

In addition to introducing novel insights on CBCP and DBCP equilibrium flows, our work also presents a bilevel optimization framework and a gradient-free optimization method to compute optimal CBCP and DBCP policies. Although bilevel optimization was also utilized in \cite{ChiuJalotaPavone2024CreditvsDiscount} and \cite{Jalota2022CreditBasedCongestionPricing} to compute optimal pricing policies, both \cite{ChiuJalotaPavone2024CreditvsDiscount} and \cite{Jalota2022CreditBasedCongestionPricing} used \textit{grid-based sampling approaches} to identify optimal toll, budget, and discount values. Unfortunately, the computational complexity of grid-based sampling methods scales exponentially in the 
dimension of the parameter space.
Thus, dense sampling methods are unsuitable for designing optimal pricing policies under the multi-segment highway setting considered in this work, in which the optimal toll or discount value may differ across distinct highway segments and time periods. Instead, we adapt the gradient-free bilevel optimization algorithm presented in \citep{Maheshwari2024CongestionPricingforEfficiencyandEquity} to approximately compute toll, budget, and discount values that characterize optimal CBCP and DBCP policies.

\paragraph{Comparisons to the Conference Manuscript \citep{ChiuJalotaPavone2024CreditvsDiscount}:}

While a preliminary version of this work appeared in a conference proceeding \citep{ChiuJalotaPavone2024CreditvsDiscount}, the present article delivers novel theoretical contributions and numerical analysis, rooted in more nuanced traffic model formulations and methodologies.
First, 
our work provides theoretical results comparing the deployment of CBCP and DBCP policies over consecutive highway segments that serve an arbitrary number of eligible and ineligible user groups, each of which may be associated with a distinct value of time and origin-destination pair. 
In contrast,
\cite{ChiuJalotaPavone2024CreditvsDiscount} formulates CBCP and DBCP policies only over single origin-destination traffic networks,
and analyzes the deployment of CBCP and DBCP policies only under the restrictive assumption that all users
share one or two VoTs.
Second, our work establishes theoretical results 
which compare the efficacy of CBCP and DBCP policies in optimizing a range of societal cost metrics that incorporate users’ travel costs and the generated toll revenue. In contrast, 
\cite{ChiuJalotaPavone2024CreditvsDiscount} 
only compares CBCP and DBCP policies on the extent to which eligible users access the tolled express lane, thus capturing the equity and efficiency objectives of congestion pricing in a far more limited sense relative to the societal cost metric studied in our article. 


\section{Sec. 
\ref{sec: CBCP and DBCP Policies}-\ref{sec: CBCP and DBCP Equilibrium Computation} Proofs and Additional Results}
\label{sec: App, Sections 3-4 Proofs and Additional Results}

\subsection{Low-VoT Users Will Not Voluntarily Pay Express Lane Tolls Out-of-Pocket in Chain Networks}
\label{subsec: Low-VoT Users Will Not Voluntarily Pay Express Lane Tolls Out-of-Pocket in Chain Networks}


\revisionApp{
Specifically, consider a CBCP policy which allows eligible users to pay express lane tolls out-of-pocket, specified by $(\boldtau, \boldB)$, where
$\boldtau := (\tau_{e,t}: e \in \edges, t \in [T]) \in \R_{\geq 0}^{|\edges|T}$ denotes express lane tolls in $\network$ across edges and periods, while $\boldB := (B^g: g \in G) \in \R_{\geq 0}^{|G|}$ denotes the total travel budget provided to each group for use throughout the $T$ periods,
with $B^g = 0$ for each $g \in G^I$.
At each period $t$, each ineligible user who uses the express lane on an edge $e \in \edges$ must pay the toll $\tau_{e,t}$ out-of-pocket, while eligible users may expend their available budget or pay out-of-pocket to access the express lane. 
For each edge $e \in \edges$ and period $t \in [T]$, we use $\tilde y_{e,1,t}^g$ (resp., $\hat y_{e,1,t}^g$) to model the flow level of users from each group $g \in \G$ who use part of their budget (resp., pays out of pocket) to access the express lane on edge $e$ at period $t$, and $y_{e,2,t}^g$ to denote the flow level of users from each group $g \in G$ who access the general purpose lane on edge $e$ at period $t$.
Then the $(\boldtau, \boldB)$-CBCP policy induces a congestion game $G^C(\boldtau, \boldB)$ with a feasible flow set $\Y(G^C(\boldtau, \boldB))$ encoding flow continuity, non-negativity, and budget constraints, as defined below in \eqref{Eqn: Flow Constraint Set, CBCP, with Out-of-Pocket Payments}:
{
\setlength{\abovedisplayskip}{3pt}
\setlength{\belowdisplayskip}{3pt}
\begin{align} \label{Eqn: Flow Constraint Set, CBCP, with Out-of-Pocket Payments}
    \Y(G^C(\boldtau, \boldB)) := \Bigg\{ \boldy \in \R_{\geq 0}^{3nT}:
    \ &d^g \cdot \textbf{1}\{i = p_o^g \} + \sum_{\hat e \in \edges_i^{\In}} (\tilde y_{\hat e, 1, t}^g + \hat y_{\hat e, 1, t}^g + y_{\hat e, 2, t}^g) \\ \nonumber
    = \ &d^g \cdot \textbf{1}\{i = p_d^g \} + \sum_{\hat e \in \edges_i^{\Out}} (\tilde y_{\hat e, 1, t}^g + \hat y_{\hat e, 1, t}^g + y_{\hat e, 2, t}^g), \forall \ i \in V, g \in G, t \in [T], \\ \nonumber
    &\sum_{t \in [T]} \sum_{e \in \edges} \tilde y_{e,1,t}^g \tau_{e,t} \leq B^g, \ \forall \ g \in G^E. \Bigg\}
\end{align}
}
}
\revisionApp{
To specify the cost incurred by each user from each group $g \in G$ on each edge $e \in \edges$ and at period $t \in [T]$, we define the maps $\tilde c_{e,1,t}^g(\cdot; G^C(\boldtau, \boldB)), \hat c_{e,1,t}^g(\cdot; G^C(\boldtau, \boldB)), c_{e,2,t}^g(\cdot; G^C(\boldtau, \boldB)): \Y(G^C(\boldtau, \boldB)) \ra \R$ below, 
which respectively denote the cost per user of accessing the express lane while paying tolls using budgets, accessing the express lane while paying tolls out-of-pocket, and using the GP lane: 
$\tilde c_{e,1,t}^g(\boldy; G^C(\boldtau, \boldB)) := v_t^g \ell_e(x_{e,1,t})$; $\hat c_{e,1,t}^g(\boldy; G^C(\boldtau, \boldB)) := v_t^g \ell_e(x_{e,1,t}) + \tau_{e,t}$; and $c_{e,2,t}^g(\boldy; G^C(\boldtau, \boldB)) := v_t^g \ell_e(x_{e,2,t})$.
}
\revisionApp{
Given a $(\boldtau, \boldB)$-CBCP policy of the form described above,
the CBCP equilibrium concept described below in Def. \ref{Def: CBCP Equilibrium, with Out-of-Pocket Payments Allowed} describes the steady-state user flow pattern at which no user can strictly lower their total travel cost via a unilateral deviation.
}

\begin{definition}[\textbf{CBCP Equilibrium, with Out-of-Pocket Payments Allowed}]
\label{Def: CBCP Equilibrium, with Out-of-Pocket Payments Allowed}
\revisionApp{
Given $\boldtau \in \R_{\geq 0}^{|\edges|T}$ and $\boldB \in \R_{\geq 0}^{|G^E|}$, we call $\boldy^\star$ a ($\boldtau, \boldB$)-CBCP equilibrium  if $\boldy^\star \in \Y(G^C(\boldtau, \boldB))$ and for each group $g \in \G$:
{
\setlength{\abovedisplayskip}{3pt}
\setlength{\belowdisplayskip}{3pt}
\begin{align} \label{Eqn: CBCP Equilibrium (with Out-of-Pocket Payments Allowed), Def}
    &\sum_{e \in \edges} \sum_{t=1}^T \Bigg[ (\tilde y_{e,1,t}^g - \tilde y_{e,1,t}^{g \star}) \tilde c_{e,1,t}^g(\boldy^\star; G^C(\boldtau, \boldB)) + (\hat y_{e,1,t}^g - \hat y_{e,1,t}^{g \star}) \hat c_{e,1,t}^g(\boldy^\star; G^C(\boldtau, \boldB)) \\ \nonumber
    &\hspace{2cm} + (y_{e,2,t}^g - y_{e,2,t}^{g \star}) c_{e,2,t}^g (\boldy^\star; G^C(\boldtau, \boldB)) \Bigg] \geq 0, \hspace{1cm} \forall \ \boldy \in \Y(G^C(\boldtau, \boldB)).
\end{align}
}
}
\end{definition}

\revisionApp{
In Prop. \ref{Prop: (CBCP) No Eligible User with Sufficiently Low VoT will Pay Express Lane Tolls Out-of-Pocket} below, we prove that when traversing consecutive highway segments (Assumption \ref{Assum: Chain Network}), no eligible group with below-median VoT, as specified by \eqref{Eqn: Eligible Users' VoT are Sufficiently Low} below, would pay out-of-pocket for express lane access under CBCP equilibrium flows specified by Def. \ref{Def: CBCP Equilibrium, with Out-of-Pocket Payments Allowed}.
}

\begin{proposition}
\label{Prop: (CBCP) No Eligible User with Sufficiently Low VoT will Pay Express Lane Tolls Out-of-Pocket}
\revisionApp{
Suppose Assumption \ref{Assum: Chain Network} holds, and for some $e \in \edges$ and $t \in [T]$:
{
\setlength{\abovedisplayskip}{3pt}
\setlength{\belowdisplayskip}{3pt}
\begin{align}
\label{Eqn: Eligible Users' VoT are Sufficiently Low}
    \sum_{g' \in G_e: v_t^{g'} > v_t^g} d^{g'} \geq \frac{1}{2} \cdot \sum_{g' \in G_e} d^{g'}, \hspace{1cm} \forall \ g \in G_e^E.
\end{align}
}
and $\tau_{e,t} > 0$. Then, for
any $\boldB \in \R_{\geq 0}^{|G|}$, if $\boldy^\star$ is a $(\boldtau, \boldB)$-CBCP equilibrium flow in the sense of Def. \ref{Def: CBCP Equilibrium, with Out-of-Pocket Payments Allowed}, then $\hat y_{e,1,t}^{g\star} = 0$ for each $g \in G_e^E$.
}
\end{proposition}

\begin{proof}{Proof}
\revisionApp{
Suppose by contradiction that $\hat y_{e,1,t}^{g\star} > 0$ for some $g \in G_e^E$. 
Since all users select cost-minimizing travel options \citep[Lecture 11]{Roughgarden2016TwentyLecturesonAlgorithmicGameTheory} at equilibrium, $\hat y_{e,1,t}^{g\star} > 0$ implies:
{
\setlength{\abovedisplayskip}{3pt}
\setlength{\belowdisplayskip}{3pt}
\begin{align}
\label{Eqn: VoT condition in Prop, No Eligible User with Sufficiently Low VoT will Pay Express Lane Tolls Out-of-Pocket}
    v_t^g \ell_e(x_{e,1,t}^\star) + \tau_{e,t} &\leq v_t^g \ell_e(x_{e,2,t}^\star).
\end{align}
}
}
\revisionApp{
Fix any group $\bar g \in G_e$ satisfying $v_t^{\bar g} > v_t^g$; by \eqref{Eqn: Eligible Users' VoT are Sufficiently Low}, at least one such group exists.
Define $\boldy \in \Y\big(G^C(\boldtau, \boldB) \big)$ such that $y_{e',k,t'}^{g'} = y_{e',k,t'}^{g'\star}$ for each $e' \in \edges$, $k \in [2]$, $t' \in [T]$, and $g' \ne g$, and:
{
\setlength{\abovedisplayskip}{3pt}
\setlength{\belowdisplayskip}{3pt}
\begin{align*}
    \tilde y_{e',1,t'}^{\bar g} &= \tilde y_{e',1,t'}^{\bar g \star}, \quad
    \hat y_{e',1,t'}^{\bar g} = \begin{cases}
        d^{\bar g} - \tilde y_{e',1,t'}^{\bar g\star}, \ &e' = e, t' = t, \\
        \hat y_{e',1,t'}^{g\star}, &\text{else}, 
    \end{cases}, \quad
    y_{e',2,t'}^{\bar g} = d^{\bar g} - \tilde y_{e',1,t'}^{\bar g} - \hat y_{e',1,t'}^{\bar g}.
\end{align*}
}
$\hspace{-1mm}$Since $\boldy^\star$ is a $(\boldtau, \boldB)$-CBCP equilibrium 
in the sense of Def. \ref{Def: CBCP Equilibrium, with Out-of-Pocket Payments Allowed}:
{
\setlength{\abovedisplayskip}{3pt}
\setlength{\belowdisplayskip}{3pt}
\begin{align}
    0 &\leq \sum_{e' \in \edges} \sum_{t' \in [T]} \Bigg[ (\tilde y_{e,1,t}^{\bar g} - \tilde y_{e',1,t'}^{\bar g \star}) \tilde c_{e',1,t'}^{\bar g}\big( \G^C(\boldtau, \boldalpha) \big) + (\hat y_{e',1,t'}^{\bar g} - \hat y_{e',1,t'}^{\bar g \star}) \hat c_{e',1,t'}^{\bar g} \big( \G^C(\boldtau, \boldalpha) \big) \\ \nonumber
    &\hspace{3cm} + (y_{e',2,t'}^{\bar g} - y_{e',2,t'}^{\bar g \star}) c_{e',2,t'}^{\bar g} \big( \G^C(\boldtau, \boldalpha) \big) \Big] \\ \nonumber
    &= (d^{\bar g} - \hat y_{e,1,t}^{\bar g \star} - \tilde y_{e,1,t}^{\bar g \star}) \cdot \big[ v_t^{\bar g} \ell_e(x_{e,1,t}^\star) + \tau_{e,t} \big] -  y_{e,2,t}^{\bar g\star} \cdot v_t^{\bar g} \ell_e(x_{e,2,t}^\star) \\ \nonumber
    &= y_{e,2,t}^{\bar g\star} \cdot \big[ v_t^{\bar g} \ell_e(x_{e,1,t}^\star) + \tau_{e,t} - v_t^{\bar g} \ell_e(x_{e,2,t}^\star) \big]
\end{align}
}
$\hspace{-1mm}$From \eqref{Eqn: VoT condition in Prop, No Eligible User with Sufficiently Low VoT will Pay Express Lane Tolls Out-of-Pocket} and the fact that $v_t^{\bar g} > v_t^g$, we have $v_t^{\bar g} \ell_e(x_{e,1,t}^\star) + \tau_{e,t} - v_t^{\bar g} \ell_e(x_{e,2,t}^\star) < 0$. Thus, $y_{e,2,t}^{\bar g\star} = 0$, so $\tilde y_{e,1,t}^{\bar g\star} + \hat y_{e,1,t}^{\bar g\star} = d^{\bar g}$ for each $\bar g \in G_e$ satisfying $v_t^{\bar g} > v_t^g$. As a result:
{
\setlength{\abovedisplayskip}{3pt}
\setlength{\belowdisplayskip}{3pt}
\begin{align}
    &x_{e,1,t}^\star = \sum_{g' \in G_e} (\tilde y_{e,1,t}^{g'\star} + \hat y_{e,1,t}^{g'\star}) \geq \sum_{\bar g \in G_e: v_t^{\bar g} > v_t^g} d^{\bar g} \geq \frac{1}{2} \sum_{g' \in G_e} d^{g'}.
\end{align}
}
$\hspace{-1mm}$Moreover, under Assumption \ref{Assum: Chain Network}, $x_{e,1,t}^\star + x_{e,2,t}^\star = \sum_{g' \in G_e} d^{g'}$, so $x_{e,1,t}^\star \geq x_{e,2,t}^\star$,
contradicting \eqref{Eqn: VoT condition in Prop, No Eligible User with Sufficiently Low VoT will Pay Express Lane Tolls Out-of-Pocket}. 
We conclude that $\hat y_{e,1,t}^{g\star} = 0$ for each $g \in G_e^E$.
}



\end{proof}

\begin{proposition}
\label{Prop: (DBCP) No Eligible User with Sufficiently Low VoT will Pay Express Lane Tolls Out-of-Pocket}
\revisionApp{
Suppose Assumption \ref{Assum: Chain Network} holds, and for some $e \in \edges$ and $t \in [T]$, \eqref{Eqn: Eligible Users' VoT are Sufficiently Low} and $\tau_{e,t} > 0$ hold.
If $\boldy^\star$ is a $(\boldtau, \boldzero)$-DBCP equilibrium flow, then $\hat y_{e,1,t}^{g\star} = 0$ for each $g \in G_e^E$.
}
\end{proposition}

\begin{proof}{Proof}
\revisionApp{
The proof of Prop. \ref{Prop: (DBCP) No Eligible User with Sufficiently Low VoT will Pay Express Lane Tolls Out-of-Pocket} parallels that of Prop. \ref{Prop: (CBCP) No Eligible User with Sufficiently Low VoT will Pay Express Lane Tolls Out-of-Pocket} and is omitted for brevity.
}
\end{proof}


\subsection{Route Flow Formulation for Credit-Based Congestion Pricing (CBCP)}
\label{subsec: Route Flow Formulation for Credit-Based Congestion Pricing (CBCP)}

\revisionApp{
In this section, we present an alternative, route-based CBCP formulation, in which budget constraints are enforced on the level of each individual eligible user, 
who would only be allowed to traverse routes over which the total express lane toll does not exceed their allotted budget.
}

\revisionApp{
We adopt the network and user population formulations described in the first two paragraphs of Sec. \ref{subsec: Setup}, and characterize a CBCP policy by a tuple $(\boldtau, \boldB)$. 
Here, $\boldtau := (\tau_{e,t}: e \in \edges, t \in [T]) \in \R_{\geq 0}^{|\edges|T}$ denotes express lane tolls on each edge and period, while $\boldB := (B^g: g \in G^E) \in \R_{\geq 0}^{|G^E|}$ describes the total travel credit allotted to each eligible group for use throughout the $T$ periods.
Consistent with our discussion in Sec. 3.3, we stipulate that ineligible users access express lanes by paying tolls out-of-pocket, while eligible users pay for express lane access solely via travel credits.
We denote by $\G^C(\boldtau, \boldB)$ the congestion game induced by a given $(\boldtau, \boldB)$-CBCP policy.
}

\revisionApp{
Next, we specify the set of admissible routes for each group $g \in G$ under a given $(\boldtau, \boldB)$-CBCP policy. First, by a \textit{route}, we refer to an ordered set $r := \{(e_1^r, k_1^r), \cdots, (e_{|r|}^r, k_{|r|}^r) \}$ of edge-lane index pairs satisfying (if $|r| \geq 2$) $(e_\alpha^r)_j = (e_{\alpha+1}^r)_i$ for each $\alpha \in [|r| - 1]$. We denote by $\routes$ the set of all routes. 
For each group $g \in G$ with o-d pair $p^g := (p_o^g, p_d^g) \in \nodes \times \nodes$, we define the set of routes connecting $p^g$ by $\routes^g := \{r \in \routes: (e_1^r)_i = p_o^g, (e_{|r|}^r)_j = p_d^g \}$.
Next, 
we define sets $\routes_{1:T}^g$ of allowable route tuples over the time horizon $T$ for each group $g$ under a given $(\boldtau, \boldB)$-CBCP policy:
{
\setlength{\abovedisplayskip}{3pt}
\setlength{\belowdisplayskip}{3pt}
\begin{subequations}
\begin{align}
\label{Eqn: Feasible Routes, Eligible Users}
    \routes_{1:T}^g(\G^C(\boldtau, \boldB)) &:= \Big\{\boldr := (r_1, \cdots, r_T): \ r_t \in \routes^g, \ \sum_{t \in [T]} \sum_{(e,k) \in r_t} \tau_{e,t} \cdot \textbf{1}\{k=1\} \leq B^g. \Big\}, \quad \forall \ g \in G^E, \\
\label{Eqn: Feasible Routes, Ineligible Users}
    \routes_{1:T}^g(\G^C(\boldtau, \boldB)) &:= \big\{ \boldr := (r_1, \cdots, r_T): \ r_t \in \routes^g \big\}, \quad \forall \ g \in G^I.
\end{align}
\end{subequations}
}
$\hspace{-1mm}$We then define the feasible user route flows for each group $g$ over the time horizon $T$ by:
{
\setlength{\abovedisplayskip}{3pt}
\setlength{\belowdisplayskip}{3pt}
\begin{subequations}
\label{Eqn: Feasible Route Flows}
\begin{align}
\label{Eqn: Feasible Route Flows, Eligible Users}
    \flows_{1:T}^g(\G^C(\boldtau, \boldB)) &:= \Big\{ \boldy^g := (y_\boldr^g: \boldr \in \routes_{1:T}^g(\G^C(\boldtau, \boldB))):
    \sum_{\boldr \in \routes_{1:T}^g(\G^C(\boldtau, \boldB))} y_\boldr^g = d^g
    \Big\}, \quad \forall g \in G^E, \\
\label{Eqn: Feasible Route Flows, per period, Ineligible Users}
    \flows^g(\G^C(\boldtau, \boldB)) &:= \Big\{(y_r^g: r \in \routes^g) \in \R_{\geq 0}^{|\routes^g|}: \sum_{r \in \routes^g} y_r^g = d^g \Big\}, \quad \forall g \in G^I, \\
\label{Eqn: Feasible Route Flows, Ineligible Users}
    \flows_{1:T}^g(\G^C(\boldtau, \boldB)) &:= \big\{ \boldy^g := \big( (y_{r_1}^g: r_1 \in \routes^g), \cdots, (y_{r_T}^g: r_T \in \routes^g) \big)
    \in (\flows^g(\G^C(\boldtau, \boldB)))^T \big\}, \
    \forall g \in G^I.
\end{align}
\end{subequations}
}
$\hspace{-1mm}$We denote the set of all feasible route flows $(\boldy^g \in \flows_{1:T}^g(\G^C(\boldtau, \boldB)): g \in G)$ under the $(\boldtau, \boldB)$-equilibrium by $\flows_{1:T}(\G^C(\boldtau, \boldB))$.
}

\revisionApp{
Users’ routing decisions across all groups and periods generate user route flows $\boldy := (\boldy^g \in \flows_{1:T}^g: g \in G)$, which induce edge flows $\boldx := (x_{e,k,t}: e \in \edges, k \in [2], t \in [T]) \in \R^{2|\edges|T}$ defined by:
{
\setlength{\abovedisplayskip}{3pt}
\setlength{\belowdisplayskip}{3pt}
\begin{align}
    x_{e,k,t} &:= \sum_{g \in G} \sum_{\boldr \in \routes_{1:T}^g(\G^C(\boldtau, \boldB))} y_\boldr^g \cdot \textbf{1}\{(e,k) \in r_t \}, \quad \forall \ e \in \edges, k \in [2], t \in [T].
\end{align}
}
$\hspace{-1mm}$Given user route flows $\boldy$ and corresponding edge flows $\boldx$, we define, in \eqref{Eqn: Route Costs} below, the cost $c_\boldr^g\big(\boldy; \G^C(\boldtau, \boldB) \big)$ incurred by each user in group $g \in G$ who selects routes $\boldr \in \routes_{1:T}^g(\G^C(\boldtau, \boldB))$:
{
\setlength{\abovedisplayskip}{3pt}
\setlength{\belowdisplayskip}{3pt}
\begin{subequations}
\label{Eqn: Route Costs}
\begin{align}
\label{Eqn: Route Costs, Eligible Users}
    c_\boldr^g\big(\boldy; \G^C(\boldtau, \boldB) \big) &:= \sum_{t \in [T]} \sum_{(e,k) \in r_t} \big[ v_t^g \ell_e(x_{e,k,t}) + \tau_{e,t} \cdot \textbf{1}\{k=1\} \big], \quad \forall g \in G^E, \\
\label{Eqn: Route Costs, per period, Ineligible Users}
    c_{r_t}^g\big(\boldy; \G^C(\boldtau, \boldB) \big) &:= \sum_{(e,k) \in r_t} v_t^g \ell_e(x_{e,k,t}), \quad \forall g \in G^I, \\
\label{Eqn: Route Costs, Ineligible Users}
    c_\boldr^g\big(\boldy; \G^C(\boldtau, \boldB) \big) &:= \sum_{t \in [T]} c_{r_t}^g\big(\boldy; \G^C(\boldtau, \boldB) \big), \quad \forall g \in G^I.
\end{align}
\end{subequations}
}
}
\revisionApp{
We now define an alternative, route flow-based formulation of the CBCP equilibrium.
}

\begin{definition}[\textbf{CBCP Equilibrium Route Flows}]
\label{Def: CBCP Equilibrium Route Flows}
\revisionApp{
Given $\boldtau \in \R_{\geq 0}^{|\edges|T}$ and $\boldB \in \R_{\geq 0}^{|G^E|}$,
we call a set of route flows $\boldy^\star := \{ \boldy^{g\star} \in \flows_{1:T}^g\big(\G^C(\boldtau, \boldB)\big): g \in G \}$ CBCP equilibrium route flows if and only if, 
for any eligible group $g \in G^E$:
{
\setlength{\abovedisplayskip}{3pt}
\setlength{\belowdisplayskip}{3pt}
\begin{align}
\label{Eqn: Variational Inequalities for Eligible groups, CBCP Equilibria for Route-Based Flows}
    \sum_{\boldr \in \routes_{1:T}^g} (y_\boldr^g - y_\boldr^{g\star}) c_\boldr^g \big(\boldy^\star; \G^C(\boldtau, \boldB) \big) \geq 0, \quad \forall \boldy \in \flows_{1:T}^g\big(\G^C(\boldtau, \boldB) \big),
\end{align}
}
and for any ineligible group $g \in G^I$:
{
\setlength{\abovedisplayskip}{3pt}
\setlength{\belowdisplayskip}{1pt}
\begin{align}
\label{Eqn: Variational Inequalities for Ineligible groups, CBCP Equilibria for Route-Based Flows}
    \sum_{r_t \in \routes^g} (y_{r_t}^g - y_{r_t}^{g\star}) c_{r_t}^g \big(\boldy^\star; \G^C(\boldtau, \boldB) \big) \geq 0, \quad \forall \boldy \in \flows^g\big(\G^C(\boldtau, \boldB) \big), \ t \in [T].
\end{align}}
}
\end{definition}

\revisionApp{
The existence of route-based CBCP equilibria (Def. \ref{Def: CBCP Equilibrium Route Flows}) follows from the continuity of costs \eqref{Eqn: Route Costs}
and standard arguments in variational inequality theory.
}

\begin{remark}
\label{Remark: Main Theoretical Results hold under CBCP Equilibrium Route Flows}
\revisionApp{
Under the above route-based CBCP formulation, our main theoretical results (Lemma \ref{Lemma: Main, Explicit Construction of DBCP Policy} and Thm. \ref{Thm: Main}) remain true, with proofs proceeding by replacing the lane flow $y_{e,k,t}^g$, for any $e \in \edges$, $k \in [2]$, $t \in [T]$, and $g \in G_e$, with the sum of group $g$'s route flows over the associated edge, lane, and period, i.e., $\sum_{\boldr \in \routes_{1:T}^g(\G^C(\boldtau, \boldB)): (e, k) \in r_t} y_\boldr^g$.
}
\end{remark}

\begin{remark}
\label{Remark: Computing CBCP Equilibrium Route Flows Entails Explicit Route Enumeration}
\revisionApp{
Although our main theoretical results (Lemma \ref{Lemma: Main, Explicit Construction of DBCP Policy} and Thm. \ref{Thm: Main}) hold under both route-based and lane-based CBCP formulations, we restrict our discussion to the latter 
in the main text. Our choice is motivated by the fact that the route-based CBCP formulation requires defining user flows for each group $g$ over all feasible routes, whose number can grow exponentially in
the network edge count and time horizon.
Thus, the route-based CBCP equilibrium flows of Def. \ref{Def: CBCP Equilibrium Route Flows} 
are in general intractable to compute.
While the transportation literature commonly applies the Flow Decomposition Theorem \citep[Thm. 3.5]{Ahuja1993NetworkFlowsTheoryAlgorithmsandApplications} \citep[Thm. 2.2]{Patriksson2015TheTrafficAssignmentProblem} to bypass explicit route enumeration and compute equilibrium lane flows via convex programming, this approach does not readily apply to the computation of equilibrium route or lane flows characterized by Def. \ref{Def: CBCP Equilibrium Route Flows}.
Crucially, 
the Flow Decomposition Theorem establishes a correspondence between route and lane flows in networked congestion games in the setting where each user is allowed to select any route connecting their origin and destination.
This condition, however, does not hold in the context of the route-based CBCP 
equilibrium flows defined in Def. \ref{Def: CBCP Equilibrium Route Flows},
due to the eligible user route restrictions imposed by \eqref{Eqn: Feasible Routes, Eligible Users}.
}
\end{remark}


\subsection{Proof of Prop. \ref{Prop: CBCP Equilibria, Sensitivity}}
\label{subsec: App, Proof of Prop, CBCP Equilibria, Sensitivity}


As presented in Prop. \ref{Prop: CBCP Equilibria, Sensitivity}, let $\tilde v^g$ be as given in \eqref{Eqn: tilde v g, Def}, and note that the minimizer set of the convex program \eqref{Eqn: Convex Program Statement, CBCP} coincides with the set of CBCP equilibria in the hypothetical scenario where the VoT of each eligible group $g \in G^E$ at each period $t \in [T]$ were replaced by the time-invariant value $\tilde v^g$. Concretely, let $\bar \boldy \in \Y(\G^C(\boldtau, \boldB))$ denote any flow in the minimizer set of the convex program \eqref{Eqn: Convex Program Statement, CBCP} where the VoTs of each eligible group $g \in G^E$ across periods, i.e., $\{v_t^g: t \in [T]\}$, are all replaced by the time-invariant value $\tilde v^g$. Then the variational inequality defining CBCP equilibria, i.e., \eqref{Eqn: CBCP Equilibrium, Def}, implies that for any feasible flow $\boldy \in \Y(\G^C(\boldtau, \boldB))$:
{
\setlength{\abovedisplayskip}{3pt}
\setlength{\belowdisplayskip}{3pt}
\begin{subequations}
\begin{align} \label{Eqn: Proof of Cor, Variational Inequality with Time-Invariant VoTs, 1}
    &\sum_{t=1}^T \sum_{e \in \edges} \sum_{k=1}^2 (y_{e,k,t}^g - \bar y_{e,k,t}^g) \tilde v^g \ell(\bar x_{e,k,t}) \geq 0, \hspace{5mm} \forall \ g \in G^E, \\ \label{Eqn: Proof of Cor, Variational Inequality with Time-Invariant VoTs, 2}
    &\sum_{t=1}^T \sum_{e \in \edges} \sum_{k=1}^2 (y_{e,k,t}^g - \bar y_{e,k,t}^g) \Big( v_t^g \ell(\bar x_{e,k,t}) + \tau_{e,t} \cdot \textbf{1}\{k = 1\} \Big) \geq 0, \hspace{5mm} \forall \ g \in G^I.
\end{align}
\end{subequations}
}

By the definitions of $\epsilon$-$(\boldtau, \boldB)$-CBCP equilibria (Def. \ref{Def: epsilon CBCP Equilibria}) and user costs under CBCP policies \eqref{Eqn: Costs, CBCP}, in order to establish Prop. \ref{Prop: CBCP Equilibria, Sensitivity}, we must prove that for any feasible flow $\boldy \in \Y(\G^C(\boldtau, \boldB))$:
{
\setlength{\abovedisplayskip}{3pt}
\setlength{\belowdisplayskip}{3pt}
\begin{subequations}
\begin{align} 
\label{Eqn: Inequality to establish to prove Prop, CBCP Equilibria, Sensitivity, for eligible users}
    &\sum_{e \in \edges} \sum_{k=1}^2 \sum_{t=1}^T (y_{e,k,t}^g - \bar y_{e,k,t}^g) \cdot
    v_t^g \ell(\bar x_{e,k,t})
    \geq -\epsilon, \hspace{5mm} \forall \ g \in G^E, \\ 
    \label{Eqn: Inequality to establish to prove Prop, CBCP Equilibria, Sensitivity, for ineligible users}
    &\sum_{e \in \edges} \sum_{k=1}^2 \sum_{t=1}^T (y_{e,k,t}^g - \bar y_{e,k,t}^g) \cdot
    \Big( v_t^g \ell(\bar x_{e,k,t}) + \tau_{e,t} \cdot \textbf{1}\{k = 1\} \Big) 
    \geq -\epsilon, \hspace{5mm} \forall \ g \in G^I,
\end{align}
\end{subequations}
}
with $\epsilon$ given by \eqref{Def: epsilon CBCP Equilibria}.
First, \eqref{Eqn: Inequality to establish to prove Prop, CBCP Equilibria, Sensitivity, for ineligible users} follows from 
\eqref{Eqn: CBCP Equilibrium, Def},
since $\epsilon > 0$. Next, we note that \eqref{Eqn: Inequality to establish to prove Prop, CBCP Equilibria, Sensitivity, for eligible users} holds, since for any $g \in G^E$:
{
\setlength{\abovedisplayskip}{3pt}
\setlength{\belowdisplayskip}{3pt}
\begin{align} \nonumber
    &\sum_{e \in \edges} \sum_{k=1}^2 \sum_{t=1}^T (y_{e,k,t}^g - \bar y_{e,k,t}^g) \cdot
    v_t^g \ell(\bar x_{e,k,t}) \\ \nonumber
    = \ &\tilde v^g \cdot \sum_{t=1}^T \sum_{e \in \edges} \sum_{k=1}^2 (y_{e,k,t}^g - \bar y_{e,k,t}^g) \cdot \ell(\bar x_{e,k,t}) + \sum_{t=1}^T (v_t^g - \tilde v^g) \cdot \sum_{e \in \edges} \sum_{k=1}^2 (y_{e,k,t}^g - \bar y_{e,k,t}^g) \cdot
    \ell(\bar x_{e,k,t}) \\ \label{Eqn: Proof of Prop, CBCP Equilibria, Sensitivity, Inequality, apply def of CBCP equilibrium}
    \geq \ &\sum_{t=1}^T (v_t^g - \tilde v^g) \cdot \sum_{e \in \edges} \sum_{k=1}^2 (y_{e,k,t}^g - \bar y_{e,k,t}^g) \cdot
    \ell(\bar x_{e,k,t}) \\ \label{Eqn: Proof of Prop, CBCP Equilibria, Sensitivity, Inequality, Rewriting in terms of difference in y and ell}
    = \ &\sum_{t=1}^T (v_t^g - \tilde v^g) \cdot \sum_{e \in \edges} (y_{e,1,t}^g - \bar y_{e,1,t}^g) \cdot
    \big[\ell(\bar x_{e,1,t}) - \ell(\bar x_{e,2,t})\big] \\ 
    \label{Eqn: Proof of Prop, CBCP Equilibria, Sensitivity, Inequality, bound each absolute value}
    \geq \ &- \delta_v T d^g \cdot \sum_{e \in \edges: g \in G_e^E} \ell(d_e) \\ \nonumber
    \geq \ &- \epsilon,
\end{align}
}
\hspace{-1.5mm}where \eqref{Eqn: Proof of Prop, CBCP Equilibria, Sensitivity, Inequality, apply def of CBCP equilibrium} follows by applying \eqref{Eqn: Proof of Cor, Variational Inequality with Time-Invariant VoTs, 1}, \eqref{Eqn: Proof of Prop, CBCP Equilibria, Sensitivity, Inequality, Rewriting in terms of difference in y and ell} follows since $\sum_{k=1}^2 y_{e,k,t}^g = \sum_{k=1}^2 \bar y_{e,k,t}^g = d^g$ for all $t \in [T]$ and all $e \in \edges$ for which $g \in G_e^E$, so $y_{e,1,t}^g - \bar y_{e,1,t}^g = - (y_{e,2,t}^g - \bar y_{e,2,t}^g)$
and \eqref{Eqn: Proof of Prop, CBCP Equilibria, Sensitivity, Inequality, bound each absolute value} follows since $|v_t^g - \tilde v^g| \leq d^g$ by definition of $\tilde v^g$ and $\delta_v$ (as given by \eqref{Eqn: tilde v g, Def} and \eqref{Eqn: delta v, Def}, respectively).




\section{Sec. \ref{sec: CBCP and DBCP Comparison Study for the Chain Network Setting} Proofs and Additional Results}
\label{sec: App, Section 5 Proofs and Additional Results}

\subsection{Latency Characteristics of Outside Options in Chain Networks}
\label{subsec: Latency Characteristics of Outside Options in Chain Networks}


\revisionApp{
While
our main results are presented over a chain network for clarity (Assumption \ref{Assum: Chain Network}), our model can incorporate congestible, toll-free off-highway roads that provide outside options between consecutive network nodes. 
Specifically, on any outside option which offers slower travel than highway lanes under congestion-free settings (e.g., residential streets), the equilibrium CBCP or DBCP flow equals either zero
or the equilibrium GP highway lane flow between the same consecutive nodes (Prop. \ref{Prop: Outside Option, with CBCP or DBCP}).
In the former case, the outside option is unoccupied at equilibrium and may thus be ignored; in the latter, the outside option exhibits latency characteristics akin to a GP highway lane, and may thus be incorporated into our modeling framework via Remark \ref{Remark: Latency Function Model for Multiple GP Lanes}.
}

\begin{proposition}
\label{Prop: Outside Option, with CBCP or DBCP}
\revisionApp{
Suppose Assumption \ref{Assum: Chain Network} holds, and each edge $e \in \edges$ comprises a tolled express lane (indexed $k = 1$), a toll-free GP lane (indexed $k = 2$), and a toll-free outside option (indexed $k = 3$), with latency functions given by $\ell_{e,1}(\cdot)$, $\ell_{e,2}(\cdot)$, and $\ell_{e,3}(\cdot)$, respectively where $\ell_{e,1}(\cdot) = \ell_{e,2}(\cdot) = \ell_e(\cdot)$, and $\ell_{e,3}(0) \geq \ell_e(0)$. Let $\boldtau \geq \boldzero$, $\emph{\boldB} \geq 0$, and $\boldalpha \in [0, 1]^{|\edges|T}$ be given, and let $\boldy^\star$ be a $(\boldtau, \emph{\boldB})$-CBCP equilibrium or a $(\boldtau, \boldalpha)$-DBCP equilibrium with corresponding edge flows $\boldx^\star$. Then, for each edge $e \in \edges$ and period $t \in [T]$, we have $x_{e,3,t}^\star = 0$ or $\ell(x_{e,3,t}^\star) = \ell(x_{e,2,t}^\star)$.
}
\end{proposition}

\begin{proof}{Proof}
\revisionApp{
It suffices to show that, if $x_{e,3,t}^\star > 0$ for some $e \in \edges$ and $t \in [T]$, then $\ell(x_{e,3,t}^\star) = \ell(x_{e,2,t}^\star)$.
First, we prove that $\ell_{e,3}(x_{e,3,t}^\star) \leq \ell_e(x_{e,2,t}^\star)$. 
Since, for each group $g$, only cost-minimizing travel options can be occupied with positive flow at equilibrium \cite[Lecture 11]{Roughgarden2016TwentyLecturesonAlgorithmicGameTheory}, we have $v_t^g \ell_{e,3}(x_{e,3,t}^\star) \leq v_t^g \ell_e(x_{e,2,t}^\star)$, and thus $\ell_{e,3}(x_{e,3,t}^\star) \leq \ell_e(x_{e,2,t}^\star)$, as desired.
Next, suppose by contradiction that $\ell_{e,3}(x_{e,3,t}^\star) < \ell_e(x_{e,2,t}^\star)$. 
Again, for each group $g \in G_e$, since only cost-minimizing travel options can be occupied with positive flow at equilibrium, we have $y_{e,2,t}^{g\star} = 0$, and thus $x_{e,2,t}^\star = 0$. 
But then $\ell_{e,3}(0) \leq \ell_{e,3}(x_{e,3,t}^\star) \leq \ell_e(x_{e,2,t}^\star) = \ell_e(0)$, a contradiction.
}
\end{proof}

\subsection{Proof of Corollary \ref{Cor: Main Thm, Strictness}}
\label{subsec: App, Proof of Cor, Main Thm, Strictness}

Given $(\boldy^\star, \boldx^\star, \boldtau^\star, \boldB^\star)$ satisfying the 
conditions presented
in the statement of Cor. \ref{Cor: Main Thm, Strictness}, the proof procedure for Lemma \ref{Lemma: Main, Explicit Construction of DBCP Policy} prescribes the construction of a $(\boldtau^\star, \boldalpha^\star)$-DBCP policy, with $\boldalpha^\star$ defined by \eqref{Eqn: bar vI and bar vE, Def} and \eqref{Eqn: alpha, main theorem, Def} with $(\boldy^\star, \boldx^\star, \boldtau^\star, \boldalpha^\star)$ replacing $(\boldy, \boldx, \boldtau, \boldalpha)$.
Moreover, the proof of Lemma \ref{Lemma: Main, Explicit Construction of DBCP Policy} establishes the existence of a corresponding $(\boldtau^\star, \boldalpha^\star)$-DBCP equilibrium flow $(\hat \boldy, \hat \boldx)$ satisfying \eqref{Eqn: Main Lemma, matching total eligible user express lane flow} and \eqref{Eqn: Main Lemma, matching ineligible user express lane flow} with $\boldy^\star$ replacing $\boldy$, i.e., $\hat y_{e,1,t}^g = (\boldy^\star)_{e,1,t}^g, \ \forall \ g \in G_e^I$ and $\sum_{g \in G_e^E} \hat y_{e,1,t}^g = \sum_{g \in G_e^E} (\boldy^\star)_{e,1,t}^g$.
The proof of Thm. \ref{Thm: Main} implies that to establish Cor. \ref{Cor: Main Thm, Strictness}, it suffices to verify that the term $(\lambda_R - \lambda_E) \cdot \sum_{t=1}^T \sum_{e \in \edges} \sum_{g \in G_e^E} (1 - \alpha_{e,t}^\star) \tau_{e,t}\hat y_{e,1,t}^g$, which appears as part of \eqref{Eqn: f D leq f C, term for extra revenue collected} with $(\boldy, \boldx, \boldtau, \boldB)$ replaced by $(\boldy^\star, \boldx^\star, \boldtau^\star, B^\star)$, must be strictly positive.
Since $\lambda_R > \lambda_E$, a sufficient condition
for this requirement
is the existence of some $e \in \edges$, $t \in [T]$ for which 
$\tau_{e,t}^\star > 0$, $\alpha_{e,t}^\star < 1$, and $\sum_{g' \in G_e^E} \hat y_{e,1,t}^{g'} > 0$:
By assumption, there exists $(\boldy^\star, \boldx^\star, \boldtau^\star, \boldB^\star) \in S_{C, \boldlambda}^\star$, $e \in \edges$, $t \in [T]$, and $g \in G_e^E$ such that $(\boldy^\star)_{e,1,t}^g > 0$ and $x_{e,1,t}^\star < x_{e,2,t}^\star$. First, if $\tau_{e,t}^\star = 0$, then as discussed in Lemma \ref{Lemma: Main, Explicit Construction of DBCP Policy}, we must have $x_{e,1,t}^\star = x_{e,2,t}^\star$, a contradiction to the fact that $x_{e,1,t}^\star < x_{e,2,t}^\star$. Thus, $\tau_{e,t}^\star > 0$.
Next, since $x_{e,1,t}^\star < x_{e,2,t}^\star$ and $\ell_e(\cdot)$ is strictly increasing for each $e \in \edges$ and $t \in [T]$, we have $\ell_e(x_{e,1,t}^\star) < \ell_e(x_{e,2,t}^\star)$. Moreover, we have established that $\tau_{e,t}^\star > 0$, so $\bar v_{e,t}^I \in (0, \infty)$. As discussed in the proof of Lemma \ref{Lemma: Main, Explicit Construction of DBCP Policy}, we also have $\bar v_{e,t}^E \in (0, \infty)$ with $\bar v_{e,t}^E < \bar v_{e,t}^I$, so $\alpha_{e,t}^\star \in (0, 1)$.
Finally, we have $\sum_{g' \in G_e^E} \hat y_{e,1,t}^{g'} = \sum_{g' \in G_e^E} (\boldy^\star)_{e,1,t}^{g'} \geq (\boldy^\star)_{e,1,t}^g > 0$.
This
concludes the proof of Cor. \ref{Cor: Main Thm, Strictness}.

\subsection{Proof of Prop. \ref{Prop: Assump lambda R geq lambda E removed}}
\label{subsec: App, Proof of Thm, Assump lambda R geq lambda E removed}


First, if $S_{D, \boldlambda} = \emptyset$, i.e., if the optimization problem \eqref{Eqn: f D lambda star, Def}, is infeasible, then $f_{D, \boldlambda}^\star = \infty \geq f_{C, \boldlambda}^\star$ and we are done. Thus, below, we assume there exists $\tau \in \R_{\geq 0}, \alpha \in [0, 1]$ and a $(\tau, \alpha)$-DBCP equilibrium flow $\boldy \in \R^{2n}$ and associated lane flows $\boldx \in \R_{\geq 0}^2$ such that $(\boldy, \boldx, \tau, \alpha) \in S_{D, \boldlambda}$. 
We first establish that, to guarantee
\eqref{Eqn: Counterexample 1, Eligible y bound} and \eqref{Eqn: Counterexample 1, x bound} hold, $\tau$ must be sufficiently large, and $\alpha$ must be sufficiently small. Formally, we will prove that there exists some $M_\tau > 0, m_\alpha \in (0, 1)$ such that, for any feasible $(\tau, \alpha, \boldy, \boldx)$ satisfying \eqref{Eqn: Counterexample 1, Eligible y bound} and \eqref{Eqn: Counterexample 1, x bound} in the statement of Prop. \ref{Prop: Assump lambda R geq lambda E removed}, we have $\tau \geq M_\tau$ and $\alpha \leq m_\alpha$.
We first define the following quantities:
{
\setlength{\abovedisplayskip}{3pt}
\setlength{\belowdisplayskip}{3pt}
\begin{align*}
    M_{\hat \boldx} &:= \frac{1}{2}\Big( M_x + \frac{1}{2}d_e \Big), \hspace{1cm}
    \bar v_\epsilon := \max\Bigg\{ v_1^g: \ g \in G, \sum_{g' \in G: v_1^{g'} \geq v_1^g} d_e^{g'} \geq M_{\hat \boldx} \Bigg\}, \\
    M_\tau &:= 
    \bar v_\epsilon \big[\ell_e(d_e - M_x) - \ell_e(M_x) \big].
\end{align*}
}
$\hspace{-1mm}$We will now prove that $\tau \geq M_\tau$. Suppose by contradiction that, for some $\tau < M_\tau$, there exists some $\alpha \in [0, 1]$ and $(\tau, \alpha)$-DBCP equilibrium flow $\boldy$, with corresponding lane flows $\boldx$, such that $x_{e,1,1} \leq M_x$. Then, for each $g \in G^I$ with $v_1^g \geq \bar v_\epsilon$:
{
\setlength{\abovedisplayskip}{3pt}
\setlength{\belowdisplayskip}{3pt}
\begin{align*}
    c_{e,1,1}^g(\boldy; \G^D(\tau, \alpha)) &= v_1^g \ell_e(x_{e,1,1}) + \tau_{e,1} = v_1^g \Big( \ell_e(x_{e,1,1}) + \frac{\tau_{e,1}}{v_1^g} \Big) < v_1^g \Big( \ell_e(M_x) + \frac{M_\tau}{\bar v_\epsilon} \Big) \\
    &= v_1^g \ell_e(d_e - M_x) \leq v_1^g \ell_e(d_e - x_{e,1,1}) = v_1^g \ell_e(x_{e,2,1}) = c_{e,2,1}^g(\boldy; \G^D(\tau, \alpha)).
\end{align*}
}
Similarly, if $v_1^{g^E} \geq \bar v_\epsilon$:
{
\setlength{\abovedisplayskip}{3pt}
\setlength{\belowdisplayskip}{3pt}
\begin{align*}
    c_{e,1,1}^{g^E}(\boldy; \G^D(\tau, \alpha)) &= v_1^{g^E} \ell_e(x_{e,1,1}) + (1 - \alpha_{e,1}) \tau_{e,1} 
    < v_1^{g^E} \Big( \ell_e(M_x) + \frac{M_\tau}{\bar v_\epsilon} \Big) = v_1^{g^E} \ell_e(d_e - M_x) \\
    &\leq v_1^{g^E} \ell_e(d_e - x_{e,1,1}) = v_1^{g^E} \ell_e(x_{e,2,1}) = c_{e,2,1}^{g^E}(\boldy; \G^D(\tau, \alpha)).
\end{align*}
}
$\hspace{-1mm}$Thus, for each group $g \in G_e$ such that $v_1^g \geq \bar v_\epsilon$, regardless of eligibility, $y_{e,1,t}^g = d_e^g$, so $x_{e,1,1} := \sum_{g \in G_e} y_{e,1,1}^g \geq \sum_{g \in G_e: v_1^g \geq \bar v_\epsilon} d_e^g \geq M_{\hat \boldx} > M_x$,
contradicting \eqref{Eqn: Counterexample 1, x bound}. 
We thus conclude that $\tau \geq M_\tau$.

Next, to prove that $\alpha < 1$, we define $\bar v_{\max}^I := \max\{v^g: g \in G_e^I \}$ and $m_\alpha := 1 - \frac{v^{g^E}}{\bar v_{\max}^I}$.
If $v^{g^E} \geq \bar v_{\max}^I$, the VoT of each eligible group would equal or exceed the VoT of each ineligible group, thus violating Assumption \ref{Assum: Eligible users cannot access express lane without subsidy under DBCP}.
Therefore, we have $v^{g^E} < \bar v_{\max}^I$, which implies $m_\alpha < 1$.

We now prove $\alpha \leq m_\alpha$.
Suppose by contradiction that $\alpha > m_\alpha$. Then, for any $g' \in G_e^I$, $v^{g'} \leq \bar v_{\max}^I = \frac{1}{1 - m_\alpha} v^{g^E} < \frac{1}{1 - \alpha} v^{g^E} = \bar v_{e,t}^{g^E}$,
where $\bar v_{e,t}^{g^E}$ denotes the effective VoT of eligible users in group $g^E$ with respect to $\alpha$.
Since at the DBCP equilibrium, users fill the express lane in decreasing order of effective VoT (Lemma \ref{Lemma: DBCP Equilibria and Effective VoTs}), the above inequality implies $y_{e,1,1}^{g^E} = d^{g^E}$, which violates the stipulation in \eqref{Eqn: Counterexample 1, Eligible y bound} that $y_{e,1,1}^{g^E} \leq M_y^E < d^{g^E}$. We conclude that $\alpha \leq m_\alpha < 1$.

Now, we define $B$ by $B^{g^E} := \tau \cdot y_{e,1,1}^{g^E}$. We denote by $\G^C$ (resp., $\G^D$) the $(\tau, B)$-CBCP (resp., $(\tau, \alpha)$-DBCP) game. For each $g \in G$, we denote by $\Y^g(\G^C)$ (resp., $\Y^g(\G^D)$) the feasible flow set with respect to the CBCP game $\G^C$ (resp., DBCP game $\G^D$).
Finally, we define $\Y(\G^C) := \prod_{g \in G} \Y^g(\G^C)$ and $\Y(\G^D) := \prod_{g \in G} \Y^g(\G^D).$
For any $g \in G_e^I$, we have $\Y^g(\G^C) = \Y^g(\G^D)$, and for any $g \in G_e^E$, we have $\Y^g(\G^C) \subseteq \Y^g(\G^D)$, since CBCP games impose limits on eligible users' express lane access, based on the amount of their assigned credit.

We now prove that 
$\boldy$ is also a $(\tau, B)$-CBCP equilibrium flow. 
For any $g \in G_e^I$:
{
\setlength{\abovedisplayskip}{3pt}
\setlength{\belowdisplayskip}{3pt}
\begin{align} \label{Eqn: Counterexample 1, CBCP eq establ, Ineligible}
    &c_{e,1,1}^g(y) (z_{e,1,1}^g - y_{e,1,1}^g) + c_{e,2,1}^g(y) (z_{e,2,1}^g - y_{e,2,1}^g) \\ \nonumber
    = \ &\big[ v_1^g \ell_e(x_{e,1,1}) + \tau \big] (z_{e,1,1}^g - y_{e,1,1}^g) + v_1^g \ell_e(x_{e,2,1}) (z_{e,2,1}^g - y_{e,2,1}^g) 
    \geq 0, \quad \forall \ z \in \Y^g(\G^D).
\end{align}
}
$\hspace{-1mm}$Since $\Y^g(\G^C) = \Y^g(\G^D)$, the above inequality holds for all $z \in \Y^g(\G^C)$. 
Similarly:
{
\setlength{\abovedisplayskip}{3pt}
\setlength{\belowdisplayskip}{3pt}
\begin{align} \label{Eqn: Counterexample 1, CBCP eq establ, Eligible}
    &c_{e,1,1}^{g^E}(y) (z_{e,1,1}^{g^E} - y_{e,1,1}^{g^E}) + c_{e,2,1}^{g^E}(y) (z_{e,2,1}^{g^E} - y_{e,2,1}^{g^E}) \\ \nonumber
    = \ &v_1^{g^E} \ell_e(x_{e,1,1}) (z_{e,1,1}^{g^E} - y_{e,1,1}^{g^E}) + v_1^{g^E} \ell_e(x_{e,2,1}) (z_{e,2,1}^{g^E} - y_{e,2,1}^{g^E}) \geq 0, \quad \forall \ z \in \Y^g(\G^C).
\end{align}
}
$\hspace{-1mm}$Since $\Y^g(\G^C) \subseteq \Y^g(\G^D)$, \eqref{Eqn: Counterexample 1, CBCP eq establ, Eligible} must hold for all $z \in \Y^g(\G^C)$. Thus, the variational inequalities characterizing CBCP equilibria are satisfied for $\boldy$, so $\boldy$ is indeed a $(\tau, B)$-CBCP equilibrium.

In summary, we have proved that $\tau \geq M_\tau$, $\alpha \leq m_\alpha$, and $y_{e,1,t}^g \geq m_y^E$, so we now have:
{
\setlength{\abovedisplayskip}{3pt}
\setlength{\belowdisplayskip}{3pt}
\begin{align*}
    f_{D, \boldlambda}(\boldy, \boldx, \tau, \alpha) = \ &\lambda_E \cdot \sum_{k=1}^2 
    v_t^{g^E} \ell_e(x_{e,k,t}) y_{e,k,t}^{g^E} - (\lambda_R - \lambda_I) \cdot \sum_{g \in G_e^I} \tau_{e,t} y_{e,1,t}^g \\
    &\hspace{1cm} - (\lambda_R - \lambda_E) \cdot (1 - \alpha_{e,t}) \tau_{e,t} y_{e,1,t}^g + \lambda_I \cdot \sum_{k=1}^2 \sum_{g \in G_e^I} v_t^g \ell_e(x_{e,k,t}) y_{e,k,t}^g \\
    \geq \ &\lambda_E \cdot \sum_{k=1}^2 
    v_t^{g^E} \ell_e(x_{e,k,t}) y_{e,k,t}^{g^E} - (\lambda_R - \lambda_I) \cdot \sum_{g \in G_e^I} \tau_{e,t} y_{e,1,t}^g \\
    &\hspace{1cm} - (\lambda_R - \lambda_E) \cdot (1 - m_\alpha) \tau_{e,t} y_{e,1,t}^g + \lambda_I \cdot \sum_{k=1}^2 \sum_{g \in G_e^I} v_t^g \ell_e(x_{e,k,t}) y_{e,k,t}^g \\
    = \ &f_{C, \boldlambda}(\boldy, \boldx, \tau, B) + (\lambda_E - \lambda_R) \cdot (1 - m_\alpha) M_\tau m_y^E.
\end{align*}
}
$\hspace{-1mm}$Thus, $f_{C, \boldlambda}^\star = f_{C, \boldlambda}(\boldy, \boldx, \tau, B) \leq f_{D, \boldlambda}(\boldy, \boldx, \tau, \alpha) - (\lambda_E - \lambda_R) \cdot (1 - m_\alpha) M_\tau m_y^E$, where the equality follows since
$\boldy$ is both a $(\tau, \alpha)$-DBCP and a $(\tau, B)$-CBCP equilibrium flow.
Taking the infimum over all $(\boldy, \boldx, \tau, \alpha) \in S_{D, \boldlambda}$, we find that \eqref{Eqn: Counterexample 1, optimal f C bound} holds, as desired.


\section{Additional Details on Numerical Experiments}
\label{sec: App, Additional Details on Numerical Experiments}

\subsection{Network Formulation and Calibration Details}
\label{subsec: App, Network Formulation and Calibration Details}

\subsubsection{Latency Calibration}
\label{subsubsec: App, Latency Calibration}

We note that the US-101 freeway segments over which the 101 Express Lanes Project is implemented, which $\network_\ELP$ models, contains one tolled express lane and $n_{GP} = 3$ toll-free GP lanes. Thus, 
we calibrate latency functions $\ell_{e,1}(\cdot), \ell_{e,2}(\cdot)$ for each $e \in \edges_\ELP$ such that $\ell_{e,2}(x) = \ell_{e,1}\big(\frac{x}{3} \big)$ for each $x \geq 0$ (Remark \ref{Remark: Latency Function Model for Multiple GP Lanes}). 
To tractably solve \eqref{Eqn: Convex Program Statement, DBCP}-\eqref{Eqn: Convex Program Statement, CBCP} at high dimensions, we adopt a standard piecewise affine approximation \citep{Jalota2022CreditBasedCongestionPricing, Salazar2019ACongestionAwareRoutingScheme} of the Bureau of Public Roads (BPR) travel latency model \citep{BPR1964TrafficAssignmentManual}: (Fig. \ref{fig: Latency params estimation})
{
\setlength{\abovedisplayskip}{4pt}
\setlength{\belowdisplayskip}{4pt}
\begin{align} 
\label{Eqn: Latency Functions for simulations, Piecewise Linear}
    \ell_{e,1}(x) &= \bar \ell_e + \beta_e \max\{x - \kappa_e, 0\}, \quad \ell_{e,2}(x) = \bar \ell_e + \beta_e \max\left\{\frac{x}{3} - \kappa_e, 0 \right\}, \hspace{5mm} \forall \ e \in \edges_\ELP.
\end{align}
}
$\hspace{-1mm}$where, for each lane on edge $e$,
$\bar \ell_e$ describes a constant travel time when the traffic flow is below $\kappa_e$, and $\beta_e$ describes the rate of travel time increase once the traffic flow exceeds $\kappa_e$.
To compute $\bar \ell_e$, $\beta_e$, and $\kappa_e$ for each $e \in \edges_\ELP$, we accessed Caltrans' Performance Measurement System (PeMS) \citep{pems-database} to collect weekday vehicle flow and travel speed data on the 101-N freeway from September to November in 2024. As in \cite{Jalota2022CreditBasedCongestionPricing}, we compute optimal values of $\bar \ell_e$, $\beta_e$, and $\kappa_e$ across $e \in \edges_\ELP$ that minimize the mean-squared error between travel latency values computed from PeMS data and from \eqref{Eqn: Latency Functions for simulations, Piecewise Linear} (see Table \ref{Table: Latency Function Parameters and VoT by group values} in App. \ref{sec: App, Data Table for Numerical Experiments}). 


\begin{figure}
    \centering
    \includegraphics[width=0.99\linewidth]{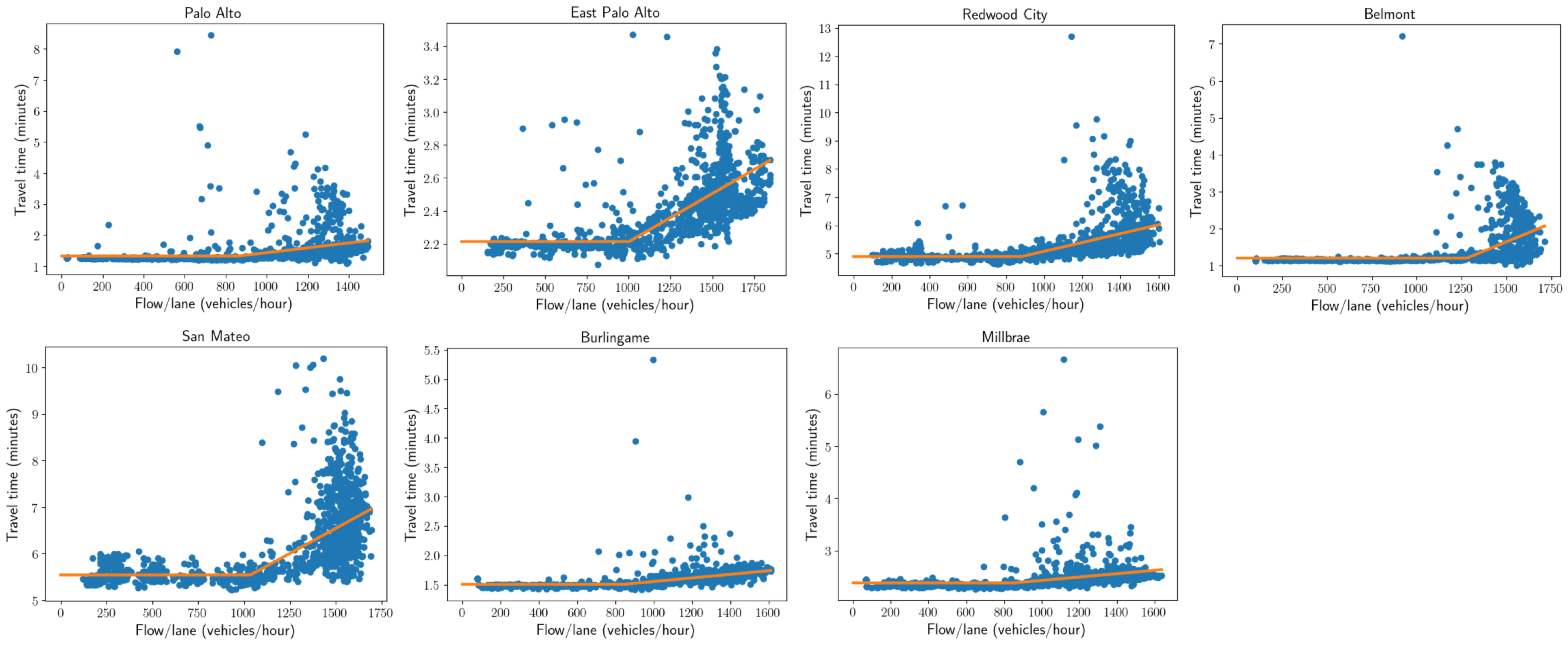}
    \caption{
    \small \sf
    Piecewise affine approximations (orange) of the latency functions $\ell_e$ across $e \in \edges_\ELP$
    in our network model $\network_\ELP$ for the 101 Express Lanes Project, calibrated using flow and travel latency data from Caltrans' Performance Measurement System (PeMS) database.
    }
    \label{fig: Latency params estimation}
\end{figure}

\subsubsection{User Demand}
\label{subsubsec: App, User Demand}

To estimate user demands between all origin-destination (o-d) pairs in $\network_\ELP$, we first retrieved, from PeMS,
the mainline, on-ramp, and off-ramp vehicle flow data on the 101-N freeway on weekdays from September to November 2024.
We use this data to estimate the traffic flow on each edge $e \in \edges_\ELP$, as well as input and output traffic flow at each node $i \in I$ in $\network_\ELP$, that satisfy flow continuity constraints at each node $i \in I$.
Then, to infer the user demand across each o-d pair in $\network_\ELP$ from traffic flow data
across edges and nodes,
we apply the entropy maximization approach proposed by \cite{Wilson1969UseOfEntropyMaximizingModels} and commonly used in the literature \citep{Kim2024EstimateThenPredict, Gonzalez2020TemporalODMatrixEstimation, Afifah2022SpatialPricing}. 

For each edge $e \in \edges_\ELP$, let $\Set_e$ denote the set of traffic flow measurements on edge $e$ collected from PeMS. For each node $i \in I$, let $\Set_i^\In$ and $\Set_i^\Out$ denote the set of data measurements of the traffic flow entering $\network_\ELP$ in and exiting $\network_\ELP$ from node $i$, respectively. Let $\mu_e$, $\mu_i^\In$, and $\mu_i^\Out$ denote the mean of $\Set_e$, $\Set_i^\In$, and $\Set_i^\Out$, respectively, and let $\sigma_e$, $\sigma_i^\In$, and $\sigma_i^\Out$ denote the standard deviation of $\Set_e$, $\Set_i^\In$, and $\Set_i^\Out$, respectively. 
We aim to compute traffic flows on each edge $e \in \edges_\ELP$, denoted $d_e$ below, and for the input and output traffic flow at each node $i \in I$, respectively denoted $d_i^\In$ and $d_i^\Out$ below, that satisfy flow continuity constraints at each node $i \in I$. For simplicity, we write $d := (d_e: e \in \edges_\ELP) \in \R_{\geq 0}^{|\edges|}$, $d^\In := (d_i^\In: i \in I) \in \R_{\geq 0}^{|I|}$, $d^\Out := (d_i^\Out: i \in I) \in \R_{\geq 0}^{|I|}$,
and denote by $i_{\ori}$ and $i_{\dest}$ the overall origin and destination nodes for $\network_\ELP$, respectively, i.e., $E_{i_{\ori}}^- = E_{i_{\dest}}^+ = \phi$.
We then solve:
{
\setlength{\abovedisplayskip}{4pt}
\setlength{\belowdisplayskip}{4pt}
\begin{subequations} \label{Eqn: Opt, Weighted LLS for d e}
\begin{align} \label{Eqn: Opt, Weighted LLS for d e, Objective}
    \min_{d \in \R_{\geq 0}^{|\edges|}, d^\In, d^\Out \in \R_{\geq 0}^{|\nodes|}} \hspace{5mm} &\sum_{e \in \edges_\ELP} \frac{(d_e - \mu_e)^2}{\sigma_e^2} + \sum_{i \in I} \frac{(d_i^\In - \mu_i^\In)^2}{(\sigma_i^\In)^2} + \sum_{i \in I} \frac{(d_i^\Out - \mu_i^\Out)^2}{(\sigma_i^\Out)^2}, \\ 
    \label{Eqn: Opt, Weighted LLS for d e, Constraint, Flow Continuity}
    \text{s.t.} \hspace{5mm} 
    &d_{e_i^\In} \cdot \textbf{1}\{i \ne i_\ori\} + d_i^\In = d_{e_i^\Out} \cdot \textbf{1}\{i \ne i_\dest\} + d_i^\Out, \quad \forall \ i \in \nodes.
\end{align}
\end{subequations}
}


Next, we aim to recover the user demand associated with the origin-destination pair $(i, j) \in P$, denoted below by $\bar d_{ij}$.
For simplicity, we write $\bar d := (\bar d_{ij}: (i,j) \in P) \in \R_{\geq 0}^{|P|}$. We then apply the entropy maximization procedure employed by \cite{Wilson1969UseOfEntropyMaximizingModels, Kim2024EstimateThenPredict}:
{
\setlength{\abovedisplayskip}{4pt}
\setlength{\belowdisplayskip}{4pt}
\begin{subequations} \label{Eqn: Opt, Trip Distribution}
\begin{align} \label{Eqn: Opt, Trip Distribution, Entropy Objective}
    \max_{\bar d \in \R_{\geq 0}^{|P|}} \hspace{5mm} &\sum_{(i,j) \in P} \big( \bar d_{ij} \log \bar d_{ij} + \bar d_{ij} \big) \\ 
    \label{Eqn: Opt, Trip Distribution, Origin and Destination Constraints}
    \text{s.t.} \hspace{5mm} &\sum_{i \in I \backslash \{i_{\dest}\}} \bar d_{ij} = d_i^\In, \quad \sum_{i \in I \backslash \{i_{\ori}\}} \bar d_{ij} = d_i^\Out. 
\end{align}
\end{subequations}
}
For our computed values of $\bar d_{ij}$, see 
Table \ref{Table: Demand by group values}.

\subsubsection{User Value-of-Time (VoT)}
\label{subsubsec: App, User Value-of-Time (VoT)}

We estimate user VoTs in cities between Palo Alto and Millbrae from the 2020 US Census American Community Survey (ACS) data \citep{ACS2021} on annual household wages. 
Although other valid methods for estimating users' VoTs exist, we assume, as is standard  \citep{Athira2016EstimationofValueofTravelTimeforWorkTrips, 
Tveter2022TheValueOfTravelTime}, that users' VoT distributions are proportional to the income distribution in their city of origin. 
For each city in $S_\city$, 
we
adjoin consecutive income intervals from ACS data to form the set of income brackets $S_\income := \{$ $\$0$-$\$14,999$, $\$15,000$-$\$34,999$, $\$35,000$-$\$99,999$, $\$100,000$-$\$199,999$, $> \$200,000 \}.$
We
designate only users in the 2 lowest income brackets to be eligible for receiving toll subsidies (credits or discounts). 
This eligibility threshold approximates that, of twice the federal poverty level, used in the real-world 101 Express Lanes Project. 
Our model designates 16.77\% of all users to be eligible, a figure consistent with existing CBCP and DBCP studies \citep{Jalota2022CreditBasedCongestionPricing, ChiuJalotaPavone2024CreditvsDiscount}.


For each city in $S_\city$ and income bracket in $S_\income$, 
we estimate the mean incomes of the four lower income brackets by taking a weighted average of the center of each income interval, with weights given by the relative fraction of users in each interval. 
We estimate \textit{individual} hourly wages of group $g$ by halving the computed family hourly wages. 
Next, we compute a \textit{baseline average VoT} for each group by dividing its estimated average individual income by the number of total work hours in a year (i.e., 40 hrs/week $\times$ 52 weeks). 
The computed baseline average VoTs are reported in Table \ref{Table: Latency Function Parameters and VoT by group values}.
To model time-varying VoTs for (only) ineligible users,
we draw $\delta_{t,g}$ uniformly from $[-0.2, 0.2]$ for each $g \in G^I$, while setting $\delta_{t,g} = 1$ for each $g \in G^E$.
We then compute the VoT of each group $g$ at each $t \in [T]$ via $v_t^g = 1 + v_g \delta_{t,g}$.

Finally, we compute the user demand for each o-d pair in $\network_{ELP}$ and income bracket in $\mathcal{S}_\income$ (Table \ref{Table: Demand by group values}), assuming for simplicity
that user groups with the same origin share VoT distributions.

\subsection{Algorithm Details for Computing First-Order Stationary DBCP and CBCP Policies}
\label{subsec: App, Algorithm Details for Computing First-Order Stationary DBCP and CBCP Policies}

\subsubsection{Zeroth-Order Optimization Algorithm}
\label{subsubsec: App, Zeroth-Order Optimization Algorithm}


To compute toll and budget parameters corresponding to the first-order stationary CBCP policy for a given set of Pareto weights $\boldlambda = (\lambda_E, \lambda_R, \lambda_I)$, we first randomly select initial iterates for the tolls $\boldtau^\paren{0} \in \R_{\geq 0}^{|\edges|T}$ and budgets $\boldB^\paren{0} \in \R_{\geq 0}^{|G^E|}$,  and fix a maximum number of algorithm iterations $N_\iter \in \N$. We then follow the iterative procedure prescribed in Alg. \ref{Alg: Opt CBCP, Zeroth-order Optimization}.
Concretely, at each iterate $i \in [N_\iter-1]$, given $\boldtau^\paren{i}$ and $\boldB^\paren{i}$, we randomly query $\hat{\boldtau}^\paren{i}$ and $\hat{\boldB}^\paren{i}$ near $\boldtau^\paren{i} \in \R_{\geq 0}^{|\edges|T}$ and $\boldB^\paren{i} \in \R_{\geq 0}^{|G^E|}$ 
(Alg. \ref{Alg: Opt CBCP, Zeroth-order Optimization}, Lines \ref{Algline: CBCP, Query v}-\ref{Algline: CBCP, Compute hat tau, hat B}). We then compute CBCP equilibrium flows corresponding to  $(\boldtau^\paren{i}, \boldB^\paren{i})$ and $(\hat{\boldtau}^\paren{i}, \hat{\boldB}^\paren{i})$, which we denote respectively by $y^\paren{i}$ and $\hat{y}^\paren{i}$, as well as the respective corresponding lane flows $x^\paren{i}$ and $\hat{x}^\paren{i}$ (Alg. \ref{Alg: Opt CBCP, Zeroth-order Optimization}, Lines \ref{Algline: CBCP, Compute hat y}-\ref{Algline: CBCP, Compute hat x}).
Next, we compute the societal cost $f_{C, \lambda}$ at 
$(\boldtau^\paren{i}, \boldB^\paren{i}, y^\paren{i}, x^\paren{i})$ and at $(\hat{\boldtau}^\paren{i}, \hat{\boldB}^\paren{i}, \hat y^\paren{i}, \hat x^\paren{i})$, and use the difference
to construct a zeroth-order surrogate $\hat \nabla f_{C, \lambda}(\boldtau^\paren{i}, \boldB^\paren{i}, y^\paren{i}, x^\paren{i})$ (Alg. \ref{Alg: Opt CBCP, Zeroth-order Optimization}, Line \ref{Algline: CBCP, Compute hat y}) for the gradient
$\nabla f_{C, \lambda}(\boldtau^\paren{i}, \boldB^\paren{i}, y^\paren{i}, x^\paren{i})$,
which we then use to compute the next-iterate toll and budget values,
i.e., $(\boldtau^\paren{i+1}, \boldB^\paren{i+1})$ (Alg. \ref{Alg: Opt CBCP, Zeroth-order Optimization}, Line \ref{Algline: CBCP, Gradient Update}), and compute the corresponding equilibrium user and lane flows, i.e., $y^\paren{i+1}$ and $x^\paren{i+1}$ (Alg. \ref{Alg: Opt CBCP, Zeroth-order Optimization}, Lines \ref{Algline: CBCP, Compute y, at iteration i+1}-\ref{Algline: CBCP, Compute x, at iteration i+1}). Finally, Alg. \ref{Alg: Opt CBCP, Zeroth-order Optimization} returns the toll and budget values corresponding to the stationary CBCP policy (Alg. \ref{Alg: Opt CBCP, Zeroth-order Optimization}, Line \ref{Algline: Extract optimal CBCP Policy Parameters}).

Next, we apply
a similar zeroth-order method
(Alg. \ref{Alg: Opt DBCP, Zeroth-order Optimization}) to compute tolls and discounts for the stationary DBCP policy given weights $\boldlambda = (\lambda_E, \lambda_R, \lambda_I)$. 
However, we leverage Lemma \ref{Lemma: DBCP Eq Decomposition, Chain Network} to compute DBCP equilibria 
\textit{separately}
across edges $e \in \edges$ and periods $t \in [T]$, in contrast to the \textit{centralized} CBCP equilibrium computation described in Alg. \ref{Alg: Opt CBCP, Zeroth-order Optimization}, Lines \ref{Algline: CBCP, Compute hat y}-\ref{Algline: CBCP, Compute hat x} and \ref{Algline: CBCP, Compute y, at iteration i+1}-\ref{Algline: CBCP, Compute x, at iteration i+1}.

\begin{algorithm} 
\small
\caption{
Zeroth-Order Optimization for Computing Stationary CBCP Policies via \eqref{Eqn: f C lambda, Def}
}
\label{Alg: Opt CBCP, Zeroth-order Optimization} 
\begin{algorithmic}
\Procedure{ZO-CBCP}{$T$, $\network_\ELP = (V_\ELP, E_\ELP)$, $\eta$, $\delta$, $\boldtau^\paren{0}$, $\boldB^\paren{0}$, $N_\iter$}


\State $S_{\tau B} \gets [0, C_\tau]^{|\edges|T} \times [0, \infty)$

\State $y^\paren{0} \gets \Y^\eq\big(\G^C(\boldtau^\paren{0}, \boldB^\paren{0}) \big)$

\State $x_{e,k,t}^\paren{0} \gets \sum_{g \in G_e} (y^\paren{0})_{e,k,t}^g$, $\hspace{1cm} \forall \ e \in \edges, k \in [2], t \in [T]$

\For{$i = 0, 1, \cdots, N_{\emph{iter}}-1$} 

    \State $v^\paren{i} \gets \Unif(\Sphere(\R^{|\edges|T+1}))$
    \label{Algline: CBCP, Query v}
    


    \State $(\hat \boldtau^\paren{i}, \hat{\boldB}^\paren{i}) \gets \proj_{S_{\boldtau, \boldB}} \left( (\boldtau^\paren{i}, \boldB^\paren{i}) + \frac{\delta}{(i+1)^{1/4} \sqrt{|\edges|T + 1}} v^\paren{i} \right)$
    \label{Algline: CBCP, Compute hat tau, hat B}

    \State $\hat y^\paren{i} \gets \Y^\eq\big(\G^C(\hat \boldtau^\paren{i}, \hat{\boldB}^\paren{i}) \big)$
    \label{Algline: CBCP, Compute hat y}

    \State $\hat x_{e,k,t}^\paren{i} \gets \sum_{g \in G_e} (\hat y^\paren{i})_{e,k,t}^g$, $\hspace{1cm} \forall \ e \in \edges, k \in [2], t \in [T]$
    \label{Algline: CBCP, Compute hat x}

    \State $\hat \nabla f_{C, \lambda}^\paren{i} 
    \gets \frac{(i+1)^{1/4} (|\edges|T + 1)^{3/2}}{\delta}  \big(f_{C, \lambda}\big(\hat y^\paren{i}, \hat x^\paren{i}, \hat \boldtau^\paren{i}, \hat{\boldB}^\paren{i}) - f_{C, \lambda}(y^\paren{i}, x^\paren{i}, \boldtau^\paren{i}, \boldB^\paren{i}) \big) v^\paren{i}$
    \label{Algline: CBCP, Zeroth-Order Gradient Estimate}


    \State $(\boldtau^\paren{i+1}, \boldB^\paren{i+1}) \gets \proj_{S_{\boldtau, \boldB}} \left( (\boldtau^\paren{i}, \boldB^\paren{i}) - \frac{\eta}{(i+1)^{1/2}(|\edges|T + 1)} \hat \nabla f_{C, \lambda}^\paren{i} 
    \right)$
    \label{Algline: CBCP, Gradient Update}

    \State $y^\paren{i+1} \gets \Y^\eq\big(\G^C(\boldtau^\paren{i+1}, \boldB^\paren{i+1}) \big)$
    \label{Algline: CBCP, Compute y, at iteration i+1}

    \State $x_{e,k,t}^\paren{i+1} \gets \sum_{g \in G_e} (y^\paren{i+1})_{e,k,t}^g$, $\hspace{1cm} \forall \ e \in \edges, k \in [2], t \in [T]$
    \label{Algline: CBCP, Compute x, at iteration i+1}
    
\EndFor

\State $i^\star \gets \arg\min_{i \in [N_\iter]} f_{C, \lambda}(\boldtau^\paren{i}, \boldB^\paren{i}, y^\paren{i}, x^\paren{i})$
\label{Algline: Extract optimal CBCP Policy Parameters}


\EndProcedure
\end{algorithmic}
\end{algorithm}

\begin{algorithm} 
\small
\caption{
Zeroth-Order Optimization for Computing Stationary DBCP Policies via \eqref{Eqn: f D lambda, Def}
}
\label{Alg: Opt DBCP, Zeroth-order Optimization} 

\begin{algorithmic}

\Procedure{ZO-DBCP}{$T$, $\network_\ELP = (V_\ELP, E_\ELP)$, $\eta$, $\delta$, $\boldalpha^\paren{0}$, $\boldB^\paren{0}$, $N_\iter$}


\State $S_{\tau \alpha} \gets [0, C_\tau] \times [0, \infty)$

\For{$(e, t) \in \edges \times [T]$}
    \State $\boldy_{e,t}^\paren{0} \gets \Y^\eq\big(\G^D(\tau_{e,t}^\paren{0}, \alpha_{e,t}^\paren{0}) \big)$
    
    \State $x_{e,k,t}^\paren{0} \gets \sum_{g \in G_e} (y^\paren{0})_{e,k,t}^g$, $\hspace{1cm} \forall \ k \in [2]$
    
    \For{$i = 0, 1, \cdots, N_{\emph{iter}}-1$} 
    
        \State $v_{e,t}^\paren{i} \gets \Unif(\Sphere(\R^2))$
        \label{Algline: DBCP, Query v}
    

        \State $(\hat \tau_{e,t}^\paren{i}, \hat \alpha_{e,t}^\paren{i}) \gets \proj_{S_{\tau\alpha}} \left( (\tau_{e,t}^\paren{i}, \alpha_{e,t}^\paren{i}) + \frac{\delta}{\sqrt{2} (i+1)^{1/4}} v_{e,t}^\paren{i} \right)$
        \label{Algline: DBCP, Compute hat tau, hat B}
    
        \State $\hat \boldy_{e,t}^\paren{i} \gets \Y^\eq\big(\G^D(\hat \tau_{e,t}^\paren{i}, \hat \alpha_{e,t}^\paren{i}) \big)$
        \label{Algline: DBCP, Compute hat y}
    
        \State $\hat x_{e,k,t}^\paren{i} \gets \sum_{g \in G_e} (\hat y^\paren{i})_{e,k,t}^g$, $\hspace{1cm} \forall \ k \in [2]$
        \label{Algline: DBCP, Compute hat x}
    
        \State $\hat \nabla f_{D, \lambda, e, t}^\paren{i} \gets \frac{(i+1)^{1/4} 2^{3/2}}{\delta} \big( f_{D, \lambda, e, t}(\hat y_{e,t}^\paren{i}, \hat x_{e,t}^\paren{i}, \hat \tau_{e,t}^\paren{i}, \hat \alpha_{e,t}^\paren{i}) - f_{D, \lambda, e, t}(y_{e,t}^\paren{i}, x_{e,t}^\paren{i}, \tau_{e,t}^\paren{i}, \alpha_{e,t}^\paren{i}) \big)$
        \label{Algline: DBCP, Zeroth-Order Gradient Estimate}
    

        \State $(\boldtau_{e,t}^\paren{i+1}, \boldalpha_{e,t}^\paren{i+1}) \gets \proj_{S_{\tau\alpha}} \left( (\tau_{e,t}^\paren{i}, \alpha_{e,t}^\paren{i}) - \frac{\eta}{2(i+1)^{1/2}} \hat \nabla f_{D, \lambda, e, t}^\paren{i} v^\paren{i} \right)$
        \label{Algline: DBCP, Gradient Update}
    
        \State $\boldy^\paren{i+1} \gets \Y^\eq\big(\G^D(\boldtau^\paren{i+1}, \boldalpha^\paren{i+1}) \big)$
        \label{Algline: DBCP, Compute y, at iteration i+1}
    
        \State $x_{e,k,t}^\paren{i+1} \gets \sum_{g \in G_e} (y^\paren{i+1})_{e,k,t}^g$, $\hspace{1cm} \forall \ k \in [2]$
        \label{Algline: DBCP, Compute x, at iteration i+1}

    \EndFor
    
\EndFor

\State $i^\star \gets \arg\min_{i \in [N_\iter]} f_{D, \lambda}(\boldtau^\paren{i}, \boldalpha^\paren{i}, y^\paren{i}, x^\paren{i})$
\label{Algline: Extract optimal DBCP Policy Parameters}


\EndProcedure

\end{algorithmic}
\end{algorithm}


\section{Data Table for Numerical Experiments}
\label{sec: App, Data Table for Numerical Experiments}

\begin{table}
\begingroup
    \centering
    \small
    \begin{tabular}{|c|c|c|c|c|c|c|c|c|} 
    \hline
    \multirow{3}{4em}{\centering o-d Node Indices} & \multicolumn{2}{c|}{o-d Cities} & \multicolumn{5}{c|}{Demand of group $g \in [5]$} & \multirow{3}{4em}{\centering Demand $\bar d_{ij}$} \\
    \cline{2-8}
    & \multirow{2}{*}{Origin} & \multirow{2}{*}{Destination} & $g=1$ & $g=2$ & $g=3$ & $g=4$ & $g=5$ &  \\
    & & & (El) & (El) & (Inel) & (Inel) & (Inel) &  \\
    \hline
    \hline
    $(1, 2)$ & Palo Alto & Palo Alto & 60.03 & 36.75 & 93.70 & 134.76 & 287.28 & 617.31 \\
    \hline 
    $(1, 3)$ & Palo Alto & East Palo Alto & 9.61 & 5.88 & 15.00 & 21.56 & 45.98 & 92.23 \\
    \hline
    $(1, 4)$ & Palo Alto & Redwood City & 111.52 & 68.28 & 174.11 & 250.36 & 533.73 & 1505.47 \\
    \hline
    $(1, 6)$ & Palo Alto & San Mateo & 95.07 & 58.20 & 148.42 & 213.41 & 454.96 & 891.13 \\
    \hline
    $(1, 7)$ & Palo Alto & Burlingame & 23.77 & 14.55 & 37.10 & 53.35 & 113.74 & 268.41 \\
    \hline
    $(1, 8)$ & Palo Alto & Millbrae & 150.04 & 91.86 & 234.25 & 336.83 & 718.07 & 1631.19 \\
    \hline
    $(2, 5)$ & East Palo Alto & Redwood City & 25.30 & 35.37 & 63.00 & 74.87 & 59.64 & 211.09 \\
    \hline
    $(2, 7)$ & East Palo Alto & San Mateo & 21.57 & 30.15 & 53.70 & 63.82 & 50.84 & 124.94 \\
    \hline
    $(2, 8)$ & East Palo Alto & Burlingame & 5.39 & 7.54 & 13.42 & 15.96 & 12.71 & 37.63 \\
    \hline
    $(2, 9)$ & East Palo Alto & Millbrae & 34.04 & 47.59 & 84.75 & 100.73 & 80.24 & 228.71 \\
    \hline
    $(3, 4)$ & Redwood City & Redwood City & 10.4 & 13.9 & 23.59 & 37.36 & 61.24 & 161.30 \\
    \hline
    $(3, 6)$ & Redwood City & San Mateo & 39.44 & 52.77 & 89.43 & 141.64 & 232.17 & 742.83 \\
    \hline
    $(3, 7)$ & Redwood City & Burlingame & 9.86 & 13.19 & 22.36 & 35.41 & 58.05 & 223.76 \\
    \hline
    $(3, 8)$ & Redwood City & Millbrae & 62.25 & 83.29 & 141.15 & 223.56 & 366.46 & 1359.80 \\
    \hline
    $(5, 6)$ & San Mateo & San Mateo & 32.91 & 34.88 & 65.44 & 108.94 & 150.09 & 285.83 \\
    \hline
    $(5, 7)$ & San Mateo & Burlingame & 22.55 & 23.89 & 44.84 & 74.64 & 102.83 & 274.06 \\
    \hline
    $(5, 8)$ & San Mateo & Millbrae & 142.38 & 150.85 & 238.06 & 471.20 & 649.17 & 1665.53 \\
    \hline
    $(6, 8)$ & Burlingame & Millbrae & 71.55 & 84.37 & 178.34 & 296.88 & 436.78 & 1044.21 \\
    \hline
    $(7, 8)$ & Millbrae & Millbrae & 113.36 & 101.04 & 192.22 & 343.78 & 480.55 & 1272.99 \\
    \hline
    \end{tabular}
    \caption{
    \small \sf
    Total demand,
    as well as demand
    for each of the 2 eligible (``el") and 3 ineligible (``inel") user groups for each o-d pair with nonzero flow in $\network_\ELP$. Group indices are assigned in order of increasing VoT.
    O-d pairs whose origin and destination cities coincide describe intra-city commutes.}
    \label{Table: Demand by group values}
\endgroup
\end{table}


\begin{table}
\begingroup
    \centering
    \small
    \begin{tabular}{|c|c||c|c|c|c|c|c|c|c|} 
    \hline
    Edge & City & $\bar \ell_e$ & $\beta_e$ & $\kappa_e$ & \multicolumn{5}{c|}{VoT (\$/min)} \\
    \cline{6-10}
    ($e$) & & (min) & (min/flow) & (flow) & $g = 1$ & $g = 2$ & $g = 3$ & $g = 4$ & $g = 5$ \\
     & & & & & (El) & (El) & (Inel) & (Inel) & (Inel) \\
    \hline
    \hline 
    1 & Palo Alto & 1.33 & $7.83 \times 10^{-4}$ & 1001.52 & 0.04 & 0.15 & 0.30 & 0.58 & 1.86 \\
    \hline
    2 & East Palo Alto & 2.21 & $5.84 \times 10^{-4}$ & 1001.52 & 0.05 & 0.15 & 0.29 & 0.57 & 1.27 \\
    \hline
    3 & Redwood City & 4.89 & $1.56 \times 10^{-3}$ & 881.18 & 0.04 & 0.15 & 0.29 & 0.57 & 1.68 \\
    \hline
    4 & Belmont & 1.20 & $1.99 \times 10^{-3}$ & 1278.95 & N/A & N/A & N/A & N/A & N/A \\
    \hline
    5 & San Mateo  & 5.54 & $2.15 \times 10^{-3}$ & 1034.09 & 0.05 & 0.15 & 0.30 & 0.59 & 1.69 \\
    \hline
    6 & Burlingame & 1.50 & $3.07 \times 10^{-4}$ & 845.15 & 0.04 & 0.15 & 0.30 & 0.59 & 1.76 \\
    \hline
    7 & Millbrae & 2.38 & $3.22 \times 10^{-4}$ & 853.18 & 0.06 & 0.16 & 0.30 & 0.60 & 1.56 \\
    \hline
    \end{tabular}
    \caption{
    \small \sf
    Travel latency parameters $\ell_e$, $\beta_e$, and $\kappa_e$, for each edge $e \in \edges_\ELP = [7]$, and VoTs for user groups whose origin city is $e$. Group indices are assigned in order of increasing VoT. (The PeMS traffic data used to compute these VoTs provided almost no data for users whose origin city is Belmont.)
    }
    \label{Table: Latency Function Parameters and VoT by group values}
\endgroup
\end{table}

\begin{table}
\begingroup
    \centering
    \footnotesize
    \begin{tabular}{|c|c|c|c|c|c|c|c|c|c|c|c|} 
    \hline
    Weight & Weights & Policy & \multicolumn{2}{c|}{Travel} & Toll & Societal & \multicolumn{3}{c|}{Express Lane} & \multicolumn{2}{c|}{Average Travel} \\
    Set & $(\lambda_E, \lambda_R, \lambda_I)$
    & & \multicolumn{2}{c|}{Costs ($10^4$ \$)} & Revenue & Cost & \multicolumn{3}{c|}{Use (\%)} & \multicolumn{2}{c|}{Time (min)} \\
    \cline{4-5}
    \cline{8-12}
    & & & \hspace{0.3mm} El. \hspace{0.3mm} & Inel. & ($10^4$ \$) & ($10^4$ \$) & All & El. & Inel. & Expr. & GP \\
    \hline
    \hline
    \multirow{16}{*}{\centering $S_{\boldlambda}^\paren{1}$} &
    \multirow{2}{*}{\centering $(1, 1, 1)$} & CBCP & $1.28$ & $75.80$ & $2.26$ & $78.90$ & 16.57 & 41.62 & 11.51 & 21.31 & 25.03 \\
    \cline{3-12}
    & & DBCP & $1.32$ & $78.00$ & $2.04$ & $77.30$ & 21.53 & 27.96 & 20.24 & 22.29 & 24.38 \\
    \cline{2-12}
    & \multirow{2}{*}{\centering $(1, 5, 1)$} & CBCP & $1.31$ & $80.40$ & $2.65$ & $68.50$ & 16.34 & 32.51 & 13.08 & 21.17 & 25.21 \\
    \cline{3-12}
    & & DBCP & $1.33$ & $79.00$ & $2.79$ & $66.40$ & 19.39 & 33.12 & 16.63 & 21.45 & 24.72 \\
    \cline{2-12}
    & \multirow{2}{*}{\centering $(1, 10, 1)$} & CBCP & $1.33$ & $79.70$ & $3.40$ & $47.10$ & 17.48 & 22.38 & 16.48 & 21.23 & 25.00 \\
    \cline{3-12}
    & & DBCP & $1.35$ & $79.60$ & $3.56$ & $45.40$ & 17.88 & 15.54 & 18.34 & 21.27 & 24.99 \\
    \cline{2-12}
    & \multirow{2}{*}{\centering $(5, 5, 1)$} & CBCP & $1.27$ & $80.10$ & $18.60$ & $77.20$ & 17.60 & 51.88 & 10.68 & 21.46 & 25.12 \\
    \cline{3-12}
    & & DBCP & $1.29$ & $78.80$ & $2.25$ & $74.00$ & 19.64 & 38.87 & 15.77 & 21.79 & 24.64 \\
    \cline{2-12}
    & \multirow{2}{*}{\centering $(5, 10, 1)$} & CBCP & $1.33$ & $80.80$ & $2.73$ & $60.10$ & 15.01 & 26.24 & 12.73 & 21.27 & 25.37 \\
    \cline{3-12}
    & & DBCP & $1.34$ & $79.70$ & $3.00$ & $56.40$ & 17.88 & 23.27 & 16.79 & 21.56 & 24.96 \\
    \cline{2-12}
    & \multirow{2}{*}{\centering $(10, 10, 1)$} & CBCP & $1.27$ & $80.90$ & $2.04$ & $73.10$ & 15.22 & 45.56 & 9.09 & 20.95 & 25.34  \\
    \cline{3-12}
    & & DBCP & $1.31$ & $80.00$ & $2.47$ & $68.40$ & 17.10 & 43.28 & 11.82 & 21.11 & 25.04 \\
    \cline{2-12}
    & \multirow{2}{*}{\centering $(1, 5, 0)$} & CBCP & $1.31$ & $81.2$ & $2.78$ & $-12.60$ & 14.07 & 32.01 & 10.44 & 20.86 & 25.49 \\
    \cline{3-12}
    & & DBCP & $1.35$ & $80.90$ & $3.21$ & $-14.70$ & 15.50 & 27.24 & 13.13 & 20.91 & 25.40 \\
    \cline{2-12}
    & \multirow{2}{*}{\centering $(5, 10, 0)$} & CBCP & $1.30$ & $81.20$ & $2.73$ & $-20.80$ & 15.03 & 33.57 & 11.30 & 20.77 & 25.46 \\
    \cline{3-12}
    & & DBCP & $1.35$ & $81.00$ & $3.21$ & $-25.30$ & 15.92 & 24.03 & 14.29 & 20.82 & 25.41 \\
    \hline
    \hline
    \multirow{18}{*}{$S_{\boldlambda}^\paren{2}$} &
    \multirow{2}{*}{\centering $(5, 1, 1)$} & CBCP & $1.27$ & $80.10$ & $1.56$ & $84.90$ & 15.67 & 58.00 & 7.13 & 21.42 & 25.12 \\
    \cline{3-12}
    & & DBCP & $1.29$ & $77.10$ & $0.78$ & $82.80$ & 23.46 & 40.68 & 19.98 & 23.14 & 24.07 \\
    \cline{2-12}
    & \multirow{2}{*}{\centering $(10, 1, 1)$} & CBCP & $1.26$ & $80.4$ & $1.39$ & $91.70$ & 15.86 & 61.21 & 6.73 & 21.40 & 25.23 \\
    \cline{3-12}
    & & DBCP & $1.29$ & $77.00$ & $0.74$ & $89.20$ & 23.45 & 40.22 & 20.09 & 23.22 & 24.06 \\
    \cline{2-12}
    & \multirow{2}{*}{\centering $(20, 1, 1)$} & CBCP & $1.26$ & $81.40$ & $1.62$ & $105.00$ & 13.89 & 49.01 & 6.80 & 20.82 & 25.53 \\
    \cline{3-12}
    & & DBCP & $1.28$ & $77.80$ & $1.25$ & $102.00$ & 22.33 & 39.98 & 18.77 & 22.53 & 24.28 \\
    \cline{2-12}
    & \multirow{2}{*}{\centering $(5, 1, 0)$} & CBCP & $1.27$ & $81.10$ & $1.73$ & $4.62$ & 14.46 & 47.84 & 7.71 & 21.11 & 25.43 \\
    \cline{3-12}
    & & DBCP & $1.30$ & $79.80$ & $2.09$ & $4.42$ & 16.78 & 41.56 & 11.77 & 21.38 & 24.97 \\
    \cline{2-12}
    & \multirow{2}{*}{\centering $(10, 1, 0)$} & CBCP & $1.26$ & $80.90$ & $2.23$ & $10.40$ & 14.73 & 49.33 & 7.76 & 20.71 & 25.35 \\
    \cline{3-12}
    & & DBCP & $1.31$ & $79.60$ & $2.76$ & $10.40$ & 18.46 & 42.65 & 13.57 & 21.00 & 24.87 \\
    \cline{2-12}
    & \multirow{2}{*}{\centering $(20, 1, 0)$} & CBCP & $1.25$ & $81.50$ & $1.61$ & $23.30$ & 13.74 & 53.88 & 5.65 & 20.77 & 25.52 \\
    \cline{3-12}
    & & DBCP & $1.28$ & $79.90$ & $1.89$ & $23.70$ & 18.73 & 57.30 & 10.95 & 21.22 & 24.95 \\
    \cline{2-12}
    & \multirow{2}{*}{\centering $(5, 0, 1)$} & CBCP & $1.27$ & $80.80$ & $2.08$ & $87.10$ & 14.69 & 46.18 & 8.34 & 20.93 & 25.34 \\
    \cline{3-12}
    & & DBCP & $1.30$ & $76.80$ & $0.44$ & $83.20$ & 23.74 & 43.10 & 19.83 & 23.44 & 23.98 \\
    \cline{2-12}
    & \multirow{2}{*}{\centering $(10, 0, 1)$} & CBCP & $1.24$ & $80.80$ & $1.51$ & $93.20$ & 13.95 & 56.46 & 5.37 & 20.95 & 25.30 \\
    \cline{3-12}
    & & DBCP & $1.30$ & $76.80$ & $0.47$ & $89.70$ & 23.77 & 39.55 & 20.59 & 23.44 & 23.97 \\
    \cline{2-12}
    & \multirow{2}{*}{\centering $(20, 0, 1)$} & CBCP & $1.25$ & $80.90$ & $1.74$ & $106.00$ & 14.94 & 52.52 & 7.35 & 20.96 & 25.36  \\
    \cline{3-12}
    & & DBCP & $1.29$ & $76.90$ & $0.43$ & $103.00$ & 23.24 & 50.69 & 17.71 & 23.32 & 24.01 \\
    \hline
    \end{tabular}
    \caption{
    \small 
    \sf
    Total travel costs, toll revenues, societal costs, express lane usage, and average travel times at equilibrium, under optimal CBCP and DBCP policies with $(\lambda_E, \lambda_R, \lambda_I)$ from $S_{\boldlambda}^\paren{1}$ (which satisfy $\lambda_R \geq \lambda_E$) and $S_{\boldlambda}^\paren{2}$ (which satisfy $\lambda_R < \lambda_E$).
    }
    \label{Table: (Avg) All equilibrium results, for all lambda}
\endgroup
\end{table}

\end{document}